\documentclass[11pt]{article}

\usepackage{0-preamble}

\begin{document}

\title{\textbf{Relating Quantum Tamper-Evident Encryption\\to Other Cryptographic Notions}}

\date{}
\newlength{\authornamewidth}
\setlength{\authornamewidth}{\textwidth}
\author{
	\makebox[\authornamewidth]{S\'ebastien Lord}\\
	University of Ottawa\\
	\texttt{sebastien.lord@uottawa.ca}
}

\maketitle

	\begin{abstract}
    A quantum tamper-evident encryption scheme is a non-interactive
    symmetric-key encryption scheme mapping classical messages to
    quantum ciphertexts such that an honest recipient of a ciphertext
    can detect with high probability any meaningful eavesdropping.
	This quantum cryptographic primitive was first introduced by
    Gottesman in 2003. Beyond formally defining this security notion,
    Gottesman's work had three main contributions: showing that any
    quantum authentication scheme is also a tamper-evident scheme,
    noting that a quantum key distribution scheme can be constructed
    from any tamper-evident scheme, and constructing a
    prepare-and-measure tamper-evident scheme using only Wiesner states
    inspired by Shor and Preskill's proof of security for the BB84
    quantum key distribution scheme.

	In this work, we further our understanding of tamper-evident
	encryption by formally relating it to other quantum cryptographic
	primitives in an information-theoretic setting. In particular, we
    show that tamper evidence implies encryption, answering a question
    left open by Gottesman, we show that it can be constructed from any
    encryption scheme with revocation and vice-versa, and we formalize an
	existing sketch of a construction of quantum money from any
	tamper-evident encryption scheme.
	These results also yield as a corollary that any scheme allowing the
	revocation of a message must be an encryption scheme.
    We also show separations between tamper evidence and other
    primitives, notably showing that tamper evidence does not imply
	authentication and does not imply uncloneable encryption.
	\end{abstract}

\newpage
\tableofcontents
\newpage

\section{Introduction}

Consider the following scenario: Alice wants to send a classical message
to Bob over a communication channel which is accessible to a malicious
and eavesdropping Eve. Alice and Bob wish for Eve to never learn any
information on this message. To help them, they have a shared secret key
which, at the time of transmission, is unknown to Eve.

Alice and Bob can achieve their goal by using any
information-theoretically secure classical symmetric-key encryption
scheme, such as the one-time pad. However, if they use a scheme which
produces a classical ciphertext, Eve can always make and keep a copy of
it. Moreover, it is impossible for Alice and Bob to detect this
eavesdropping. The secrecy of the message is then maintained if and only
if Eve never learns the secret key at a later time. This general
scenario is illustrated in \cref{fg:standard-scenario}.

\begin{figure}[H]
	\begin{center}
		\begin{tikzpicture}[yscale=0.7]
		\node        (a) at (1,0) {Alice};
		\draw[thick] ($(a) + (1,0.75)$) rectangle ($(a) - (1,0.75)$);
		\node        (b) at (9,0) {Bob};
		\draw[thick] ($(b) + (1,0.75)$) rectangle ($(b) - (1,0.75)$);
		\node        (e) at (5,1.5) {Eve};
		\draw[thick] ($(e) + (1,0.75)$) rectangle ($(e) - (1,0.75)$);

		\draw[thick,->] ($(a) + (1.05,0)$) to ($(b) - (1.05,0)$);
		\draw[thick,->] ($(e) - (0,1.5)$) to ($(e) - (0,0.8)$);

		\node[below] (c) at (5,0) {$c$};
		\node (m0) at ($(a) - (2,0)$) {$m$};
		\node (m1) at ($(b) + (2,0)$) {$m$};
		\node (k0) at ($(a) + (0,1.75)$) {$k$};
		\node (k1) at ($(b) + (0,1.75)$) {$k$};

		\draw[thick,->] (m0) to ($(a) - (1.05,0)$);
		\draw[thick,->] ($(b) + (1.05,0)$) to (m1);

		\draw[thick,->] (k0) to ($(a) + (0,0.8)$);
		\draw[thick,->] (k1) to ($(b) + (0,0.8)$);
	\end{tikzpicture}
	\end{center}
	\caption{\label{fg:standard-scenario}
		A classical symmetric-key encryption scenario. Alice
		and Bob share a classical key $k$ which Alice uses to encrypt
		a classical plaintext $m$ as a classical ciphertext $c$ from
		which Bob recovers $m$. An eavesdropping Eve obtains,
		undetected, a copy of $c$.}
\end{figure}
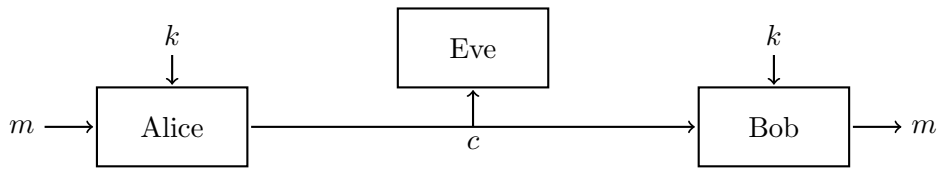

In the theoretical study of a cryptographic system, it is often asserted
and assumed that an eavesdropper never learns the secret key
$k$ and a proof of security in then derived in this context. In practice,
however, ensuring that this assumption holds over a long period of time
is a non-trivial task which falls under the umbrella of
\emph{key management}. We do not discuss this topic any further in
this work, but we direct the interested reader to guidelines provided by
the Unites States' \emph{National Institute of Standards and Technology}
(NIST) on this topic \cite{Ela20}.

However, due to the
\emph{information-disturbance principle} \cite{Fuc96arxiv},\footnote{%
	The information-disturbance principle is closely related to the
	\emph{no-cloning principle} \cite{Par70, Die82, WZ82, Ort18} which
	states that it is, in general, impossible to make a perfect copy of
	an unknown quantum state. Notably, if one could make such perfect
	copies, then information on a state could be obtained without
	disturbing it by making and measuring such a copy.}
also known as the  \emph{information-disturbance trade-off}, it may be
impossible for Eve to make, undetected, a copy of the
ciphertext and wait to learn
the secret key at a later point in time, if this ciphertext is a quantum state.
Indeed, this principle implies that extracting previously unknown
information from a quantum state must, in general, change it.

By leveraging this property of quantum mechanics, Alice and Bob can
detect attempts by Eve to eavesdrop on their communication.  More
precisely, we can devise schemes where Bob, while decrypting a
ciphertext, also determines whether he \emph{accepts} the transmission
--- indicating that he did not detect any eavesdropping --- or
\emph{rejects} it. This is illustrated in \cref{fg:te-scenario}. If Bob
accepts, then he is assured that Eve cannot learn anything about the
message \emph{even if she later learns the key}. Of course, such a
method to \emph{detect} eavesdropping does not \emph{prevent} it.
However, we can suppose that Alice and Bob may go above-and-beyond in
protecting their key or implement other mitigating measures if they
detect any eavesdropping.

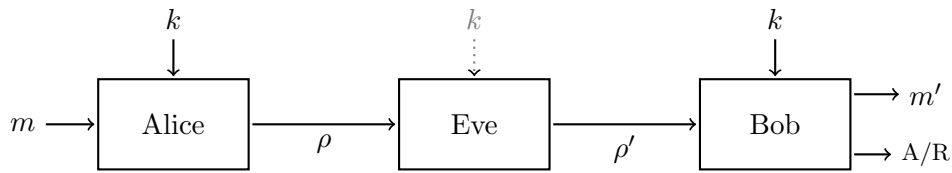
\begin{figure}[H]
	\begin{center}
		\begin{tikzpicture}[yscale=0.8]
		\node        (a) at (1,0) {Alice};
		\draw[thick] ($(a) + (1,0.75)$) rectangle ($(a) - (1,0.75)$);
		\node        (b) at (9,0) {Bob};
		\draw[thick] ($(b) + (1,0.75)$) rectangle ($(b) - (1,0.75)$);
		\node        (e) at (5,0) {Eve};
		\draw[thick] ($(e) + (1,0.75)$) rectangle ($(e) - (1,0.75)$);

		\draw[thick,->]
			($(a) + (1.05,0)$) -- ($(e) - (1.05,0)$)
			node[midway,below] () {$\rho$};
		\draw[thick,->]
			($(e) + (1.05,0)$) -- ($(b) - (1.05,0)$)
			node[midway,below] () {$\rho'$};

		\node (m0) at ($(a) - (2,0)$)    {$m$};
		\node (m1) at ($(b) + (2,0.5)$)  {$m'$};
		\node (k0) at ($(a) + (0,1.75)$) {$k$};
		\node (k1) at ($(b) + (0,1.75)$) {$k$};

		\node (ar) at ($(b) + (2,-0.5)$) {\footnotesize A/R};

		\draw[thick,->] (m0) to ($(a) - (1.05,0)$);
		\draw[thick,->] (k0) to ($(a) + (0,0.8)$);

		\draw[thick,->] (k1) to ($(b) + (0,0.8)$);
		\draw[thick,->] ($(b) + (1.05,0.5)$) to (m1);
		\draw[thick,->] ($(b) + (1.05,-0.5)$) to (ar);

		\node[text=gray] (k2) at ($(e) + (0,1.75)$) {$k$};
		\draw[thick,dotted,->,draw=gray] (k2) to ($(e) + (0,0.8)$);

	\end{tikzpicture}
	\end{center}
	\caption{\label{fg:te-scenario}
		The communication scenario considered for tamper-evident
		schemes. Alice and Bob share a classical key $k$ which Alice
		uses to encrypt a classical message $m$ as a \emph{quantum}
		ciphertext $\rho$. Eve attempts to eavesdrop on this ciphertext,
		possibly introducing some changes. Bob then decrypts, possibly
		obtaining different a message, and either accepts or rejects the
		transmission to indicate if he detected 
		Eve eavesdropping or not.
		Eve only learns the key $k$ \emph{after} her eavesdropping
		attempt.}
\end{figure}

In 2003, Gottesman formalized this security notion and gave two
constructions achieving it~\cite{Got03}. Following prior
work~\cite{BL20}, we refer to schemes achieving this notion as
\emph{tamper-evident schemes}.\footnote{%
	Gottesman originally called tamper-evident schemes ``uncloneable
	encryption'' schemes, but conceded that this may not be the proper
	term. Indeed, nothing \emph{prevents} Eve from
	cloning the ciphertext, or at least the information it encodes. She
	simply cannot do without leaving evidence of her tampering.}
Gottesman's first construction consisted in showing that any quantum
authentication scheme\footnote{%
	Quantum authentication schemes are essentially message
	authentication codes for quantum states.}
\cite{BCG+02} was already, without any modification, a tamper-evident
scheme. The second was a specific construction obtained by
adapting Bennet and Brassard's quantum key distribution scheme
\cite{BB84} and the Shor-Preskill proof of its security \cite{SP00}.

While tamper evidence and quantum key distribution are similar in the
sense that they both attempt to detect an eavesdropper, they do exhibit
important differences. The latter is typically interactive, does not
assume a pre-shared key, and produces a shared random string between
Alive and Bob, while the former is non-interactive, assumes a pre-shared
key, and allows Alice to send a \emph{specific} string to Bob.

The specifics of Gottesman's work on tamper evidence has attracted
little attention in the twenty years since its publication.
Tamper evidence is often cited as an example of a security notion
offered by quantum cryptography which cannot be achieved with only
classical tools, but little work has been done to formally situate
tamper evidence with respect to other quantum cryptographic primitives.
This leaves the following as an open avenue of research:
\begin{quote}
	\textit{How does tamper evidence fit in the landscape of
	cryptographic security notions? Which
	notions imply it? Which does it imply? Which are independent?}
\end{quote}
We offer in this work multiple answers to these questions.

\paragraph{Organization.}
We complete this section by reviewing prior work on tamper evidence
in \cref{te:sc:sota}, providing an overview of our contributions and
results in \cref{te:sc:contributions}, and highlighting a few open
questions in \cref{te:sc:questions}.

As for the remainder, we review in detail the necessary
mathematical background and notions from quantum information theory in
\cref{tee:sc:prelims}.
This includes an in-depth discussion of the trace distance, a
crucial notion in this work, as well as a lemma pertaining to gentle
measurements. We then state
the basic definitions surrounding tamper evidence in \cref{te:sc:aqecm}.
Finally, we conclude with the full details of our results in
\cref{sc:te=>enc,te:sc:te=>qm,te:sc:te<=>rev,te:sc:separations}.

\paragraph{Acknowledgements.}
The author thanks Anne Broadbent for many fruitful discussions over the
course of this work, in particular concerning the connections between
tamper evidence, revocation, and certified deletion. The author also
acknowledges support from the Natural Sciences and Engineering Research
Council of Canada (NSERC) during his time at the University of Ottawa.

%----------------------------------------------------------------------%
\subsection{The State of the Art}
\label{te:sc:sota}
%----------------------------------------------------------------------%

We review the works on tamper-evident schemes and quickly summarize
their results.

Gottesman's original work \cite{Got03} showed that every
quantum authentication scheme is also a tamper-evident scheme. It also
highlighted a connection with quantum key distribution (QKD) by
sketching a construction of a QKD scheme from any tamper-evident scheme
and by obtaining a tamper-evident scheme based on Bennet and Brassard's
QKD scheme and its analysis by Shor and Preskill. Finally, a few years
later, Gottesman sketched a construction of a quantum money scheme from
any tamper-evident scheme \cite{Got11se}.

Leermaker and \v{S}kori\'c established the existence of
a tamper-evident scheme which also achieved the notion of
\emph{key recycling} \cite{LS21}, essentially meaning that a
portion of a key can be reused with little security degradation
\cite{FS17}.

Finally, van der Vecht, Coiteux-Roy, and \v{S}kori\'c have also given a
tamper-evident scheme with short keys \cite{vdVCRS22}.

%----------------------------------------------------------------------%
\subsection{Our Contributions}
\label{te:sc:contributions}
%----------------------------------------------------------------------%

Before giving an overview of our contributions, we emphasize one place
where we do \emph{not} make a contribution: we do not define any new
tamper-evident schemes, nor any other type of cryptographic schemes,
``from scratch'' or from first principles. Instead, we define various
constructions mapping arbitrary tamper-evident schemes to other
cryptographic schemes or generically constructing tamper-evident schemes
from other cryptographic primitives.

Due to this level of abstraction
and the wide variety of notions under consideration, we place a
particular level of care in stating our various definitions and formally
proving our results. We emphasize that Gottesman has already
shown the unconditional existence of arbitrarily good
tamper-evident schemes~\cite{Got03}. This is recorded here in
\cref{te:th:te-exist}. Thus,
the constructions we define and study can be instantiated and the
conclusions we derive hold unconditionally.

Our contributions can be divided into four categories. We have three
``positive'' categories: we show that tamper evidence implies
encryption, we formalize how it it yields quantum money schemes, and we
relate it to encryption schemes with revocation. These are developed in
\cref{sc:te=>enc,te:sc:te=>qm,te:sc:te<=>rev}, respectively.
We conclude in \cref{te:sc:separations} with our final
category where we collect various ``negative'' results. There, we
provide counterexamples by constructing various schemes which are
tamper evident but fail to achieve other properties of interest.

In what follows, we review the three positive categories and we indicate
the relevant counterexamples from the fourth one as they
arise.

%----------------------------------------------------------------------%
\subsubsection{Tamper Evidence, Encryption, and Authentication}
\label{te:sc:contributions-enc}
%----------------------------------------------------------------------%

Gottesman defines a tamper-evident scheme as an \emph{encryption scheme}
which \emph{permits the honest receiver to detect meaningful
eavesdropping}. We refer to this second property
as the ``tamper-evident requirement''. Gottesman then states the
following \cite{Got03}:
\begin{quote}
	Note that [a tamper-evident scheme] by definition must also encrypt
	the message, so that Eve, even if she gets caught, cannot read the
	message until she learns the key. It is unclear whether this is an
	independent condition or whether it would follow from the
	[tamper-evident] requirement alone. (A classical message sent
	completely unencrypted can always be copied, but it may be possible
	to create partially encrypted messages which are [tamper-evident].)
\end{quote}

\paragraph{Contribution.} We answer this question in \cref{sc:te=>enc}
by showing that encryption is \emph{not} an independent property and
that \emph{it does} follow from the tamper-evident requirement alone.

There was already some evidence pointing towards this result. Gottesman
had shown that quantum authentication schemes were also, without any
modification, tamper-evident schemes \cite{Got03} and nearly concurrent
work by Barnum, Cr\'epau, Gottesman, Smith, and Tapp had shown that any
quantum authentication scheme is, without any modification, a quantum
encryption scheme \cite{BCG+02}. Thus, it was already known that
tamper evidence and encryption were somewhat related inasmuch that they
are both implied by a common cryptographic property:\footnote{%
	These implications hold with at most a polynomial loss of security,
	and no loss of correctness.
	}
\begin{equation}
\label{eq:tee:got-impl}
	\text{\small Quantum Authentication}
	\implies
	\left(
		\text{\small Tamper Evidence}
		\land
		\text{\small Encryption of Quantum States}
	\right).
\end{equation}
We show that a finer-grained version of
the previous implication holds, namely that
\begin{equation}
\label{eq:tee:our-impl}
	\text{\small Quantum Authentication}
	\implies
	\text{\small Tamper Evidence}
	\implies
	\text{\small Encryption of Classical Messages}
\end{equation}
and that the converses do not hold. Specifically, we show that the
second implication holds in \cref{th:te=>enc} and we show that the
converse of the first does not hold in \cref{te:sc:te-not-auth}. The
first implication was shown by Gottesman and the second converse
trivially does not hold. We note that Gottesman did \emph{conjecture}
the existence of non-authenticating tamper-evident schemes, but did not
prove this fact \cite{Got03}.

From \cref{eq:tee:got-impl} to \cref{eq:tee:our-impl}, we have shifted
from the encryption of \emph{quantum states} to the more restrictive ---
and hence less powerful --- encryption of \emph{classical messages},
seen as a subset of quantum states. This is a necessary specification as
we show in \cref{te:sc:te-not-qenc} how to construct a tamper-evident
scheme which does not encrypt all quantum states. Thus, a more
appropriate representation of the relations between encryption,
tamper-evidence, and authentication exhibited in this work is
given by the Venn diagram in \cref{fg:tee:venn}.

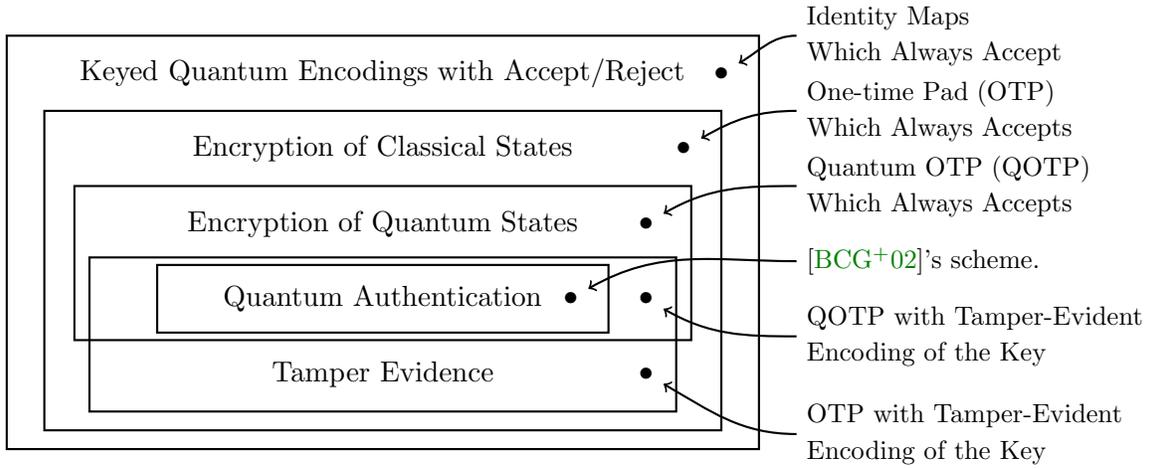
\begin{figure}
	\begin{center}
	\begin{tikzpicture}

		\node (qa) at (0,0) {Quantum Authentication};
		\node (qe) at ($(qa) + (0,1)$) {Encryption of Quantum States};
		\node (te) at ($(qa) - (0,1)$) {Tamper Evidence};
		\node (ce) at ($(qe) + (0,1)$) {Encryption of Classical States};
		\node (al) at ($(ce) + (0,1)$)
			{Keyed Quantum Encodings with Accept/Reject};

		%Quantum Authentication Box
		\draw[thick] ($(qa)+(3,0.45)$) rectangle ($(qa)-(3, 0.45)$);
		% Tamper Evidence Box
		\draw[thick] ($(te)+(3.9,-0.5)$) rectangle ($(qa)-(3.9,-0.55)$);
		% Quantum Encryption Box
		\draw[thick] ($(qe)+(4.1,0.5)$) rectangle ($(qa)-(4.1,0.55)$);
		% Classical Encryption Box
		\draw[thick] ($(ce)+(4.5,0.5)$) rectangle ($(te) - (4.5,0.75)$);
		\draw[thick] ($(al)+(5,0.5)$) rectangle ($(te)-(5,1)$);

		\node (qap) at ($(qa) + (2.5,0)$) {$\bullet$};
		\node (alp) at ($(al) + (4.5,0)$) {$\bullet$};
		\node (cep) at ($(ce) + (4,0)$)   {$\bullet$};
		\node (qep) at ($(qe) + (3.5,0)$) {$\bullet$};
		\node (tep) at ($(te) + (3.5,0)$) {$\bullet$};
		\node (nap) at ($(qa) + (3.5,0)$) {$\bullet$};

		\node[anchor=west, align=left] (alx) at ($(al) + (5.5,0.5)$)
			{\small Identity Maps\\\small Which Always Accept};
		\draw[thick, ->] (alx) to[out=-180,in=30] (alp);

		\node[anchor=west, align=left] (cex) at ($(ce) + (5.5,0.5)$)
			{\small One-time Pad (OTP)\\\small Which Always Accepts};
		\draw[thick, ->] (cex)  to[out=-180,in=30](cep);

		\node[anchor=west, align=left] (qex) at ($(qe) + (5.5,0.5)$)
			{\small Quantum OTP (QOTP)\\\small Which Always Accepts};
		\draw[thick, ->] (qex) to[out=-180,in=30] (qep);

		\node[anchor=west, align=left] (qax) at ($(qa) + (5.5,0.5)$)
			{\small \cite{BCG+02}'s scheme.};
		\draw[thick, ->] (qax) to[out=-180,in=30] (qap);

		\node[anchor=west, align=left] (nax) at ($(qa) + (5.5,-0.5)$)
			{\small QOTP with Tamper-Evident\\\small Encoding of the Key};
		\draw[thick, ->] (nax) to[out=-180,in=-30] (nap);

		\node[anchor=west, align=left] (tex) at ($(te) + (5.5,-0.8)$)
			{\small OTP with Tamper-Evident\\\small Encoding of the Key};
		\draw[thick, ->] (tex) to[out=-180,in=-30] (tep);

	\end{tikzpicture}
	\end{center}
	\caption{\label{fg:tee:venn}%
		A Venn diagram showing the relations between
		quantum authentication schemes, tamper-evident schemes, and
		encryption schemes within all keyed quantum
		encoding schemes where the decoding map must also produce an
		accept/reject flag. Some schemes are pinpointed to
		illustrate that each discernable region shown here is distinct.
		The identity, OTP, and QOTP schemes have decoding maps augmented
		to always accept.}
\end{figure}

\paragraph{Proving that tamper evidence implies encryption.}
Our proof method for showing that tamper evidence implies encryption is
similar to Barnum \emph{et al.}'s proof that quantum authentication
implies the encryption of quantum states \cite{BCG+02}. We proceed, as
they did, in three steps.

\emph{Step 1: Show that bad encryption implies bad tamper
evidence.}
The main novel observation underpinning our proof is that a bad
encryption scheme must be a bad tamper-evident scheme
(\cref{th:benc=>bte}). Indeed, consider an encoding scheme which admits
two messages whose encodings are near-perfectly distinguishable, even
without knowledge of the key. By definition, such a scheme cannot be a
good encryption scheme. Moreover, via a gentle measurement, the encoded
states can be distinguished essentially without disturbing them. This
allows an eavesdropper to obtain information on the ciphertext, and
hence on the underlying message, in a way which will not be detected by
the honest receiver with high probability. Such a scheme thus fails to
be sufficiently tamper-evident.

Unfortunately, we cannot simply take the contrapositive of this result
and be satisfied. The precise quantification we obtain is not enough to
conclude via this method that good tamper-evident schemes are good
encryption schemes. A bit more work is needed.

\emph{Step 2: Analyse the parallel composition of tamper-evident
schemes.} Next, we show that the parallel composition of tamper-evident
schemes only yields a \emph{linear} loss of security
(\cref{th:te-parallel}). This is in contrast to the fact made explicit
by Barnum \textit{et al.} that the parallel composition of encryption
schemes yields an \emph{exponential} loss in security \cite{BCG+02}.

\emph{Step 3: Leverage these different rates of security
degradation.}
We play this linear loss of security as a
tamper-evident scheme against the exponential loss of security as an
encryption scheme to obtain our result.

Indeed, if a good tamper-evident scheme was not a sufficiently good
encryption scheme, then the repeated parallel composition of this scheme
would yield a contradiction: It would pass the threshold of being a bad
encryption scheme before passing the threshold of being a bad
tamper-evident scheme, contradicting our prior result.

\paragraph{Proving that tamper evidence does not imply authentication.}
To show in \cref{te:sc:te-not-auth} that tamper evidence does not imply
authentication, we construct a tamper-evident scheme which makes it easy
for an adversary to change the encoded message without detection. The
idea is to encrypt a message $m$ with a classical one-time pad $p$, and
then encrypt $p$ with a tamper-evident scheme. The overall ciphertext is
then composed of two parts: the classical string $m \xor p$ and the
tamper-evident encryption of $p$.

Upon reception of a ciphertext from this scheme, the honest recipient
can check if any eavesdropping occurred on the tamper-evident encoding
of $p$. If none is detected, then they can be assured that an adversary
is unlikely to be able to learn $m$ even if they kept a copy of
$m \xor p$ and later learn the key used by the tamper-evident scheme.
Indeed, the tamper-evident guarantee ensures that they are unlikely to
learn $p$ in this scenario.

On the other hand, the only part of the ciphertext which depends on the
message --- the classical string $m \xor p$ --- is sent without any
additional layers of protection. This means that an adversary can
pick any string $e$ and modify this part of the ciphertext to become
$(m \xor e) \xor p$ without being detected by the honest recipient.
Indeed, the recipient will simply believe that $m \xor e$ was the
original message sent. In other words, this construction yields a
tamper-evident scheme which is \emph{malleable} because the one-time pad
is itself malleable.

\paragraph{Other separations.}
As for the remaining separating examples illustrated in
\cref{fg:tee:venn}, we claim that the always-accepting identity maps,
the always-accepting one-time pad, and the always-accepting quantum
one-time pad are trivial to situate in this diagram. As for the quantum
one-time pad with a tamper-evident encoded key, we claim without
further proof that it can be situated with an argument analogous to the
one we have just made for the classical one-time pad with a
tamper-evident encoded key.

%----------------------------------------------------------------------%
\subsubsection{Formalizing Quantum-Money from Tamper Evidence}
\label{sc:contributions-qm}
%----------------------------------------------------------------------%

Quantum money was perhaps the first widely known quantum cryptographic
primitive, having been published by Wiesner in 1983 \cite{Wie83} about a
decade after it was initially conceived.\footnote{%
	Those interested should read Brassard's account of the early days of
	quantum cryptography \cite{Bra05}.}
In a quantum money scheme, a bank issues ``quantum banknotes'' which, in
general, consist of a quantum state $\rho_s$ and a classical piece of
information $s$ often understood to be the serial number
of the banknote. Any user may return a banknote to the bank to have it
validated to ensure that it is not a forgery. The relevant security
guarantee here is that it should be infeasible for any adversary to turn
one honestly produced banknote $(\rho_s, s)$ into \emph{two} banknotes
$(\sigma', s')$ and $(\sigma'',s'')$ which would \emph{both} be accepted
by the bank. The basic template for the bank's verification process is
for it to maintain a database associating each public serial number $s$
to a secret key $k_s$ which the bank can use to validate the quantum
state $\rho_s$.

In this foundational work, Wiesner presented an explicit
quantum money scheme and sketched an argument that it should
achieve the desired security notion. A formal proof of
security was presented about thirty years later \cite{MVW13}, prompted
by discussions on an internet forum. During these same discussions,
Gottesman sketched a way to transform any tamper-evident scheme into a
quantum money scheme \cite{Got11se}.

Gottesman's idea was fairly straightforward: The bank samples a random
key $k$, message $m$, and serial number $s$. It keeps $(s, k, m)$ in its
database and encodes the message $m$ with the tamper-evident scheme
using the key $k$ to obtain a quantum state $\rho$. The bank then issues
$(\rho, s)$ as a banknote. To verify the authenticity of a supposed
banknote $(\rho', s')$, the bank recovers the corresponding key $k'$ and
message $m'$ from its database and decodes the state $\rho'$ with the
tamper-evident scheme using the key $k'$. The bank accepts the banknote
if and only if the message $m'$ is recovered \emph{and} if it does not
detect any eavesdropping. At a high-level, the security of this quantum
money scheme follows from the security of the tamper-evident scheme.
If the bank does not detect any eavesdropping, then a counterfeiter is
highly unlikely to have obtained any information on the message, making
it difficult for the counterfeiter to have produced a second valid
encoding of this same message.

\paragraph{Contribution.} We formalize Gottesman's construction
and formally prove and quantify the security of the resulting quantum
money scheme in \cref{te:sc:te=>qm}. Notably, we account for necessary
technical details which were not directly addressed in the initial sketch.

\paragraph{Related counterexamples.} Interestingly, we later show in
\cref{te:sc:separations} that it is
\emph{necessary} in this construction for the bank to check both that no
eavesdropping has occurred \emph{and} that the correct message is
recovered. We do so by constructing two explicit counterexamples.

First, we construct in
\cref{te:sc:double} a tamper-evident scheme which allows an adversary
given one honestly produced ciphertext to produce two states which will
both decode to the original plaintext, but neither of which will pass
the eavesdropping test. Our construction for this counterexample is a scheme which produces two tamper-evident
encodings of the same message $m$ with independently sampled keys
and that this overall ciphertext passes the eavesdropping check
if and only if \emph{both} encryptions of $m$ pass their individual
eavesdropping check. In particular, this implies that tamper-evident
schemes need not be uncloneable encryption schemes \cite{BL20}.

For our second class of counterexamples, we construct in
\cref{te:sc:ast} secure tamper-evident schemes allowing an adversary
given one honestly produced ciphertext to
create two states which will both pass the eavesdropping check, but
neither of which are likely to decode to the correct message. 
Our counterexample construction splits the plaintext with a simple group-based
$2$-out-of-$2$ classical secret sharing scheme (\eg: \cite{Sha79}) and
then encode both of these shares with a tamper-evident scheme using
independent keys. The overall ciphertext passes the eavesdropping
check if at least one, but not necessarily both, encrypted
shares pass their individual eavesdropping check. By splitting the
shares, an adversary can create two states passing verification,
but that cannot individually be used to recover the plaintext.

%----------------------------------------------------------------------%
\subsubsection{Revocation From Tamper Evidence, and Vice-Versa}
\label{te:sc:contributions-rev}
%----------------------------------------------------------------------%

An encryption scheme with revocation, or simply a \emph{revocation
scheme}, is a symmetric-key encryption scheme equipped with an
additional revocation procedure. This procedure allows the recipient of
a ciphertext, Bob, to return it to the sender, Alice. If Alice accepts
the return \emph{and} Bob did not yet know the key, then Alice is
assured that Bob has learned no information on the plaintext and that
this will remain the case \emph{even if he later learns the key}.
Evidently, it is impossible to achieve this if the ciphertext is
classical, but it does become possible by using quantum ciphertexts due
to the information-disturbance principle. In general, a revocation
procedure could consist in multiple rounds of classical or quantum
communication between Alice and Bob. In this work however, we will only
consider revocation procedures which consists of a single classical or
quantum message from Bob to Alice.

Revocation schemes were first studied by Unruh \cite{Unr15b}. There,
the revocation procedure consisted of Bob simply returning the quantum
ciphertext followed by Alice verifying that the ciphertext state
was unchanged. A few years later, Broadbent and Islam identified an
important subclass of revocation schemes, still with quantum ciphertexts, which they called
\emph{encryption with certified deletion} \cite{BI20}, or simply
a certified deletion scheme. In a nutshell, a certified deletion scheme
is a revocation scheme where the revocation procedure only requires a
single \emph{classical} message. In other words, Bob produces a
classical bit string, \textit{i.e.}~a \emph{certificate}, proving to
Alice that he has effectively deleted the ciphertext. Broadbent and
Islam also constructed a certified deletion scheme and provided a
complete proof of security.\footnote{%
	Contemporary and independent work by Coiteux-Roy and Wolf
	\cite{CW19} also examined a notion and protocol similar to the ones
    presented by Broadbent and Islam but gave no satisfactory proof of
    security.}

Certified deletion has since become an important, promising, and
fruitful area of research in quantum cryptography. Research on this
topic has focused on expanding the scope of primitives
for which a notion of certified deletion can be defined and achieved
--- \textit{e.g.:} the works of Hiroka, Morimae, Nishimaki, and Yamakawa
\cite{HMNY21} and of Bartusek and Khurana \cite{BK23} which consider
public-key, attribute-based, and fully-homomorphic encryption schemes
with certified deletion --- or on reducing the assumptions needed to
achieve this notion, \textit{e.g.}: Kundu and Tan's device-independent
protocol for certified deletion \cite{KT23}.

\paragraph{Contribution.}
We demonstrate in \cref{te:sc:te<=>rev} a close relationship between
tamper evidence and the notions of revocation and certified deletion.
We do so in two ways.

First, we provide in \cref{sc:revocation-security-definition},
specifically \cref{df:revocation}, an information-theoretic
definition of security for revocation schemes which is similar to
Gottesman's definition of security for tamper-evident schemes. We then
demonstrate that, up to a linear loss in the security quantification,
this new definition is equivalent to the existing game-based definition
of revocation found in the literature.

Second, we give in \cref{sc:te<=>rev} two generic constructions relating
tamper-evidence and revocation. We describe in
\cref{te:sc:te=>rev} a construction which transforms any tamper-evident
scheme into a revocation scheme and in \cref{te:sc:rev=>te} a
construction which transforms any revocation scheme --- including any
certified deletion scheme --- into a tamper-evident scheme. In both of
these constructions, we show that there is at most a polynomial loss of
security and correctness.

Conceptually, the relationship between tamper evidence and revocation
follows from the fact that we can, at a high-level, identify the roles of parties in a revocation scheme with roles of parties in a tamper-evident scheme. First, we note that we can identify
the role of ``revocation Bob'' with ``tamper-evident Eve'' as they both attempt to extract undetected information from the quantum ciphertext. To achieve
either security notion, the remaining honest parties must ensure that
these adversaries remain ignorant of the plaintext even if they later
learn the key. To complete the comparison, we see
``revocation Alice'' as playing the part of both ``tamper-evident
Alice'' while generating the ciphertext and of ``tamper-evident Bob''
when determining if she accepts the state she received. This is
shown in \cref{te:fg:revocation-scenario}.

\begin{figure}
\begin{center}
	\begin{tikzpicture}[yscale=0.75]

		\node        (te-a) at (0,0) {Tamper-Evident Alice};
		\draw[thick] ($(te-a)+(2,0.75)$) rectangle ($(te-a)-(2,0.75)$);

		\node[align=center] (te-b) at ($(te-a)-(0,2)$) {Tamper-Evident Bob};
		\draw[thick] ($(te-b)+(2,0.75)$) rectangle ($(te-b)-(2,0.75)$);

		\node (rev-a) at ($(te-a) + (0,1.25)$) {Revocation Alice};
		\draw[thick] ($(te-a)+(2.5,1.75)$) rectangle ($(te-b)-(2.5,1)$);

		\node[align=center] (adv) at ($0.5*(te-a) + 0.5*(te-b) + (7,0)$)
		{Revocation Bob\\\emph{or}\\Tamper-Evident Eve};
		\draw[thick] ($(adv)+(2,1.75)$) rectangle ($(adv)-(2,1.75)$);

        \draw[thick,->] ($(te-a) + (2.1,0)$) -- ($(te-a) + (4.9,0)$)
            node[midway,above] () {$\rho$};
        \draw[thick,<-] ($(te-b) + (2.1,0)$) -- ($(te-b) + (4.9,0)$)
            node[midway,below] () {$\rho'$};

        \node (m) at ($(te-a) - (3.5,0)$) {$m$};
        \draw[thick,->] (m) to ($(te-a) - (2.1,0)$);

        \node (mp) at ($(te-b) - (3.5,-0.5)$) {$m'$};
        \draw[thick,->] ($(te-b) - (2.1,-0.5)$) to (mp);
        \node (ar) at ($(te-b) - (3.5,0.5)$) {$A/R$};
        \draw[thick,->] ($(te-b) - (2.1,0.5)$) to (ar);

        \node (k) at ($(te-a) + (0,2.5)$) {$k$};
        \draw[thick,->] (k) to ($(te-a) + (0,1.85)$);

        \node[text=gray] (kp) at ($(k) + (7,0)$) {$k$};
        \draw[thick,dotted,->,draw=gray] (kp) to ($(adv) + (0,1.85)$);

	\end{tikzpicture}
\end{center}
\caption{\label{te:fg:revocation-scenario}%
    Conceptual identification between the parties in a tamper-evident
    scheme and revocation scheme. Contrast with \cref{fg:te-scenario}.
    Note that in the revocation scenario Alice does not output a
    decoded message $m'$ and $\rho'$ may be on a different space than
    $\rho$.}
\end{figure}
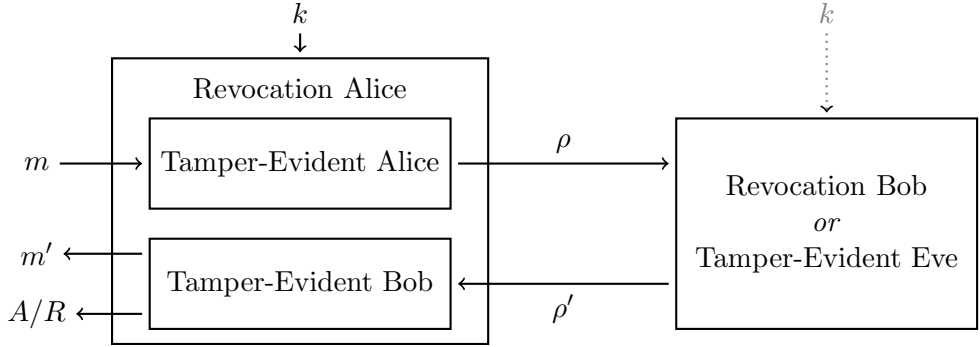

Our construction of a revocation scheme from a tamper-evident scheme is
also illustrated by \cref{te:fg:revocation-scenario}. We show that
Bob returning the ciphertext unmodified and Alice checking the flag
produced by the tamper-evident decryption procedure yields a secure
revocation procedure.

The construction of a tamper-evident scheme from a revocation scheme is
a bit more involved.
The main observation we use here is that being able to prove that you
have deleted the only instance of some information also implicitly
proves that you were the only party in possession of this
information.\footnote{%
	We thank Anne Broadbent for first making and sharing this
	observation.}
Hence, Bob in a tamper-evident scenario can convince himself that there
was no meaningful eavesdropping by successfully executing and verifying
a revocation procedure on the ciphertext by himself. At first glance,
this is not quite sufficient as Bob, in a tamper-evident scheme, must
also recover the message. However, if Bob executes a revocation
procedure, perhaps even one from a certified deletion scheme, how can he
hope to also recover the message? Has he not deleted the ciphertext? To
overcome this issue, we have Bob instead execute a gentle-measurement
version of the complete revocation procedure and subsequent verification
where only the final accept/reject output, and not the intermediate
state that would be returned, is measured.

Finally, we obtain as a direct corollary that every revocation
scheme must also be an encryption scheme. This follows from the fact
that tamper evidence implies encryption, that we can construct a
tamper-evident scheme from any revocation scheme, and that this
construction leaves both the encryption and key procedures unmodified.
This is recorded in \cref{sc:revocation=>encryption}.

%----------------------------------------------------------------------%
\subsection{Open Questions}
\label{te:sc:questions}
%----------------------------------------------------------------------%

While our work advances the understanding of the relations between
tamper evidence and other cryptographic security notions, it leaves much
to be done. We highlight here a few questions which remain largely open.

\paragraph{Do our results on tamper evidence hold in a computational
setting?}
The whole of this work is completed in the information-theoretic setting
where we do not formally consider the computational costs to implement
or attack cryptographic schemes. However, tamper evidence can also be
studied in a setting where such computational costs are considered. A
natural question then is to ask which of our results hold in such a
setting where schemes and attacks must be efficiently implementable. We
expect most of them to hold, and we believe most of our constructions
remain efficient when applied to efficient schemes.

One notable exception is our formalization of Gottesman's construction
of a quantum money scheme in \cref{sc:qm-got} when the underlying
tamper-evident scheme is not perfectly correct. There, we require the
bank to sample key-message pairs which will  lead to a correct
decryption with sufficiently high probability. It is unclear if such a
sampling can always be done efficiently.
However, as we discuss at the end of \cref{sc:qm-got}, this problem can
be side-stepped by weakening the correctness notion for quantum money
schemes from one where we require \emph{all} banknotes produced by a
bank to be accepted with high probability, to one where we only require
\emph{the large majority} of banknotes to be accepted with high
probability.

\paragraph{Can we formalize the connection between tamper evidence and
quantum key distribution?}
Gottesman noted that a quantum key distribution scheme can be easily
obtained from a tamper-evident scheme \cite{Got03}. Indeed, have Alice
encode a random string $s$ with a random key $k$ using a tamper-evident
scheme and send the resulting quantum state to Bob. Bob then
acknowledges reception of the state over an authenticated classical
channel, after which Alice discloses $k$ over this same channel. Bob can
now decode and verify the quantum state with $k$. If Bob accepts, they
can be sure that an eavesdropping Eve has essentially no information on
the string $s$, even if Eve also learned $k$ by listening to the
authenticated channel. Moreover, Gottesman sketched how minimal changes
to Bennet and Brassard's QKD scheme \cite{BB84} yields a tamper-evident
scheme. At first glance, this construction is specific to the BB84
scheme and may not be applicable to all QKD schemes.

This state of affairs yields two open questions which we have not
broached in this work:
\begin{enumerate}
	\item
		Can one formally prove and quantify the security of a QKD scheme
		obtained from Gottesman's generic construction based on any
		tamper-evident scheme?
	\item
		Can one formulate a minimal set of conditions on QKD schemes
		such that satisfying these conditions ensures that it is
		possible to generically transform it into a tamper-evident
		scheme? And, if so, what is the level of security achieved by
		the resulting tamper-evident scheme?
\end{enumerate}
We expect answering the first question to be a straightforward exercise.
However, the second appears more complex to tackle.

\paragraph{Is there a device-independent tamper-evident scheme?}
A device-independent quantum cryptographic scheme is one in which the
honest users need not have complete control of their quantum devices to
still achieve their desired security notions. Such schemes exist
for both QKD (\emph{e.g.}: \cite{VV14b}) and certified deletion
(\emph{e.g.} \cite{KT23}). As both of these security notions are related
to tamper evidence, it would appear likely that the methods used to
attain them in a device-independent manner should also allow us to
construct a device-independent tamper-evident scheme. Is this indeed the
case? In particular, it is unclear if our construction of a
tamper-evident scheme from any certified deletion scheme would preserve
device-independence, and we are inclined to believe it would not.

%======================================================================%
\section{Preliminaries}
\label{tee:sc:prelims}
%======================================================================%

We review in \cref{tee:sc:notation} the basic mathematical concepts used
in this work, how they relate to quantum mechanics, and the notation we
use. Next, in \cref{tee:sc:td}, we examine in some detail the notion of
the trace distance between two linear operators, a concept central to
this work. In particular, we highlight how some well known properties of
the trace distance between density operators generalize to
larger classes of linear operators, something which will be necessary
for us.
Finally, in \cref{tee:sc:cgm}, we examine the concept of a gentle
measurement.

%~~~~~~~~~~~~~~~~~~~~~~~~~~~~~~~~~~~~~~~~~~~~~~~~~~~~~~~~~~~~~~~~~~~~~~%
\subsection{Notation and Basic Concepts}                               %
\label{tee:sc:notation}                                                %
%~~~~~~~~~~~~~~~~~~~~~~~~~~~~~~~~~~~~~~~~~~~~~~~~~~~~~~~~~~~~~~~~~~~~~~%

Except for some differences in notation, we generally follow Watrous
\cite{Wat18}. We assume the reader is familiar with the notion of
tensor products on Hilbert spaces.

\paragraph{Sets, alphabets, and Hilbert spaces.}
Sets will generally be denoted with ``calligraphic'' uppercase characters
such as $\mc{X}$ or $\mc{M}$. An \emph{alphabet} is a non-empty finite
set. All Hilbert spaces in this work are complex and finite dimensional
and hence isomorphic to $\C^n$ for a suitable value of $n \in \N$. We
often use the Dirac, or ``bra-ket'' notation \cite{Dir39}, to denote the
elements of a vector space and their duals, \emph{i.e.:} the dual of a
vector $\ket{\psi}$ is $\ket{\psi}^\dag = \bra{\psi}$. For any alphabet
$\mc{A}$, we let $\C^{\mc{A}}$ be a Hilbert space which admits the set
$\{\ket{a} \,:\, a \in \mc{A}\}$ as an orthonormal
basis and we call this distinguished basis the \emph{computational
basis} of the space.\footnote{%
	More formally, we take $\C^\mc{A}$ to be the set of all maps
	$f : \mc{A} \to \C$ equipped with the usual point-wise addition,
	i.e.: $(f+g)(a) = f(a) + g(a)$, the usual scalar multiplication,
	i.e. $(c\cdot f)(a) = c \cdot f(a)$, and the inner product
	$(f,g) \mapsto \sum_{a \in \mc{A}} f(a)^\dag \cdot g(a)$ where
	$c^\dag$ denotes the complex conjugate of $c$.
	For all $a \in \mc{A}$, we identify the element $\ket{a}$ with the
	map taking $a$ to $1$ and all other elements of $\mc{A}$ to $0$. See
	\cite[Sec.~1.1.1]{Wat18} for more details on this construction.
	The details of this construction are never explicity used in this
	work.}
We will typically denote Hilbert spaces using a sans-serif font, such as
$\tsf{H}$, and moreover aim to systematically match symbols between
alphabets and Hilbert spaces constructed from them, for example
$\tsf{X} = \C^\mc{X}$.

Unless otherwise specified, we take the norm on a Hilbert space to be
the square root of the inner-product of an element with itself:
$\norm{\ket{\psi}} = \sqrt{\braket{\psi}}$.

We identify $\C$ with $\C^\mc{S}$ for any set $\mc{S}$ composed of a
single element. A common example is $\C = \C^{\{\varepsilon\}}$ where
$\varepsilon \in \{\bz,\bo\}^{0}$ is the unique bit string of length
$0$. Note that we will distinguish between the numbers $0,1 \in \N$ and
the symbols $\bz$ and $\bo$ used to denote the states of a classical
bit.\footnote{%
    In some edge cases not explicitly used in this work, this
    distinction can avoid some notational collisions.}

\paragraph{Operators on Hilbert spaces, their adjoints, and the trace.}
Let $\tsf{H}$ and $\tsf{H}'$ be Hilbert spaces. We denote the set of
linear operators whose domain and codomain are $\tsf{H}$ and $\tsf{H}'$,
respectively, by $\mc{L}(\tsf{H}, \tsf{H}')$ and we let $\mc{U}(\tsf{H},
\tsf{H}') \subseteq \mc{L}(\tsf{H}, \tsf{H}')$ denote the subset of
these operators which are \emph{isometries}, i.e.: those which preserve
the norm. If $\tsf{H}'=\tsf{H}$, we abbreviate these notations to
$\mc{L}(\tsf{H})$ and $\mc{U}(\tsf{H})$, respectively.
Thus, $\mc{U}(\tsf{H})$ is the set of \emph{unitary} operators on
$\tsf{H}$.

The identity map on a Hilbert space $\tsf{H}$ is denoted $I_\tsf{H}$ and
is unitary. For any linear operator $L \in \mc{L}(\tsf{H},
\tsf{H}')$, we define its \emph{adjoint}, denoted $L^\dag$, as the
unique element of $\mc{L}(\tsf{H}',\tsf{H})$ satisfying
$\left(L^\dag \ket{a}\right)^\dag\ket{b} = \bra{a}L\ket{b}$ for all
$\ket{a} \in \tsf{H}'$ and $\ket{b} \in \tsf{H}$.
For any Hilbert space $\tsf{H}$, the \emph{trace}
$\Tr_\tsf{H} : \mc{L}(\tsf{H}) \to \C$ is the unique linear map such
that $\Tr_\tsf{H}(\ketbra{\psi}{\phi}) = \braket{\phi}{\psi}$ for all
elements $\ket{\psi},\ket{\phi} \in \tsf{H}$. Recall that the trace is
\emph{cyclical}: $\Tr_\tsf{C}(ABC)=\Tr_\tsf{A}(BCA)=\Tr_\tsf{B}(CAB)$
for any linear operators
$A \in \mc{L}(\tsf{A}, \tsf{C})$, $B \in \mc{L}(\tsf{B}, \tsf{A})$, and
$C \in \mc{L}(\tsf{C}, \tsf{B})$ and that it also satisfies
$\Tr_\tsf{H}(L) = \sum_{i \in \mc{J}} \bra{b_j} L \ket{b_j}$ for any
$L \in \mc{L}(\tsf{H})$ and any orthonormal basis $\{\ket{b_j}\}_{j \in
\mc{J}}$ for $\tsf{H}$.

A linear operator $X \in \mc{L}(\tsf{H})$ is said to be \emph{positive
semidefinite} if there exists a linear operator $Y \in \mc{L}(\tsf{H})$
such that $X = Y^\dag Y$. This implies that $\bra{\psi} X \ket{\psi}$ is
a real non-negative number for any element $\ket{\psi} \in \tsf{H}$. We
let $\Pos(\tsf{H})$ denote the set of all positive operators
on~$\tsf{H}$. We let $\mc{D}(\tsf{H})$ denote the set of \emph{density}
operators, which is to say positive semidefinite operators with a trace
of one. We also define $\mc{D}_\bullet(\tsf{H})$ be the set of
\emph{subnormalized density operators}, which is to say
positive semidefinite operators whose trace is an element of the real
interval $[0,1]$.
Finally, an operator $X \in \mc{L}(\tsf{H})$ is \emph{Hermitian} if
it is self-adjoint, which is to say that $X^\dag = X$. We denote the set
of Hermitian operators on $\tsf{H}$ by $\text{Herm}(\tsf{H})$. Every
Hermitian operator $X$ admits a \emph{square root} $\sqrt{X}$ which is
the unique positive semidefinite operator satisfying
$\sqrt{X}\sqrt{X} = X$.

Thus, for any Hilbert space $\tsf{H}$ we have the inclusions
\begin{equation}
	\mc{D}(\tsf{H})
	\subseteq
	\mc{D}_\bullet(\tsf{H})
	\subseteq
	\Pos(\tsf{H})
	\subseteq
	\text{Herm}(\tsf{H})
	\subseteq
	\mc{L}(\tsf{H})
	\qq{and}
	\mc{U}(\tsf{H})
	\subseteq
	\mc{L}(\tsf{H}).
\end{equation}

\paragraph{Channels.}
We will be interested in a class of linear maps taking linear
operators to linear operators, namely \emph{channels}.
These are completely positive and trace-preserving
linear maps $\Phi : \mc{L}(\tsf{H}) \to \mc{L}(\tsf{H}')$ for
Hilbert spaces $\tsf{H}$ and $\tsf{H}'$. By \emph{completely positive},
we mean that for any Hilbert space $\tsf{Z}$ the map $\Phi \tensor
\Id_\tsf{Z}$ --- where $\Id_\tsf{Z}$ is the identity on
$\mc{L}(\tsf{Z})$ --- maps positive operators to positive operators,
\emph{i.e.}:
$\left(\Phi \tensor \Id_\tsf{Z}\right)(\Pos(\tsf{H} \tensor \tsf{Z}))
\subseteq  \Pos(\tsf{H'} \tensor \tsf{Z})$.
By \emph{trace preserving}, we mean that
$\Tr_{\tsf{H}'}\left(\Phi(L)\right) = \Tr_\tsf{H}(L)$ for all $L \in
\mc{L}(\tsf{H})$, \emph{i.e.}: $\Tr_{\tsf{H}'} \circ \Phi =
\Tr_{\tsf{H}}$.
The set of all channels with domain $\mc{L}(\tsf{H})$
and codomain $\mc{L}(\tsf{H}')$ is denoted $\text{Ch}(\tsf{H},
\tsf{H'})$. Note that channels map density operators to density
operators. Channels are also called \emph{CPTP maps}.

The trace $\Tr_{\tsf{H}} : \mc{L}(\tsf{H}) \to \C$ and
conjugation $L \mapsto V L V^\dag$ by any isometry
$V \in \mc{L}(\tsf{H}, \tsf{H}')$ are channels, as well as any
composition or tensor product of two channels. These observations
are sufficient to characterize and construct \emph{all} channels
\cite[Cor.~2.27]{Wat18}.

\paragraph{Measurements and swap channels.}
Let $\mc{A}$ be an alphabet. A \emph{measurement} on a Hilbert space
$\tsf{H}$ with outcomes in $\mc{A}$ is a
map~$\mu : \mc{A} \to \Pos(\tsf{H})$ such that
$\sum_{a \in \mc{A}} \mu(a) = I_\tsf{H}$. This yields
a channel~$\Ms(\mu) : \mc{L}(\tsf{H}) \to \mc{L}(\C^\mc{A})$
defined by $
	L
	\mapsto
	\sum_{a \in \mc{A}} \Tr\left(\mu(A) L\right)\cdot\ketbra{a}
	$.\footnote{%
		From our prior characterization, the fact that this is a channel
		follows by composing conjugation by the isometry $
			\sum_{a \in \mc{A}}
			\sqrt{\mu(a)} \tensor \ket{a}
			\in \mc{L}(\tsf{H}, \tsf{H} \tensor \C^\mc{A})
		$ with the channel $
			(\Tr_\tsf{H} \tensor \Id_{\C^{\mc{A}}}):
			\mc{L}(\tsf{H} \tensor \C^{\mc{A}})
			\to
			\mc{L}(\C^{\mc{A}})$.
	}

For any alphabet $\mc{X}$, we define the channel
$\Delta_\mc{X} : \mc{L}(\C^\mc{X}) \to \mc{L}(\C^\mc{X})$ by
\begin{equation}
	\rho
	\mapsto
	\sum_{x \in \mc{X}} \ketbra{x} \rho \ketbra{x}
	=
	\sum_{x \in \mc{X}} \bra{x}\rho\ket{x} \cdot \ketbra{x}
	=
	\sum_{x \in \mc{X}} \Tr(\ketbra{x} \rho) \cdot \ketbra{x}
\end{equation}
often called the \emph{completely dephasing channel}. It coincides
with a measurement in the computational basis,
\emph{i.e.}: $\Delta_\mc{X} = \Ms(\mu_\text{c.b.})$ where
$\mu_\text{c.b.} : \mc{X} \to \Pos(\tsf{X})$ is defined by $x \mapsto
\ketbra{x}$.

Finally, we define ``swap'' channels which reorder the
Hilbert spaces of their domains. Let
$\tsf{H}_1, \ldots \tsf{H}_n$ be a sequence of $n$ Hilbert spaces and
let $\pi : \{1, \ldots n\} \to \{1, \ldots n\}$ be a permutation. Then,
$\Swap_{\tsf{H}_1, \ldots, \tsf{H}_n}^\pi : \mc{L}(\tsf{H}_1 \tensor
\ldots \tensor \tsf{H}_n) \to \mc{L}(\tsf{H}_{\pi^{-1}(1)} \tensor \ldots
\tensor \tsf{H}_{\pi^{-1}(n)})$ is the channel which maps any simple tensor
$\ketbra{h_1} \tensor \ldots \tensor \ketbra{h_n}$ to
$\ketbra{h_{\pi^{-1}(1)}} \tensor \ldots \ketbra{h_{\pi^{-1}(n)}}$ and extends
linearly to all other elements. In practice, we will denote the
permutation $\pi$ with the ``one-line'' notation. For example,
$\Swap_{\tsf{H}, \tsf{H}', \tsf{H}''}^{(2,3,1)}$ maps
$\mc{L}(\tsf{H} \tensor \tsf{H}' \tensor \tsf{H}'')$
to $\mc{L}(\tsf{H}'' \tensor \tsf{H} \tensor \tsf{H}')$.

\paragraph{Quantum mechanics.}
How is this prior mathematical machinery relevant to quantum mechanics?
Quantum mechanics asserts that every quantum system is associated to a
Hilbert space of non-zero dimension called its \emph{state space}. At
any given moment, the state of a quantum system whose state space is
$\tsf{H}$ is described by a density operator $\rho \in \mc{D}(\tsf{H})$
on this space. Any action on a quantum system is described by a channel
$\Phi$ and the state of the system after the action is given by
$\Phi(\rho)$ if the state prior to the action was $\rho$. Finally,
the joint state space of two quantum systems when the first system has
state space $\tsf{H}$ and the second $\tsf{H}'$ is the tensor product
$\tsf{H} \tensor \tsf{H}'$.

A measurement $\mu : \mc{A} \to \Pos(\tsf{H})$, and the channel
$\Ms(\mu)$ it yields, models an observation of a quantum system which
provides classical information. More specifically, this measurement
is applied to a state $\rho \in \mc{D}(\C^{\mc{A}})$, the outcome $a \in
\mc{A}$ is obtained with probability $\Tr(\mu(a)\rho)$, destroying the
state $\rho$ and leaving a new system in the state $\ketbra{a}$ to
record this outcome. In particular, the cyclicity of the
trace implies that measuring a state $\rho \in \mc{D}(\C^\mc{A})$ in the
computational basis yields outcome $a \in \mc{A}$ with probability
$\Tr(\ketbra{a} \rho) = \bra{a} \rho \ket{a}$.

For more details, the reader is directed to
the textbooks of Nielsen and Chuang \cite{NC10} and Watrous
\cite{Wat18} and the references therein.

\paragraph{Markov's inequality.}
We recall Markov's inequality (\emph{e.g.}: \cite{Bil95}). This
is a concentration inequality which gives an upper bound on the
probability that a non-negative random variable exceeds a given
threshold in terms of the expectation of this random variable.
\begin{theorem}
\label{th:markov}
	Let $X$ be a random variable distributed on the non-negative reals
	$\R^+_0$. Then, for any strictly positive $\alpha \in \R^+$, we have
	that
	\begin{equation}
		\Pr\left[X \geq \alpha\right]
		\leq
		\frac{\E(X)}{\alpha}
		\qq{and}
		\Pr\left[X < \alpha\right]
		\geq
		1 - \frac{\E\left(X\right)}{\alpha}.
	\end{equation}
\end{theorem}
The following can be seen as an analogue to Markov's inequality if we
can also assume that $X$ almost surely falls in an interval $[0,\beta]$.
\begin{lemma}
\label{th:concentration}
	Let $X$ be a random variable distributed on $\R^+_0$ and let $\beta
	\in \R^+_0$ be such that $\Pr[0 \leq X \leq \beta] = 1$. Then, for
	any strictly positive $\alpha \in \R^+$ satisfying $\alpha < \beta$
	we have that
	\begin{equation}
		\Pr\left[X > \alpha\right]
		\geq
		\frac{\E(X) - \alpha}{\beta - \alpha}
		\qq{and}
		\Pr\left[X  \leq \alpha\right]
		\leq
		1
		-
		\frac{\E(X) - \alpha}{\beta - \alpha}
		.
	\end{equation}
\end{lemma}
\begin{proof}
	Since $\Pr\left[X > \beta\right] = 0$, we have that
	\begin{equation}
	\begin{aligned}
		\E(X)
		&\leq
		\beta \cdot \Pr\left[X > \alpha\right]
		+
		\alpha \cdot \Pr\left[X \leq \alpha\right]
		\\&=
		\beta\cdot\Pr\left[X > \alpha\right]
		+
		\alpha\cdot\left(1 - \Pr\left[X > \alpha\right]\right)
		\\&=
		\alpha+\Pr\left[X>\alpha\right]\cdot\left(\beta-\alpha\right)
	\end{aligned}
	\end{equation}
	which yields the first equation. Noting
	that $\Pr\left[X\leq\alpha\right] = 1-\Pr\left[X<\alpha\right]$ then
	yields the second.
\end{proof}

\paragraph{Final miscellaneous conventions.}
We freely identify $\C$ with $\mc{L}(\C)$ in this work, as well as
$\tsf{H}$, $\C \tensor \tsf{H}$, and $\tsf{H} \tensor \C$ for any
Hilbert space $\tsf{H}$.

To lighten notation when there is no risk of
confusion, we may identify an element of an alphabet $a \in \mc{A}$ with
the corresponding density operator $\ketbra{a} \in \mc{D}(\C^\mc{A})$.
For example, if $\Phi : \mc{L}(\C^\mc{A}) \to \tsf{H}$ is a channel, we
may write $\Phi(a-a')$ instead of $\Phi(\ketbra{a} - \ketbra{a'})$.

We emphasize that an empty quantum system has $\C$ as its state
space and that there is a unique density operator on this space, namely
$1$. By identifying $\C$ with $\C^{\{\varepsilon\}}$ --- as previously
discussed --- and using the previous convention, we can also denote the
state of an empty quantum system as $\varepsilon$.

When convenient, we may treat a density operator $\rho \in
\mc{D}(\C^\mc{A})$ as a random variable on the set $\mc{A}$ where
the probability of sampling $a \in \mc{A}$ from $\rho$ is given by
$\Tr(\ketbra{a} \rho)$. In other words, $a \gets \rho$ means that $a$
is sampled by measuring $\rho$ in the computational basis. In
particular, if $A : \C \to \mc{L}(\C^\mc{A})$ is a channel, then
$a \gets A(\varepsilon)$ means that $a$ is sampled from the set
$\mc{A}$ by measuring in the computational basis the output of $A$ on
the empty input.

On occasion and to add clarity, we will add a subscript to a ket, bra,
or linear operator from a Hilbert space to itself to make explicit the
Hilbert space related to this object. For example, we may write
$\ket{\psi}_\tsf{F} \in \tsf{F}$ instead of simply
$\ket{\psi} \in \tsf{F}$. In this case, $\bra{\psi}_\tsf{F}$ is the
associated dual. Similarly, we may write $L_\tsf{H}$ if
$L \in \mc{L}(\tsf{H})$. Context will make clear when a subscript
denotes a Hilbert space in this manner, or when it is used to convey
other information.

Finally, we occasionally refer to a quantum system as a \emph{register},
and to any quantum system with states space $\C^{\{\bz,\bo\}}$ as a
\emph{qubit}, the quantum analogue of a bit.

%----------------------------------------------------------------------%
\subsection{Trace Distance}
\label{tee:sc:td}
%----------------------------------------------------------------------%

This work will make extensive use of the trace distance between linear
operators and so we review here its definition, some related basic
facts, as well as a few simple --- but less well known --- technical
properties.

First, recall that the trace distance between two linear operators
$A, B \in \mc{L}(\tsf{H})$ is defined to be half the Schatten 1-norm of
their difference:
\begin{equation}
	\frac{1}{2}\norm{A - B}_1 = \frac{1}{2}\Tr\sqrt{(A-B)^\dag(A-B)}.
\end{equation}
It is a metric on the space of linear operators. In particular it
satisfies the triangle inequality. Moreover, for pairs of density
operators, its value is between zero and one.

The importance of the trace distance in quantum information theory comes
from the fact that it is an appropriate ``operational'' measure of the
distance between quantum states. Indeed, suppose a party is presented
uniformly at random with one of two quantum systems whose states are
described by two density operators $\rho_0, \rho_1 \in \mc{D}(\tsf{H})$.
Then, it can be shown that the party can correctly determine with which
system it was presented with a probability of at most
$\frac{1}{2}\left(1 + \frac{1}{2}\norm{\rho_0 - \rho_1}_1\right)$
\cite[Prop.~1]{FG99}\footnote{Note that Fuchs and van de Graaf consider
in their proposition the probability of error when trying to distinguish
between $\rho_0$ and $\rho_1$, while we consider the probability of
success.}, a value which ranges from one half to one as the trace
distance ranges from zero to one. Standard proofs of this fact can be
found in many textbooks (e.g.: \cite[Sec.~9.2.1]{NC10}) and we will
offer our own, for a slightly more general setting, shortly in
\cref{th:trace-channel}.

Before, however, we recall the following two other facts concerning
the trace distance:
\begin{itemize}
	\item
		The trace distance is contractive under the action of channels.
		More formally, we mean that for any linear operators
		$A, B \in \mc{L}(\tsf{H})$ and any channel $\Phi$ whose domain
		is~$\mc{L}(\tsf{H})$, it holds that $
			\frac{1}{2}\norm{\Phi(A)-\Phi(B)}_1
			\leq
			\frac{1}{2}\norm{A-B}_1
		$. \emph{Operationally, this implies that applying channels
		cannot increase the distinguishability of two quantum systems.}
	\item
		The trace distance is invariant under conjugation by isometries.
		More formally, we mean that for any isometry
		$V \in \mc{U}(\tsf{H}, \tsf{H}')$ and any linear operators $A, B
		\in \mc{L}(\tsf{H})$ it holds that
		$\frac{1}{2}\norm{VAV^\dag - VBV^\dag}_1 = \frac{1}{2}\norm{A -
		B}_1$.
\end{itemize}

As mentioned, the trace distance between two density operators is
essentially a measure of how well quantum systems whose states are
represented by these density operators can be distinguished.
For our work, it will be useful to frame this via bounds concerning
channels which attempt to map one state to $\bz$ and the other to $\bo$.
We formalize this idea below for a class of linear operators which
extends beyond only density operators, including, in particular, any
pair of subnormalized density operators.
We give an explicit proof of this lemma as it is a slightly unorthodox
way to state this result and most standard references only treat the
more restrictive case considering pairs of density operators.

\begin{lemma}
\label{th:trace-channel}
	Let $A, B \in \mc{L}(\tsf{H})$ be two linear operators such that
	$A - B$ is Hermitian.
	Then, for all channels $\Phi : \mc{L}(\tsf{H}) \to \mc{L}(\C^{\bs})$
	we have that
	\begin{equation}
		\abs{\bra{\bz}\Phi(A - B)\ket{\bz}}
		+
		\abs{\bra{\bo}\Phi(A - B)\ket{\bo}}
		\leq
		\norm{A - B}_1
	\end{equation}
	which implies that
	\begin{equation}
		\bra{\bz}\Phi(A)\ket{\bz}
		+
		\bra{\bo}\Phi(B)\ket{\bo}
		\leq
		\frac{\Tr(A) + \Tr(B)}{2}
		+
		\frac{1}{2}
		\norm{A - B}_1.
	\end{equation}
	Moreover, there exists a channel, which depends on $A$
	and $B$, saturating these inequalities.
\end{lemma}

We make two remarks on this lemma before giving its proof.
\begin{enumerate}
	\item
		If $\rho_0, \rho_1 \in \mc{D}(\tsf{H})$ are density operators,
		implying that $\Tr(\rho_0) = \Tr(\rho_1) = 1$, then
		\begin{equation}
			\frac{
				\bra{\bz} \Phi(\rho_0) \ket{\bz}
				+
				\bra{\bo} \Phi(\rho_1) \ket{\bo}
			}{2}
			\leq
			\frac{1}{2}
			\left(1 + \frac{1}{2} \norm{\rho - \sigma}_1\right)
			.
		\end{equation}
		The left-hand side of this inequality is the success probability
		of a party attempting to distinguish, with the channel $\Phi$,
		between the state $\rho_0$ (by mapping it to $\bz$) and the
		state $\rho_1$ (by mapping it to $\bo$) if presented with one of
		these states sampled uniformly randomly. The right-hand side is
		precisely the previously advertised upper bound on the success
		probability of such an attempt.
	\item
		If $\Tr(A) = \Tr(B)$, then
		\begin{equation}
			\abs{\bra{\bz}\Phi(A - B)\ket{\bz}}
			=
			\abs{\Tr(A - B) - \bra{\bo}\Phi(A - B)\ket{\bo}}
			=
			\abs{\bra{\bo} \Phi(A - B)\ket{\bo}}
		\end{equation}
		as $\Phi$ is trace preserving. This implies that
		$\abs{\bra{b}\Phi(A-B)\ket{b}} \leq \frac{1}{2}\norm{A - B}_1$
		for both values of $b \in \bs$. However, in the general case,
		such as if $A$ and $B$ are two subnormalized density operators
		with distinct traces, the right-hand side of this inequality may
		be too small: it may only hold that
		$\abs{\bra{b}\Phi(A-B)\ket{b}}\leq\norm{A-B}_1$.\footnote{%
			\emph{E.g.}: Taking $\Phi = \Id_{\C^{\bs{}}}$,
            $A = \ketbra{\bz}$, and $B = \frac{1}{2}A$ yields
			$\abs{\bra{\bz}\Phi(A - B)\ket{\bz}} = \frac{1}{2}$ and
			$\norm{A - B}_1 = \frac{1}{2}$.}
		Cases like these --- which we will encounter --- are often
		omitted from standard treatments.
\end{enumerate}
\begin{proof}
	We first show the two inequalities and then demonstrate the
	existence of a saturating channel.

	Define the channel $\tilde{\Phi} = \Delta_{\bs} \circ \Phi$.
	Then,
	\begin{equation}
		\tilde\Phi(A) - \tilde\Phi(B)
		=
		\bra{\bz} \Phi(A - B) \ket{\bz} \cdot \ketbra{\bz}
		+
		\bra{\bo} \Phi(A - B) \ket{\bo} \cdot \ketbra{\bo}
	\end{equation}
	from which it follows, by direct computation of the relevant
	Schatten $1$-norm, that
	\begin{equation}
		\norm{\tilde\Phi(A) - \tilde\Phi(B)}_1
		=
		\abs{\bra{\bz} \Phi(A - B) \ket{\bz}}
		+
		\abs{\bra{\bo} \Phi(A - B) \ket{\bo}}.
	\end{equation}
	As the trace distance is contractive under the action of channels,
	we have that
	\begin{equation}
		\abs{\bra{\bz} \Phi(A - B) \ket{\bz}}
		+
		\abs{\bra{\bo} \Phi(A - B) \ket{\bo}}
		\leq
		\norm{A - B}_1
	\end{equation}
	which yields the first inequality.
	To obtain the second inequality, we note that
	\begin{equation}
	\begin{aligned}
		2\bra{\bz}\Phi(A)\ket{\bz}
		-
		\Tr(A)
		&=
		2\bra{\bz}\Phi(A)\ket{\bz}
		-
		\Tr\circ \Phi(A)
		\\&=
		2\bra{\bz}\Phi(A)\ket{\bz}
		-
		\bra{\bz}\Phi(A)\ket{\bz} - \bra{\bo}\Phi(A)\ket{\bo}
		\\&=
		\bra{\bz}\Phi(A)\ket{\bz} - \bra{\bo}\Phi(A)\ket{\bo}
	\end{aligned}
	\end{equation}
	and, similarly, that $2\bra{\bo}\Phi(B)\ket{\bo}-\Tr(B) = \bra{\bo}
	\Phi(B) \ket{\bo} - \bra{\bz} \Phi(B) \ket{\bz}$.
	Thus,
	\begin{equation}
	\begin{aligned}
		2\bra{\bz}\Phi(A)\ket{\bz} + 2\bra{\bo}\Phi(B)\ket{\bo}
		-
		\Tr(A + B)
		&=
		\bra{\bz}\Phi(A - B)\ket{\bz}
		+
		\bra{\bo}\Phi(B - A)\ket{\bo}
		\\&\leq
		\abs{\bra{\bz}\Phi(A - B)\ket{\bz}}
		+
		\abs{\bra{\bo}\Phi(A - B)\ket{\bo}}
		\\&\leq
		\norm{A - B}_1
	\end{aligned}
	\end{equation}
	from which the desired result follows after simple algebraic
	manipulations.

	We now show the existence of a channel
	$\Psi : \mc{L}(\tsf{H}) \to \mc{L}(\C^{\bs})$ which saturates these
	bounds. This part of the proof follows closely the corresponding
	proof given by Nielsen and Chuang \cite[Sec.~9.1.2]{NC10}.
	Note that since $A - B$ is Hermitian by assumption, we may write
	\begin{equation}
		A - B
		=
		\sum_{j \in \mc{J}} \lambda_j \Pi_j
		=
			\sum_{\substack{j \in \mc{J}\\\lambda_j \geq 0}}
			\lambda_j
			\Pi_j
		+
			\sum_{\substack{j \in \mc{J}\\\lambda_j < 0}}
			\lambda_j
			\Pi_j
	\end{equation}
	for some indexing set $\mc{J}$ and where each $\lambda_j$ is a real
	number and the $\Pi_j$'s are rank
	one projectors on $\tsf{H}$ which together sum to the identity
	$I_\tsf{H}$.\footnote{%
		This is the spectral decomposition of $A - B$
		\cite[Th.~1.3]{Wat18}. The fact that the $\lambda_j$'s are
		strictly real, and not complex, follows from the fact
		that $A - B$ is Hermitian.}
	At this point, we recall an alternative characterization of the
	Schatten 1-norm of an operator, namely that it is the sum of its
	singular values \cite[Sec.~1.1]{Wat18}. Thus,
	$\norm{A - B}_1 = \sum_{j \in \mc{J}} \abs{\lambda_j}$.
	Now, let
	\begin{equation}
		P = \sum_{\substack{j \in \mc{J}\\\lambda_j \geq 0}} \Pi_j
		\qq{and}
		Q = \sum_{\substack{j \in \mc{J}\\\lambda_j < 0}} \Pi_j
	\end{equation}
	and note that both are positive semidefinite operators which sum to
	the identity.
	Hence, we may define the measurement $\mu : \bs \to \Pos(\tsf{H})$
	by $\mu(\bz) = P$ and $\mu(\bo) = Q$ and consider the channel
	$\Psi = \Ms(\mu)$, \emph{i.e.}:
	$\rho \mapsto \Tr(P\rho) \cdot \ketbra{\bz} + \Tr(Q\rho) \cdot
	\ketbra{\bo}$. Then,
	\begin{equation}
		\bra{\bz} \Psi(A - B)\ket{\bz}
		=
		\Tr\left(P(A - B)\right)
		=
		\sum_{\substack{j \in \mc{J}\\\lambda_j \geq 0}}
		\lambda_j
		=
		\sum_{\substack{j \in \mc{J}\\\lambda_j \geq 0}}
		\abs{\lambda_j}
	\end{equation}
	and
	\begin{equation}
		\bra{\bo} \Psi(B - A)\ket{\bo}
		=
		\Tr\left(Q(B - A)\right)
		=
		-\Tr\left(Q(A - B)\right)
		=
		-
		\sum_{\substack{j \in \mc{J}\\\lambda_j < 0}}
		\lambda_j
		=
		\sum_{\substack{j \in \mc{J}\\\lambda_j < 0}}
		\abs{\lambda_j}
	\end{equation}
	from which it follows that
	\begin{equation}
	\begin{aligned}
		\abs{\bra{\bz}\Psi(A - B)\ket{\bz}}
		+
		\abs{\bra{\bo}\Psi(A - B)\ket{\bo}}
		&=
		\bra{\bz}\Psi(A - B)\ket{\bz}
		+
		\bra{\bo}\Psi(B - A)\ket{\bo}
		\\&=
		\sum_{j \in \mc{J}} \abs{\lambda_j}
		\\&=
		\norm{A - B}_1,
	\end{aligned}
	\end{equation}
	which yields the desired result.
\end{proof}

We finish this exposition on the trace distance with three technical
lemmas giving values or bounds of this norm in particular cases.

First, we recall an explicit computation of the trace distance for
certain simple cases.
\begin{lemma}[{\cite[Eq.~1.184]{Wat18}}]
\label{th:td-pure}
	Let $\ket{\psi}, \ket{\phi}$ be arbitrary vectors in a Hilbert space
	$\tsf{H}$. Then,
	\begin{equation}
		\frac{1}{2}\norm{\ketbra{\psi} - \ketbra{\phi}}_1
		=
		\sqrt{
			\left(\frac{\braket{\psi} + \braket{\phi}}{2}\right)^2
			-
			\abs{\braket{\psi}{\phi}}^2
		}.
	\end{equation}
\end{lemma}

Second, the following lemma pertains to the trace distance between
multiple copies of two density operators. The proof can be found in the
full version of \cite{BCG+02}. Conceptually, this can be understood as
giving a lower bound on the ability to distinguish between $n$ copies of
a state $\rho$ and $n$ copies of a state $\sigma$ by attempting $n$
times to distinguish between a single copy of $\rho$ and a single copy
of $\sigma$ and taking a majority vote of the result.

\begin{lemma}[{\cite{BCG+02}}]
\label{th:trace-distance-copies}
	Let $\rho, \sigma \in \mc{D}(\tsf{H})$ be two density operators.
	Then, for all $t \in \N$, we have that
	\begin{equation}
		\frac{1}{2}\norm{\rho^{\tensor t} - \sigma^{\tensor t}}_1
		\geq
		1
		-
		2\exp\left(
			-\frac{t}{2}
			\left(\frac{1}{2}\norm{\rho - \sigma}_1\right)^2
		\right)
	\end{equation}
    where $\rho^{\tensor t} = \rho \tensor \rho \tensor \cdots \tensor
    \rho$ is the tensor product of $t$ copies of $\rho$, and similarly
    for $\sigma^{\tensor t}$.
\end{lemma}

And, third, we give a lemma concerning the trace distance between
two linear operators with a particular tensor decomposition. For
completion, we provide a proof of this lemma.

\begin{lemma}
\label{th:trace-orthogonal}
	Let $\mc{X}$ be an alphabet. For any linear
	operators
	$\{A_x\}_{x \in \mc{X}}, \{B_x\}_{x \in \mc{X}} \subseteq
	\mc{L}\left(\tsf{H}\right)$
	it holds that
	\begin{equation}
		\norm{
			\sum_{x \in \mc{X}}x \tensor A_{x}
			-
			\sum_{x \in \mc{X}}x \tensor B_{x}
		}_1
		=
		\sum_{x \in \mc{X}}
		\norm{A_x - B_x}_1.
	\end{equation}
\end{lemma}
\begin{proof}
	First, factor the common $x$ terms. Then, by definition, we have
	that
	\begin{equation}
	\begin{aligned}
		\norm{
			\sum_{x \in \mc{X}}x\tensor\left(A_{x} - B_{x}\right)
		}_1
		&=
		\Tr\sqrt{
			\left(
				\sum_{x \in \mc{X}}
				x \tensor
				\left(A_{x} - B_{x}\right)
			\right)^\dag
			\left(
				\sum_{x \in \mc{X}}
				x \tensor
				\left(A_{x} - B_{x}\right)
			\right)
		}
		\\&=
		\Tr\sqrt{
			\sum_{x \in \mc{X}} x \tensor (A_x - B_x)^\dag(A_x - B_x)
		}
		\\&=
		\Tr\left(
			\sum_{x \in \mc{X}}
			x \tensor \sqrt{(A_x -
		B_x)^\dag(A_x - B_x)}\right)
		\\&=
		\sum_{x \in \mc{X}}
		\Tr\left(\sqrt{(A_x - B_x)^\dag(A_x -
		B_x)}\right)
		\\&=
		\sum_{x \in \mc{X}}\norm{A_x - B_x}_1
	\end{aligned}
	\end{equation}
	where the second and third equalities follow from the fact that
	$xx' = \ketbra{x}\ketbra{x'}$ is the
	zero operator if $x \not= x'$ and is $x = \ketbra{x}$ if $x = x'$.
	The fourth equality follows from the linearity of the trace, the
	fact that $\Tr(Y \tensor Z) = \Tr(Y) \cdot \Tr(Z)$, and that
	$\Tr(x) = 1$.
\end{proof}

%----------------------------------------------------------------------%
\subsection{The Coherent Gentle Measurement}
\label{tee:sc:cgm}
%----------------------------------------------------------------------%

In general, a measurement is a destructive or at least perturbative
operation on a quantum state. However, it is well known that when a
particular measurement outcome is highly likely, a measurement may not
disturb a state by much. Formalizations of this idea include
\emph{gentle measurements} \cite{Win99}, \emph{coherent gentle
measurements} \cite[Sec. 9.4]{Wil17}, and the ``\emph{almost as good
as new}'' lemma \cite{Aar16arxiv}.

Here, we give a formalization of this idea by constructing a
mapping that takes any channel $\Phi$ to its ``coherent gentle
measurement'' counterpart $\CGM(\Phi)$. This construction satisfies the
following gentle measurement-like property: Suppose $\rho$ was a state
such that measuring $\Phi(\rho)$ in the computational basis yielded a
particular outcome with very high probability. Then, $\CGM(\Phi)(\rho)$
yields a state very close to $\rho$ tensored with the likely measurement
outcome. In other words, $\CGM(\Phi)$ allows us to keep $\rho$
\emph{and} obtain the outcome of measuring $\Phi(\rho)$ in the
computational basis, up to some well understood error.

Our formalization is a slight extension of an exercise in Wilde's
textbook \cite[Ex.~9.4.2]{Wil17}. Our result is directly applicable to
all channels instead of only measurements, yields a well defined map
instead of only providing a proof of existence, and is applicable to all
positive semidefinite operators, not only density operators.

We will proceed in two steps. We first define the map $\Phi \mapsto
\CGM(\Phi)$ and then we prove that it possesses the desired properties.
First, however, we recall the following theorem which states the
existence of a one-to-one correspondence between measurements and a
certain class of channels.
\begin{theorem}[{\cite[Th.~2.37]{Wat18}}]
	Let $\mc{A}$ be an alphabet and
	$\Phi : \mc{L}(\tsf{H}) \to \mc{L}(\C^\mc{A})$ be a channel.
	Then, there exists a unique measurement
	$\mu_\Phi : \mc{A} \to \Pos(\tsf{H})$
	satisfying
	\begin{equation}
			\Delta_{\mc{A}}\circ\Phi=\Ms(\mu_\Phi).
	\end{equation}
\end{theorem}

\begin{definition}
\label{df:cgm}
	Let $\tsf{X}$ be a Hilbert space and $\mc{Y}$ be an alphabet.
	The \emph{coherent gentle measurement} map $
		\CGM :
		\text{Ch}(\tsf{X}, \C^\mc{Y})
		\to
		\text{Ch}(\tsf{X}, \tsf{X} \tensor \C^\mc{Y})
	$ is defined by
	\begin{equation}
		\label{eq:cgm-def}
		\Phi \mapsto
		\left[
			\rho
			\mapsto
			\left(
				\sum_{y \in \mc{Y}}
				\sqrt{\mu_\Phi(y)} \tensor \ket{y}
			\right)
			\rho
			\left(
				\sum_{y \in \mc{Y}}
				\sqrt{\mu_\Phi(y)}^\dag \tensor \bra{y}
			\right)
		\right]
	\end{equation}
	where $\mu_\Phi : \mc{Y} \to \Pos(\tsf{X})$ is the unique
	measurement satisfying $\Delta_\mc{Y} \circ \Phi = \Ms(\mu_\Phi)$.
\end{definition}

\begin{remark}
	\Cref{df:cgm} asserts that $\CGM$ maps channels to \emph{channels}.
	This needs to be shown.
	To show that $\CGM(\Phi)$ is a channel, it suffices to show that
	$V = \sum_{y \in \mc{Y}} \sqrt{\mu(y)} \tensor \ket{y}$ is an
	isometry in $\mc{U}(\tsf{X}, \tsf{X} \tensor \tsf{Y})$.
	This follows from the fact that
	\begin{equation}
		V^\dag V
		=
		\sum_{y,y' \in \mc{Y}}
		\left(\sqrt{\mu(y)}^\dag \tensor \bra{y}\right)
		\left(\sqrt{\mu(y')} \tensor \ket{y'}\right)
		=
		\sum_{y \in \mc{Y}}
		\mu(y)
		=
		I_\tsf{X}
	\end{equation}
	where we recall that $\sqrt{\mu(y)}^\dag = \sqrt{\mu(y)}$ since
	$\sqrt{\mu(y)}$ is positive and hence self-adjoint.
	By the same token, we can replace the $\sqrt{\mu(y)}^\dag$ terms in
	\cref{eq:cgm-def} with $\sqrt{\mu(y)}$, which is what we will do in
	practice.
\end{remark}

The following theorem gives two properties of the coherent gentle
measurement map. The first bounds the trace distance between the
real output $\CGM(\Phi)(\rho)$ and an ``ideal'' output $\rho \tensor
\ketbra{\tilde{y}}$ based on the probability that measuring $\Phi(\rho)$
in the computational basis would yield $\tilde{y}$ as outcome. The
second states that measuring the $\tsf{Y}$ register
of either $\Phi(\rho)$ or $\CGM(\Phi)(\rho)$ in the computational basis
yields the same results.

\begin{theorem}
\label{th:cgm}
	Let $\tsf{X}$ be a Hilbert space and $\mc{Y}$ be an alphabet.
	Then, the coherent gentle measurement map $
		\CGM
		:
		\text{Ch}(\tsf{X}, \C^\mc{Y})
		\to
		\text{Ch}(\tsf{X}, \tsf{X} \tensor \C^\mc{Y})
	$ given in \cref{df:cgm} satisfies the following two properties.
	\begin{enumerate}
		\item
			For any channel $\Phi: \mc{L}(\tsf{X}) \to \mc{L}(\tsf{Y})$,
			any positive semidefinite operator $\rho \in \Pos(\tsf{X})$,
			and any $y \in \mc{Y}$, we have that
			\begin{equation}
			\label{eq:cgm}
				\frac{1}{2}
				\norm{
					\CGM(\Phi)(\rho) - \rho \tensor \ketbra{y}
				}_1
				\leq
				\sqrt{
					\Tr(\rho)^2
					-
					(\bra{y} \Phi(\rho) \ket{y})^2
				}.
			\end{equation}
		\item
			For any channel $\Phi: \mc{L}(\tsf{X}) \to \mc{L}(\tsf{Y})$,
			we have that
			\begin{equation}
				\left(\Tr_\tsf{X} \tensor \Delta_\mc{Y}\right)
				\circ
				\CGM(\Phi)
				=
				\Delta_\mc{Y} \circ \Phi.
			\end{equation}
	\end{enumerate}
\end{theorem}
\begin{proof}
	We begin by showing the second point.
	Let $\mu_\Phi : \mc{Y} \to \Pos(\tsf{X})$ be the unique measurement
	satisfying $\Delta_\mc{Y} \circ \Phi = \Ms(\mu_\Phi)$.
	Then, for any linear operator $\rho \in \mc{L}(\tsf{X})$ we have
	that
	\begin{equation}
	\begin{aligned}
		&
		\left(\Tr_\tsf{X} \tensor \Delta_\mc{Y}\right)
		\circ
		\CGM(\Phi)(\rho)
		\\&=
		\left(\Tr_\tsf{X} \tensor \Delta_\mc{Y}\right)
		\left(
			\sum_{y,y' \in \mc{Y}}
			\sqrt{\mu_\Phi(y)} \rho \sqrt{\mu_\Phi(y')}
			\tensor
			\ketbra{y}{y'}
		\right)
		\\&=
		\left(\Tr_\tsf{X} \tensor \Id_\tsf{Y}\right)
		\left(
			\sum_{\tilde{y} \in \mc{Y}}
			\sum_{y,y' \in \mc{Y}}
			\sqrt{\mu_\Phi(y)} \rho \sqrt{\mu_\Phi(y')}
			\tensor
			\bra{\tilde{y}}\ket{y}\braket{y'}{\tilde{y}}
			\cdot
			\ketbra{\tilde{y}}
		\right)
		\\&=
		\left(\Tr_\tsf{X} \tensor \Id_\tsf{Y}\right)
		\left(
			\sum_{\tilde{y} \in \mc{Y}}
			\sqrt{\mu_\Phi(\tilde{y})} \rho \sqrt{\mu_\Phi(\tilde{y})}
			\tensor
			\ketbra{\tilde{y}}
		\right)
		\\&=
		\sum_{\tilde{y} \in \mc{Y}}
		\Tr\left(\mu_\Phi(\tilde{y})\rho\right)\cdot\ketbra{\tilde{y}}
		\\&=
		\Ms(\mu_\Phi)(\rho)
		\\&=
		\Delta_\mc{Y} \circ \Phi(\rho)
	\end{aligned}
	\end{equation}
	which yields the desired result.

	We now show the first point.
	Once again, let $\mu_\Phi : \mc{Y} \to \Pos(\tsf{X})$ be the unique
	measurement satisfying $\Delta_\mc{Y} \circ \Phi = \Ms(\mu_\Phi)$
	and let $
		V
		=
		\sum_{\tilde{y} \in \mc{Y}}
		\sqrt{\mu_\Phi(\tilde{y})} \tensor \ket{\tilde{y}}$.
	As $\rho \in \Pos(\tsf{X})$ is a positive operator, there is
	a vector $\ket{\psi} \in \tsf{X} \tensor \tsf{X}$ such that
	$
		\rho =
		\left(
			\Tr_\tsf{X}\tensor\Id_\tsf{X}
		\right)(\ketbra{\psi})
	$.\footnote{%
		Such a vector is known as a \emph{purification} of
		$\rho$ and is guaranteed to exist \cite[Th.~2.20]{Wat18}.}
	Note that $\ket{\psi}$ need not have norm one as we are not assuming
	that $\rho$ is a density operator, only that it is positive
	semidefinite. As the trace distance is contractive under the action
	of channels, we have that
	\begin{equation}
	\begin{aligned}
		&
		\frac{1}{2}\norm{
			\CGM(\Phi)(\rho)
			-
			\rho \tensor \ketbra{y}
		}_1
		\\&=
		\frac{1}{2}\norm{
			\left(\Tr_\tsf{X} \tensor \CGM(\Phi)\right)(\ketbra{\psi})
			-
			\left(\Tr_\tsf{X} \tensor \Id_{\tsf{X}\tensor\tsf{Y}}\right)
			\left(\ketbra{\psi} \tensor \ketbra{y}\right)
		}_1
		\\&\leq
		\frac{1}{2}\norm{
			\left(\Id_\tsf{X} \tensor \CGM(\Phi)\right)(\ketbra{\psi})
			-
			\left(\ketbra{\psi} \tensor \ketbra{y}\right)
		}_1.
		\\&=
		\frac{1}{2}\norm{
			(I_\tsf{X} \tensor V)
			\ketbra{\psi}
			(I_\tsf{X} \tensor V^\dag)
			-
			\ketbra{\psi} \tensor \ketbra{y}
		}_1.
	\end{aligned}
	\end{equation}
	We can then use \cref{th:td-pure} to compute this final trace
	distance. As $V$ is an isometry, we have that
	\begin{equation}
		\Tr\left(
			(I_\tsf{X} \tensor V)
			\ketbra{\psi}
			(I_\tsf{X} \tensor V^\dag)
		\right)
		=
		\bra{\psi}(I_\tsf{X} \tensor V^\dag V)\ket{\psi}
		=
		\braket{\psi}.
	\end{equation}
	We also have that
	\begin{equation}
		\Tr(\ketbra{\psi} \tensor \ketbra{y})
		=
		\braket{\psi} \cdot \braket{y}
		=
		\braket{\psi}.
	\end{equation}
	Further, both these quantities are also equal to
	$\Tr(\rho)$. Thus, by \cref{th:td-pure}, it follows that
	\begin{equation}
	\begin{aligned}
		\frac{1}{2}\norm{
			(I_\tsf{X} \tensor V)
			\ketbra{\psi}
			(I_\tsf{X} \tensor V^\dag)
			-
			\ketbra{\psi} \tensor \ketbra{y}
		}_1
		&=
		\sqrt{
			\Tr(\rho)^2
			-
			\abs{
				\left(\bra{\psi} \tensor \bra{y}\right)
				(I_\tsf{X} \tensor V)
				\ket{\psi}
			}^2
		}
		\\&=
		\sqrt{
			\Tr(\rho)^2
			-
			\abs{
				\bra{\psi}
				(I_\tsf{X} \tensor \sqrt{\mu(y)})
				\ket{\psi}
			}^2
		}
	\end{aligned}
	\end{equation}
	Now, note that
	\begin{equation}
		\bra{\psi} (I_\tsf{X} \tensor \sqrt{\mu(y)}) \ket{\psi}
		=
		\Tr_{\tsf{X} \tensor \tsf{X}}
		\left(
			\left(I_\tsf{X} \tensor \sqrt{\mu(y)}\right)
			\ketbra{\psi}
		\right)
		=
		\Tr_{\tsf{X}}\left(
			\sqrt{\mu(y)} \rho
		\right)
	\end{equation}
	and that this quantity must be real and non-negative by virtue
	of $I_\tsf{X} \tensor \sqrt{\mu(y)}$ being a positive
	semidefinite operator.
	We now claim that $\Tr_\tsf{X}(\sqrt{\mu(y)} \rho) \geq
	\Tr_\tsf{X}(\mu(y) \rho)$.
	This is obtained by considering the spectral decomposition
	$\mu(\tilde{y}) = \sum_{j \in [\dim\tsf{X}]} \lambda_j \Pi_j$ where
	each $\Pi_j$ is a projector on orthogonal subspaces and
	noting that that every $\lambda_j$ is real, non-negative, and no
	greater than one. The first two properties of these $\lambda_j$'s
	follow from the fact that
	$\mu(y) \in \Pos(\tsf{X})$ and the last from the fact that
	$\sum_{\tilde{y} \in \mc{Y}} \mu(\tilde{y}) = I_\tsf{X}$. Thus, we
	find that
	$\sqrt{\mu(y)} = \sum_{j \in [\dim\tsf{X}]} \sqrt{\lambda_j} \Pi_j$
	and so
	\begin{equation}
		\Tr_\tsf{X}\left(\sqrt{\mu(y)} \rho\right)
		=
		\sum_{j \in [\dim\tsf{X}]}
			\sqrt{\lambda_j}
			\Tr\left(\Pi_j \rho\right)
		\geq
		\sum_{j \in [\dim\tsf{X}]} \lambda_j
		\Tr\left(\Pi_j \rho\right)
		=
		\Tr\left(\mu(y) \rho\right).
	\end{equation}
	Finally, $
		\Tr(\mu(y) \rho)
		=
		\bra{y} \Ms(\mu)(\rho) \ket{y}
		=
		\bra{y} \Delta_\mc{Y} \circ \Phi(\rho) \ket{y}
		=
		\bra{y} \Phi(\rho) \ket{y}
	$. Collecting the above facts immediately yields that
	\begin{equation}
		\frac{1}{2}
		\norm{
			\CGM(\Phi)(\rho) - \rho \tensor \ketbra{y}
		}_1
		\leq
		\sqrt{\Tr(\rho)^2 -
		\left(\bra{y}\Phi(\rho)\ket{y}\right)^2}
	\end{equation}
	which is the desired result.
\end{proof}

%======================================================================%
\section{Defining Tamper Evidence}                                  %
\label{sc:qecm}                                                        %
\label{sc:aqecm}                                                       %
\label{te:sc:aqecm}
%======================================================================%

The goal of this section is twofold. First, we formally define the main
objects which we will study in this work, namely \emph{quantum
encryptions of classical messages} (QECM) schemes and their
\emph{augmented} variants, AQECM schemes. Second, we formalize the
notion of \emph{tamper evidence}, a property which can be possessed by
augmented quantum encryption of classical messages schemes.

Recall that QECMs and AQECMs are keyed encoding schemes for classical
messages which produce quantum states. While such schemes were implicit
in the literature since at least the early 2000s, such as in Gottesman's
work \cite{Got03}, this terminology and precise formalization was first
introduced in our prior work on uncloneable
encryption \cite{BL20} for QECMs, and subsequently expanded by Broadbent
and Islam to AQECMs in the course of their work on certfied deletion
\cite{BI20}.

%~~~~~~~~~~~~~~~~~~~~~~~~~~~~~~~~~~~~~~~~~~~~~~~~~~~~~~~~~~~~~~~~~~~~~~%
\subsection{Quantum Encryption of Classical Messages Schemes}
\label{sc:te-correctness}
%~~~~~~~~~~~~~~~~~~~~~~~~~~~~~~~~~~~~~~~~~~~~~~~~~~~~~~~~~~~~~~~~~~~~~~%

A quantum encryption of classical messages (QECM) scheme is a triplet of
channels $(K, E, D)$ called, respectively, the \emph{key generation},
\emph{encoding}, and \emph{decoding} channels.
These channels are required to satisfy certain constraints on their
domains and codomains to ensure that they are compatible with one
another. Implicitly, these requirements on the domains and codomains
also define two alphabets $\mc{K}$ and $\mc{M}$ which represent,
respectively, the set of classical keys and messages for this scheme.

\begin{definition}
\label{df:qecm}
	A \emph{quantum encryption of classical messages} (QECM) scheme is a
	triplet of channels $(K, E, D)$ of the form
	\begin{align}
%		&
		K :
		\mc{L}(\C)
		\to
		\mc{L}\left(\tsf{K}\right)
		\qq*{,}
%		\\
%		&
		E :
		\mc{L}\left(\tsf{K} \tensor \tsf{M}\right)
		\to
		\mc{L}\left(\tsf{C}\right)
		\text{,}\qq{and}
%		\\
%		&
		D :
		\mc{L}\left(\tsf{K} \tensor \tsf{C} \right)
		\to
		\mc{L}\left(\tsf{M}\right)
	\end{align}
	for Hilbert spaces $\tsf{K}$, $\tsf{M}$, and $\tsf{C}$ where,
	$\tsf{K} = \C^\mc{K}$ and $\tsf{M} = \C^\mc{M}$ for alphabets
	$\mc{K}$ and $\mc{M}$.
	We call the elements of these alphabets, respectively, the
	\emph{keys} and \emph{messages} of the scheme.
	
	To lighten notation, for every key $k \in \mc{K}$ we define the
	channels
	$
		E_k :
		\mc{L}\left(\tsf{M}\right)
		\to
		\mc{L}\left(\tsf{C}\right)
	$ and~$
		D_k :
		\mc{L}\left(\tsf{C}\right)
		\to
		\mc{L}\left(\tsf{M}\right)
	$ by
	$\rho\mapsto E(k\tensor\rho)$ and $\rho\mapsto D(k\tensor\rho)$,
	respectively.
\end{definition}

From this point on, to ease notation, whenever we talk of a
``QECM scheme $(K, E, D)$ as given in \cref{df:qecm}'', we are assuming
that the domains and codomains are defined and denoted exactly as in
\cref{df:qecm}.

We now define a quantitative notion of \emph{correctness} for QECM
schemes. In short, a QECM scheme is $\epsilon$-correct if for every
message $m \in \mc{M}$ the probability that it is correctly recovered by
the decoding map $D_k$ after being processed by the encoding map $E_k$
is at least~$1 - \epsilon$. Here, the probability is taken over the
action of the maps $E_k$ and $D_k$ as well as the randomness of
sampling the key $k \gets K(\varepsilon)$ from the key generation
procedure.

\begin{definition}
	Let $\epsilon \in \R$.
	A QECM scheme $(K, E, D)$ as given in \cref{df:qecm} is
	\emph{$\epsilon$-correct} if for all messages $m \in \mc{M}$ we have
	that
	\begin{equation}
		\E_{k \gets K(\emptystring)}
			\bra{m}
				D_k \circ E_k
				\left(
					m
				\right)
			\ket{m}
			\geq
			1 - \epsilon.
	\end{equation}
	We say that a QECM scheme is \emph{perfectly correct}, or simply
	\emph{correct}, if it is $0$-correct.
\end{definition}

We emphasize that correctness for QECMs only considers
\emph{classical} messages $m \in \mc{M}$. A QECM need not be able to
recover any arbitrary quantum state to be correct.
As an example, consider the QECM with $\mc{M} = \bs{}$ and
$\mc{K} = \{\emptystring\}$ where the channel $E_\emptystring$ is given
by $\rho \mapsto Z \rho Z^\dag$ for the Pauli operator
$Z = \ketbra{\bz} - \ketbra{\bo}$ and $D_\emptystring = \Id$.
A direct calculation yields $
	D_\emptystring \circ E_\emptystring(\ketbra{+})
	\not=
	\ketbra{+}
$ where $\ket{+} = \frac{1}{\sqrt{2}}\left(\ket{\bz}+\ket{\bo}\right)$.
However, for both $b \in \bs{}$, it holds that
$D_\emptystring \circ E_\emptystring(\ketbra{b}) = \ketbra{b}$ and so
this QECM is perfectly correct.

We now proceed to define \emph{augmented encryption of classical
messages} (AQECM) schemes.
In short, an AQECM scheme is identical to a QECM scheme, with the
exception that the decoding channel $D$ produces an additional ``flag''
qubit as part of its output.
We will denote the state space of this qubit by $\tsf{F}$.

\begin{definition}
\label{df:aqecm}
	An \emph{augmented encryption of classical messages} (AQECM) scheme
	is a triplet of channels $(K, E, D)$ of the form
	\begin{align}
%		&
		K :
		\mc{L}(\C)
		\to
		\mc{L}\left(\tsf{K}\right)
		\qq*{,}
%		\\&
		E :
		\mc{L}\left(\tsf{K} \tensor \tsf{M}\right)
		\to
		\mc{L}\left(\tsf{C}\right)
		\text{,}\qq{and}
%		\\&
		D :
		\mc{L}\left(\tsf{K} \tensor \tsf{C}\right)
		\to
		\mc{L}\left(\tsf{M} \tensor \tsf{F}\right)
	\end{align}
	for Hilbert space $\tsf{K}$, $\tsf{M}$, $\tsf{C}$, and $\tsf{F}$
	where, in particular, $\tsf{K} = \C^\mc{K}$, $\tsf{M} = \C^\mc{M}$,
	and $\tsf{F} = \C^{\bs{}}$ for alphabets $\mc{K}$ and $\mc{M}$.
	As in \cref{df:qecm}, the elements of $\mc{K}$ and $\mc{M}$ are
	called the keys and messages of the scheme.
	From $D : \mc{L}(\tsf{K} \tensor \tsf{M}) \to \mc{L}(\tsf{M} \tensor \tsf{F})$, we also define the map $\overline{D} :
	\mc{L}(\tsf{K} \tensor \tsf{C}) \to \mc{L}(\tsf{M})$ by
	\begin{equation}
		\rho
		\mapsto
		\left(I_\tsf{M} \tensor \bra{\bo}_\tsf{F}\right) D(\rho)
	\left(I_\tsf{M} \tensor \ket{\bo}_\tsf{F}\right),
	\end{equation}
	which is to say that $\overline{D}$ first applies the channel $D$
	and then ``reduces'' to the case where the qubit is measured to be
	in the $\bo$ state, followed by discarding this qubit.

	For every key $k \in \mc{K}$, we also define the maps $
		E_k : \mc{L}(\tsf{M}) \to \mc{L}(\tsf{C})
	$, $
		D_k : \mc{L}(\tsf{C}) \to \mc{L}(\tsf{M})
	$, and $
		\overline{D}_k : \mc{L}(\tsf{C}) \to \mc{L}(\tsf{M})
	$ by $
		\rho \mapsto E(k \tensor \rho)
	$, $
		\rho \mapsto D(k \tensor \rho)
	$, and $
		\rho \mapsto \overline{D}(k \tensor \rho)
	$, respectively.
\end{definition}

As for QECM schemes, from this point on whenever we mention an
``AQECM scheme $(K, E, D)$ as given in \cref{df:aqecm}''  we are
assuming that the domains and codomains are defined and denoted exactly
as in \cref{df:aqecm}.

We emphasize that the \emph{only} difference between the syntaxes
of a QECM and an AQECM scheme is the additional ``flag'' qubit in the
codomain of the decoding map $D$.

Note that the map $\overline{D}$, and by extension every map
$\overline{D}_k$, may not be a channel.
Indeed, the map $\rho \mapsto \bra{\bo} \rho \ket{\bo}$ applied
to the $\tsf{F}$ register produced by $D$ to obtain $\overline{D}$ is
not, in general, trace preserving.
It may in fact be trace \emph{decreasing},
\ie~$\Tr(\overline{D}(\sigma)) < \Tr(\sigma)$.
Nonetheless, these maps will be useful mathematical tools to define and
study properties of AQECM schemes.
We take a moment to look at them from a more operational perspective.

For any density operator $\rho$ and key $k \in \mc{K}$, we may
express $D_k(\rho)$ as
\begin{equation}
	D_k(\rho)
	=
	\rho_{\bz\bz} \tensor \ketbra{\bz}{\bz}
	+
	\rho_{\bz\bo} \tensor \ketbra{\bz}{\bo}
	+
	\rho_{\bo\bz} \tensor \ketbra{\bo}{\bz}
	+
	\rho_{\bo\bo} \tensor \ketbra{\bo}{\bo}
\end{equation}
where both $\rho_{\bz\bz}$ and $\rho_{\bo\bo}$ are subnormalized density
operators whose traces sum to one, which is to say that
$\rho_{\bz\bz}, \rho_{\bo\bo} \in \mc{D}_\bullet(\tsf{M})$ and
$\Tr(\rho_{\bz\bz}) + \Tr(\rho_{\bo\bo}) =1$.
If a computational basis measurement is then applied to the $\tsf{F}$
register, the resulting state is given by
\begin{equation}
	\left(\Id_\tsf{M} \tensor \Delta_\tsf{F}\right)
	D_k(\rho)
	=
	\rho_{\bz\bz} \tensor \ketbra{\bz}{\bz}
	+
	\rho_{\bo\bo} \tensor \ketbra{\bo}{\bo}.
\end{equation}
Then, the map $\overline{D}_k$ simply ``picks out'' the $\rho_{\bo\bo}$
operator from this expression: $\overline{D}_k(\rho) = \rho_{\bo\bo}$.
Note that $\Tr(\rho_{\bo\bo})$ is precisely the probability that the
measurement outputs $\bo$.
From this perspective, it is easy to see that $\overline{D}_k(\rho)$ is
simply the state on the $\tsf{M}$ register after application of $D_k$ to
$\rho$, conditioned on the $\tsf{F}$ register being measured to be in
the $\bo$ state and scaled by the probability that this particular
measurement outcome is produced.

The precise meaning of the flag qubit produced by the decryption map of
an AQECM scheme can and will vary by context.
In general, the flag qubit is set to the $\bz$ state if the
decryption channel $D_k$ ``rejects'' the input state and is set to
$\bo$ if it ``accepts'' it, where the precise meaning of these
words depends of the context.
For example, Broadbent and Islam \cite{BI20} use this flag to study
AQECM schemes as authentication schemes for classical messages which
produce quantum outputs.
They call these \emph{robust AQECM} schemes.
There, a flag qubit in the state $\bz$ is meant to indicate that
noise or adversarial action changed the encoded state sufficiently that
the wrong message will be recovered.
In the context of tamper evidence, the flag qubit will be set to the
state $\bz$ if non-trivial eavesdropping is detected.\footnote{
	As previously discussed, we sketch later in \cref{te:sc:te-not-auth}
	a separation robustness and tamper evidence.
}
Nonetheless, in all contexts we will expect an AQECM to accept any
unmodified encodings of its messages when the same key is used.

This discussion leads us to the following definition of correctness for
AQECM schemes.

\begin{definition}
\label{df:aqecm-correct}
	Let $\epsilon \in \R$.
	An AQECM scheme $(K, E, D)$ as given in \cref{df:aqecm} is
	\emph{$\epsilon$-correct} if for all messages $m \in \mc{M}$ we have
	that
	\begin{equation}
		\E_{k \gets K(\emptystring)}
		\bra{m} \overline{D}_k \circ E_k (m) \ket{m}
		\geq
		1 - \epsilon.
	\end{equation}
\end{definition}

Our expectation that the flag qubit should be found in the $\bo$
state is represented in the above definition by the fact that it uses
the maps $\overline{D}_k$ and not the channels $D_k$.

It is intuitively clear that AQECM schemes are, as the name suggests, an
augmentation of QECM schemes.
The following lemma formalizes this idea by confirming that discarding
the flag register $\tsf{F}$ of an AQECM scheme yields a QECM scheme and
that correctness is preserved through this transformation.

\begin{lemma}
	Let $(K, E, D)$ be an $\epsilon$-correct AQECM scheme as given in
	\cref{df:aqecm}.
	Then, $(K, E, (\Id_\tsf{M}\tensor\Tr_\tsf{F}) \circ D)$ is a QECM
	scheme which is $\epsilon$-correct.
\end{lemma}
\begin{proof}
	We directly see that $(\Id_\tsf{M}\tensor\Tr_\tsf{F}) \circ D$ is
	a channel from $\mc{L}(\tsf{K} \tensor \tsf{C})$ to
	$\mc{L}(\tsf{M})$ and so the triplet
	$(K, E, (\Id_\tsf{M} \tensor \Tr_\tsf{F}) \circ D)$ satisfies the
	syntactical conditions to be a QECM.
	It now suffices to show that this QECM is $\epsilon$-correct.

	By recalling that $\Tr_{\tsf{F}}(\rho) = \sum_{b \in
	\bs} \bra{b} \rho \ket{b}$ and the definition of
	$\overline{D}$, for every key
	$k \in \mc{K}$ and message $m \in \mc{M}$ we find that
	\begin{equation}
	\begin{aligned}
		\left(\Id_\tsf{M} \tensor \Tr_\tsf{F}\right) \circ D_k \circ E_k
		(m)
		&=
		\sum_{b \in \bs{}}
		\left(I_\tsf{M} \tensor \bra{b}\right)
			(D_k \circ E_k)(m)
		\left(I_\tsf{M} \tensor \ket{b}\right)
		\\&=
		\overline{D}_k \circ E_k(m)
		+
		\left(I_\tsf{M} \tensor \bra{\bz}\right)
			(D_k \circ E_k)(m)
		\left(I_\tsf{M} \tensor \ket{\bz}\right)
	\end{aligned}
	\end{equation}
	from which it follows that
	\begin{equation}
		\bra{m}
		\left(\Id_\tsf{M} \tensor \Tr_\tsf{F}\right) \circ D_k \circ E_k
		(m)
		\ket{m}
		\geq
		\bra{m} \overline{D}_m \circ E_k(m) \ket{m}
	\end{equation}
	since $
		\left(\bra{m} \tensor \bra{\bz}\right)
			(D_k \circ E_k)(m)
		\left(\ket{m} \tensor \ket{\bz}\right)
		\geq
		0
	$ by virtue of $D_k \circ E_k(m)$ being a positive semidefinite
	operator.
	Thus, for all $m \in \mc{M}$, we have that
	\begin{equation}
		\E_{k \gets K(\emptystring)}
		\bra{m}
			\left(
				\Id_\tsf{M} \tensor \Tr_\tsf{F}
			\right) \circ D_k \circ E_k
			(m)
		\ket{m}
		\geq
		\E_{k \gets K(\emptystring)}
			\bra{m} \overline{D}_k \circ E_k(m) \ket{m}
		\geq
		1 - \epsilon,
	\end{equation}
	as the original AQECM scheme is assumed to be $\epsilon$-correct,
	which is the desired result.
\end{proof}

%----------------------------------------------------------------------%
\subsection{Tamper Evidence}
%----------------------------------------------------------------------%

With the notion of AQECM schemes in hand, we can now define
tamper evidence.

In a nutshell, an AQECM scheme exhibits tamper evidence if the
party which receives a ciphertext --- the one implementing the decoding
channel $D$ --- can detect any non-trivial eavesdropping.
This notion was first presented by Gottesman \cite{Got03}, and we review
it here.
To place Gottesman's notion of tamper evidence in our AQECM framework,
we proceed in two steps.
We first formalize in \cref{df:te-attack} \emph{tamper attacks} against
a given AQECM schem and we then state in \cref{df:te-security} the
condition that an AQECM scheme must satisfy against all tamper attacks
to be considered \emph{tamper-evident}.
Up to a difference in presentation and
some minor differences discussed later, this is the definition of
tamper evidence given in \cite{Got03}.

A tamper attack can be understood as an adversary's attempt to extract,
undetected, information from an intercepted ciphertext  $E_k(m) \in
\mc{D}(\tsf{C})$. We model any information that the adversary managed to
extract as a quantum system held in a register $\tsf{A}$. Moreover, we
assume that the eavesdropping adversary must allow \emph{something},
\ie~a possibly modified ciphertext, to reach the receiver. Otherwise, we
may assume that the eavesdropping will necessarily be detected. Under
this model, we can represent the adversary's action as a channel which
acts on the space of the encoded messages $\tsf{C}$ to produce a
state on the space $\tsf{C} \tensor \tsf{A}$. This is formalized in
\cref{df:te-attack} and sketched in \cref{fg:te-security}.

\begin{definition}
\label{df:te-attack}
	Let $(K, E, D)$ be an AQECM scheme as given in \cref{df:aqecm}.
	A \emph{tamper attack} against it is a channel
	$A : \mc{L}(\tsf{C}) \to \mc{L}(\tsf{C} \tensor \tsf{A})$ where
	$\tsf{A}$ is a Hilbert space.
\end{definition}

\begin{figure}
	\begin{center}
		\begin{tikzpicture}[x=1.5em,y=1.5em]

\tikzset{str note/.style={xshift=+1.5em,font=\scriptsize}}
\tikzset{end note/.style={xshift=-1.5em,font=\scriptsize}}
\tikzset{mid note/.style={midway,font=\scriptsize}}

%	\draw[step=1, gray, thin] (-7,-5) grid (15,4);

	\node (E) at (-2,-1) {$E_k$};
	\draw[thick] ($(E) + ( 1, 1)$) rectangle ($(E) - ( 1, 1)$);

	\node (D) at ($(E) + (8,1.5)$) {$D_k$};
	\draw[thick] ($(D) + ( 1, 2.5)$) rectangle ($(D) - ( 1, 2.5)$);
	
	\node (A) at ($(E) + (4,-1.5)$) {$A$};
	\draw[thick] ($(A) + ( 1, 2.5)$) rectangle ($(A) - ( 1, 2.5)$);

	\node (T) at ($(D) + (4,-1.5)$) {$\Tr$};
	\draw[thick] ($(T) + (1,1)$) rectangle ($(T) - (1,1)$);

	\node (S) at ($(D) + (4, 1.5)$)
		{$\bra{\bo}\hspace{-0.1em}\cdot\hspace{-0.1em}\ket{\bo}$};
	\draw[thick,rounded corners]
		($(S) + (1,1)$) rectangle ($(S) - (1,1)$);

	\node (olD) at ($(D) +(2,3)$) {$\overline{D}_k$};
	\draw[dashed]
		($(olD) + (-3.5,0.5)$)
		to
		($(S) + (1.5,2)$)
		to
		($(S) + (1.5,-1.5)$)
		to
		($(S) + (-2.5,-1.5)$)
		to
		($(D) + (1.5,-3)$)
		to
		($(D) + (-1.5,-3)$)
		to
		($(olD) + (-3.5,0.5)$);

	% Message input.
	\draw[thick]
		($(E) - (5,0)$)
		node[str note, above] {$m$}
		node[str note, below] {$\tsf{M}$}
		to
		($(E) - (1,0)$)
	;

	% First cipher.
	\draw[thick]
		($(E) + (1, 0)$)
		to
		node[mid note, below] {$\tsf{C}$}
		($(E) + (3, 0)$)
	;

	\draw[thick]
		($(A) + (1,1.5)$)
		to
		node[mid note, below] {$\tsf{C}$}
		($(A) + (3,1.5)$)
	;

	\draw[thick]
		($(D) + (1,1.5)$)
		to
		node[mid note, below] {$\tsf{F}$}
		($(D) + (3,1.5)$)
	;

	\draw[thick]
		($(D) + (1,-1.5)$)
		to
		node[mid note, below] {$\tsf{M}$}
		($(D) + (3,-1.5)$)
	;

	\draw[thick]
		($(A) + (1,-1.5)$)
		to
		($(A) + (13,-1.5)$)
		node[end note, below] {$\tsf{A}$}
		node[end note, above] {$\rho_{k,m}$}
	;

\end{tikzpicture}
	\end{center}
	\caption[%
			Representation of the maps used to define tamper evidence.]
		{\label{fg:te-security}%
			An AQECM scheme $(K, E, D)$ subject to a
			tamper attack $A$, as considered in \cref{df:te-attack} and
			\cref{df:te-security}. For any given
			message $m$ and key $k$, the final subnormalized state held
			by the adversary if the receiver does not detect any
			eavesdropping is $\rho_{k,m} \in \mc{D}_\bullet(\tsf{A})$.}
\end{figure}
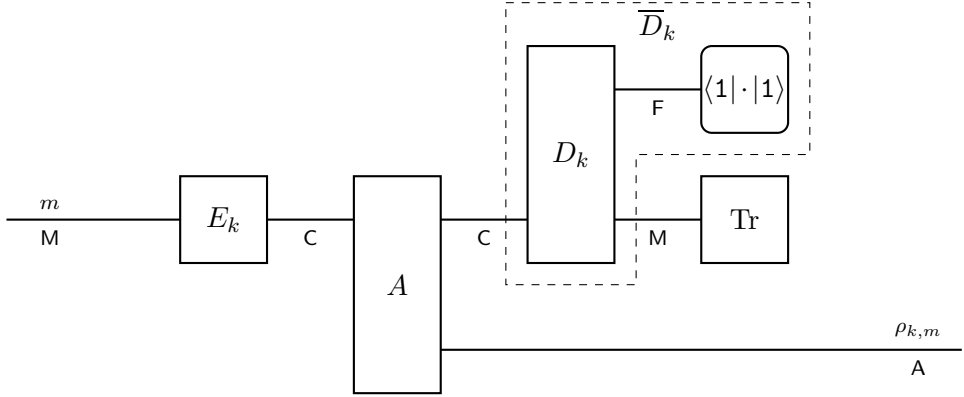

We can now formalize the notion of tamper evidence. To do so will
require us to consider certain subnormalized states acting on the
additional Hilbert space $\tsf{A}$ produced by a tamper attack $A$. We
take a moment to examine these subnormalized states in some detail
before defining tamper evidence.

For any key $k$, message $m$, and tamper attack
$A : \mc{L}(\tsf{C}) \to \mc{L}(\tsf{C} \tensor \tsf{A})$ against an
AQECM scheme $(K, E, D)$, we consider the operator
\begin{equation}
	\left(
		\left(\Tr_\tsf{M} \circ \overline{D}_k\right)
		\tensor \Id_\tsf{A}
	\right)
	\circ
	A
	\circ
	E_k
	(m),
\end{equation}
and we note that this is a subnormalized density operator on the
additional register $\tsf{A}$ produced by the tamper attack $A$,
\emph{i.e.} an element of $\mc{D}_\bullet(\tsf{A})$. In
\cref{fg:te-security}, this is the operator $\rho_{k,m}$.

Operationally, this represents the state the adversarial
eavesdropper holds when the message $m$ is sent using the key $k$,
\emph{conditioned} on the receiver accepting \emph{and scaled} by the
probability of this acceptance.

An AQECM scheme will be tamper evident if there is a high probability,
taken over the choice the of key $k$, that the trace distance between
the subnormalized states $\rho_{k,m}$ and $\rho_{k,m'}$ is small for any
two messages $m, m' \in \mc{M}$. This is formalized in
\cref{df:te-security}. Recall that $
	\left(
		\left(\Tr_\tsf{M} \circ \overline{D}_k\right)
		\tensor \Id_\tsf{A}
	\right)
	\circ
	A
	\circ
	E_k
	:
	\mc{L}(\tsf{M})
	\to
	\mc{L}(\tsf{A})
$ is a linear operator. Hence, the difference between the image of $m$
and $m'$ under the action of this map can be more succinctly expressed
as $
	\left(
		\left(\Tr_\tsf{M} \circ \overline{D}_k\right)
		\tensor \Id_\tsf{A}
	\right)
	\circ
	A
	\circ
	E_k
	(m - m')
$. We will often make use of this property to make our notation more
compact.

\begin{definition}
\label{df:te-security}
	Let $\delta \in \R$.
	An AQECM scheme $(K, E, D)$ as given in \cref{df:aqecm}
	is~\emph{$\delta$-tamper-evident} if for all messages
	$m, m' \in \mc{M}$ and any tamper attack $A$ against it we have
	that
	\begin{equation}
	\label{eq:te-security}
		\Pr_{k \gets K(\emptystring)}\left[
			\frac{1}{2}
			\norm{
				\left(
					\left(\Tr_\tsf{M} \circ \overline{D}_k\right)
					\tensor \Id_\tsf{A}
				\right)
				\circ
				A
				\circ
				E_k
				(m - m')
			}_1
			\leq
			\delta
		\right]
		\geq
		1 - \delta
		.
	\end{equation}
\end{definition}

We know that we can interpret the trace-distance between two density
operators as a measure of how well two quantum states can be
distinguished. But, how should we interpret the trace distance between
two \emph{subnormalized} density operators in this context? The
following lemma can be understood as stating that if the trace distance
between two subnormalized density operators is small, then
either the trace distance between their corresponding renormalized
operators is small, or both of their traces are small.\footnote{%
	This point was also made, but not quantified as we did here, by
	Gottesman \cite{Got03}.}

\begin{lemma}
	Let $\rho_0, \rho_1 \in \mc{D}(\tsf{A})$ be two
	density operators and let $t_0, t_1 \in \R$ be two non-negative
	real numbers.
	Then,
	\begin{equation}
		\frac{\max\{t_0, t_1\}}{2}
		\cdot
		\frac{1}{2}\norm{\rho_0 - \rho_1}_1
		\leq
		\frac{1}{2}\norm{t_0 \rho_0 - t_1 \rho_1}_1.
	\end{equation}
\end{lemma}
\begin{proof}
	First, note that $
		\frac{1}{2}\abs{t_0 - t_1}
		\leq
		\frac{1}{2}\norm{t_0 \rho_0 - t_1 \rho_1}_1
	$ by the fact that the trace distance is contractive under the
	action of any channel, including, as needed for this case, the
	trace.
	Without loss of generality, assume that $t_0 = \max\{t_0, t_1\}$.
	Then,
	\begin{equation}
	\begin{aligned}
		\frac{t_0}{2}
		\norm{\rho_0 - \rho_1}_1
		&=
		\frac{1}{2}\norm{t_0\rho_0 - t_1\rho_1 + (t_1 - t_0)\rho_1}_1
		\\&\leq
		\frac{1}{2}\norm{t_0\rho_0 - t_1\rho_1}_1
		+
		\frac{1}{2} \cdot \abs{t_0 - t_1}
	\end{aligned}
	\end{equation}
	which, with the previous inequality, yields that $
		\frac{t_0}{2}\norm{\rho_0 - \rho_1}_1
		\leq
		\norm{t_0\rho_0 - t_1\rho_1}_1
	$.
	Dividing both sides by two then yields the desired result.
\end{proof}

In the context of tamper-evident encryption, this means that if
$\frac{1}{2}\norm{\rho_{k,m} - \rho_{k,m'}}_1$ is small, then either the
eavesdropping attempt will be detected with high probability (in the
case that the traces of both $\rho_{k,m}$ and $\rho_{k,m'}$ are small),
or it will yield very little information on the plaintext if it is
undetected (in the case where the trace distance between the
renormalized states $\tilde{\rho}_{k,m}=\Tr(\rho_{k,m})^{-1}\rho_{k,m}$
and $\tilde{\rho}_{k,m'} = \Tr(\rho_{k,m'})^{-1}\rho_{k,m'}$ is small, assuming
for simplicity here that neither has trace $0$). Indeed, in this second
case, the adversary will not be able to distinguish between the
scenario where the message $m$ is sent or the one where the message $m'$
is sent, \emph{even if they later learn the key} $k$ \emph{which was
used}.

\Cref{df:te-security} has two notable differences with the definition of
tamper evidence originally given by Gottesman.
First, we do not require the AQECM scheme to satisfy any type of
encryption security guarantee to be a tamper-evident scheme.
We consider tamper evidence and encryption completely separately for the
time being, but we will show later in \cref{th:te=>enc} that any
sufficiently good tamper-evident scheme will indeed be an encryption
scheme.
Second, Gottesman does not make an explicit statement on the probability
of the trace distance being small. Instead, it is required that the
trace distance considered in \cref{eq:te-security}
be at most $\delta$ for at least $(1-\delta)\cdot\abs{\mc{K}}$ keys in
$\mc{K}$.
Of course, these notions coincide if the keys are sampled uniformly
at random by the key generation procedure, but our formulation of this
definition is a bit more general.

Before proceeding further, we make explicit two simple but useful facts
which immediately follow from \cref{df:te-security}.

\begin{fact}
	A $\delta$-tamper-evident AQECM scheme is also $\delta'$-tamper
	evident if $\delta' > \delta$.
\end{fact}

\begin{fact}
\label{ft:te-bounds}
	Every AQECM scheme is $1$-tamper-evident, none is
	$\delta$-tamper-evident if $\delta < 0$.
\end{fact}

The definition of the security notion of tamper evidence is a bit
unorthodox because it is not expressed as an expectation over the choice
of the key.
The following lemma expresses sufficient and necessary conditions
for an AQECM scheme to be $\delta$-tamper-evident as bounds concerning
the expected trace distance

\begin{lemma}
\label{tee:th:te-ns}
	Let $S = (K, E, D)$ be an AQECM as given in \cref{df:aqecm}.
	\begin{enumerate}
		\item
			If for all tamper attacks $
				A :
				\mc{L}(\tsf{C})
				\to
				\mc{L}(\tsf{C} \tensor \tsf{A})
			$
			and all messages $m,m' \in \mc{M}$ we have that
			\begin{equation}
				\E_{k \gets K(\emptystring)}
				\frac{1}{2}
				\norm{
					\left(
						\left(\Tr_\tsf{M} \circ \overline{D}_k\right)
						\circ
						\Id_{\tsf{A}}
					\right)
					\circ
					A
					\circ
					E_k(m - m')
				}_1
				\leq
				\delta
			\end{equation}
			then $S$ is $\sqrt{\delta}$-tamper-evident.
		\item
			If $S$ is $\delta$-tamper-evident, then for all tamper
			attacks $
				A :
				\mc{L}(\tsf{C})
				\to
				\mc{L}(\tsf{C} \tensor \tsf{A})
			$
			and all messages $m,m' \in \mc{M}$ we have that
			\begin{equation}
				\E_{k \gets K(\emptystring)}
				\frac{1}{2}
				\norm{
					\left(
						\left(\Tr_\tsf{M} \circ \overline{D}_k\right)
						\circ
						\Id_{\tsf{A}}
					\right)
					\circ
					A
					\circ
					E_k(m - m')
				}_1
				\leq
				2\delta - \delta^2
				\leq
				2\delta.
			\end{equation}
	\end{enumerate}
\end{lemma}
\begin{proof}
	We begin by showing the first point.
	By Markov's inequality (\cref{th:markov}) and our assumed upper
	bound on the expectation we have
	\begin{equation}
		\left(1
		-
		\Pr_{k \gets K(\emptystring)}\left[
			\frac{1}{2}
			\norm{
				\left(
				\left(\Tr_\tsf{M} \circ \overline{D}_k\right)
					\circ
					\Id_{\tsf{A}}
				\right)
				\circ
				A
				\circ
				E_k(m - m')
			}_1
			<
			\sqrt{\delta}
		\right]
		\right)
		\leq
		\frac{\delta}{\sqrt{\delta}}
		=
		\sqrt{\delta}
	\end{equation}
	which implies that
	\begin{equation}
		\Pr_{k \gets K(\emptystring)}\left[
			\frac{1}{2}
			\norm{
				\left(
				\left(\Tr_\tsf{M} \circ \overline{D}_k\right)
					\circ
					\Id_{\tsf{A}}
				\right)
				\circ
				A
				\circ
				E_k(m - m')
			}_1
			<
			\sqrt{\delta}
		\right]
		\geq
		1 - \sqrt{\delta}
	\end{equation}
	and so $S$ is $\sqrt{\delta}$-tamper-evident.

	We now show the second point.
	Note that if $\delta \geq \frac{1}{2}$, then the desired inequality
	on the expectation immediately holds as the trace distance between
	two subnormalized states can never be more than $1$.
	If $\delta < \frac{1}{2}$, then it suffices to invoke the
	concentration inequality of \cref{th:concentration}, with an upper
	bound of $\beta = 1$.
	This yields, after a few algebraic manipulations and recalling our
	assumption that $S$ is $\delta$-tamper-evident, that
	\begin{equation}
				\E_{k \gets K(\emptystring)}
				\frac{1}{2}
				\norm{
					\left(
						\left(\Tr_\tsf{M} \circ \overline{D}_k\right)
						\circ
						\Id_{\tsf{A}}
					\right)
					\circ
					A
					\circ
					E_k(m - m')
				}_1
				\leq
				\delta + (1  - \delta)\cdot \delta
				=
				2\delta - \delta^2
				\leq
				2\delta
	\end{equation}
	which is the desired result.
\end{proof}

We conclude this section by reiterating what is, morally, Gottesman's
main theorem, namely that the theory of tamper-evident schemes is not
vacuous. Tamper-evident schemes can be explicitly constructed.
Better yet, for any message set $\mc{M}$ and $\delta > 0$, there exists
a perfectly correct AQECM scheme which is $\delta$-tamper-evident.

\begin{theorem}[{\cite{Got03}}]
\label{te:th:te-exist}
Let $\mc{M}$ be an alphabet and $\delta \in \R^+$ be a strictly positive
real. There exists a perfectly correct $\delta$-tamper-evident AQECM
scheme with messages $\mc{M}$.
\end{theorem}

Gottesman explicitly demonstrates this for cases where $\mc{M}=\bs{}^n$
for any $n \in \N^+$ by constructing
such schemes. This result can be
extended to arbitrary alphabets $\mc{M}'$ by adding an injective
map $i : \mc{M} \to \bs{}^{\ceil{\log_2\abs{\mc{M}}}}$ and any map
$j : \bs{}^{\ceil{\log_2\abs{\mc{M}}}} \to \mc{M}$ such that $j
\circ i$ is the identity on $\mc{M}$. Messages $m \in \mc{M}$ are then
pre-processed with $i$ before being encoded and post-processed by
$j$ after being decoded. It is trivial to see that this preserves
correctness and tamper evidence.

%======================================================================%
\section{Tamper Evidence Implies Encryption}
\label{sc:te=>enc}
%======================================================================%

In this section, we show that any sufficiently good tamper-evident
scheme is a good encryption scheme.
This answers, in the positive, one of Gottesman's open questions
\cite{Got03}.

We begin by giving a simple definition of encryption for QECMs and
AQECMs, which states that the ciphertexts corresponding to two different
messages $m_0$ and $m_1$ should be close in trace distance when
averaging over all possible keys. Our definition is essentially the same
as the one given by Ambainis, Mosca, Tapp, and de Wolf \cite{AMTW00},
restricted to only considering the computational basis states and with a
relaxation allowing the resulting ciphertexts to differ slightly in
trace distance. This is captured in the following definition.

\begin{definition}
\label{tee:df:enc}
	Let $\delta \in \R$.
	A QECM or AQECM scheme $(K, E, D)$ is $\delta$-encrypting if, for
	all messages $m_0, m_1 \in \mc{M}$, we have that
	\begin{equation}
		\frac{1}{2}
		\norm{
			\E_{k \gets K(\emptystring)} E_k\left(m_0-m_1\right)
		}_1
		\leq
		\delta.
	\end{equation}
\end{definition}
Note that linearity yields
$
	\E_{k \gets K(\emptystring)}E_k(m_0 - m_1)
	=
	\E_{k \gets K(\emptystring)}E_k(m_0)
	-
	\E_{k \gets K(\emptystring)}E_k(m_1)
$, where the right-hand side is a bit less succinct but a bit more
descriptive than what is used in \cref{tee:df:enc}.

The main theorem we prove in this section, \cref{th:te=>enc}, is that
any $\epsilon$-correct $\delta$-tamper-evident AQECM scheme is
also $\sqrt{19(\delta + \sqrt{2} \epsilon)}$-encrypting.
As discussed in \cref{te:sc:contributions-enc}, our proof follows the
same ideas as the proof given by Barnum, Cr\'epeau, Gottesman, Smith,
and Tapp showing that quantum authentication schemes are also quantum
encryption schemes \cite{BCG+02} and proceeds in three steps:
\begin{enumerate}
	\item
		We first show in \cref{th:benc=>bte} of \cref{sc:benc=>bte} that
		any AQECM scheme which does not offer much security as an
		encryption scheme cannot be very tamper-evident.
	\item
		Next, we show in \cref{th:te-parallel} of
		\cref{sc:aqecm-repetition} that the parallel applications of a
		$\delta$-tamper-evident scheme and a $\delta'$-tamper-evident
		scheme yields a $(\delta + \delta')$-tamper-evident scheme.
		In other words, the parallel applications of AQECM schemes
		yields a \emph{linear} loss of tamper-evident security.
	\item
		Contrasting the above point with an observation from Barnum
		\etal~that the parallel applications of (A)QECM
		schemes yields an \emph{exponential} loss of encryption
		security due to \cref{th:trace-distance-copies}, we can boost
		the weak bound of \cref{th:benc=>bte}.
		This yields our main result of this section, \cref{th:te=>enc}
		of \cref{sc:gte=>genc}.
\end{enumerate}

%~~~~~~~~~~~~~~~~~~~~~~~~~~~~~~~~~~~~~~~~~~~~~~~~~~~~~~~~~~~~~~~~~~~~~~%
\subsection{Bad Encryption Implies Bad Tamper Evidence}
\label{sc:benc=>bte}
%~~~~~~~~~~~~~~~~~~~~~~~~~~~~~~~~~~~~~~~~~~~~~~~~~~~~~~~~~~~~~~~~~~~~~~%

We show in this subsection that an AQECM scheme which is not a good
encryption scheme cannot be a good tamper-evident scheme.

\begin{theorem}
\label{th:benc=>bte}
	Let $S = (K, E, D)$ be an $\epsilon$-correct AQECM scheme.
	If $S$ is \emph{not} $\delta$-encrypting for $\delta \geq 0$, then
	it is \emph{not}
	$(1 - \sqrt[3]{4(1-\delta)} - \sqrt{2\epsilon})$-tamper-evident.
\end{theorem}
\begin{proof}
	We begin by identifying pairs of values $(\epsilon, \delta)$ for
	which the claim trivially holds.
	We can assume that $\delta < 1$ as there is no scheme which is
	not $\delta$-encrypting for $\delta \geq 1$. We can further assume
	that $\epsilon \geq 0$ as there does not exist an AQECM
	scheme which is $\epsilon$-correct for $\epsilon < 0$.
	If $1 - \sqrt[3]{4(1-\delta)} - \sqrt{2\epsilon} < 0$,
	then the result trivially holds as no AQECM scheme is
	$\tilde{\delta}$-tamper-evident if $\tilde{\delta} < 0$.
	Thus, we can further assume that
	$1 - \sqrt[3]{4(1-\delta)} - \sqrt{2\epsilon} \geq 0$ for the
	remainder of the proof.

	Next, we identify two messages and a tamper attack which we will use
	to demonstrate that $S$ is not sufficiently tamper-evident.
	As $S$ is not $\delta$-encrypting by assumption, there are two
	messages $m_\bz ,m_\bo \in \mc{M}$ such that
	\begin{equation}
		\frac{1}{2}\norm{
			\E_{k \gets K(\emptystring)} E_k(m_\bz)
			-
			\E_{k \gets K(\emptystring)} E_k(m_\bo)
		}_1
		>
		\delta.
	\end{equation}
	By \cref{th:trace-channel}, this implies the existence of a
	distinguishing channel $
		\Phi : \mc{L}(\tsf{C}) \to \mc{L}(\C^{\bs{}})
	$ satisfying
	\begin{equation}
		\bra{\bz}
			\Phi\left(\E_{k \gets K(\emptystring)} E_k(m_\bz)\right)
		\ket{\bz}
		+
		\bra{\bo}
			\Phi\left(\E_{k \gets K(\emptystring)} E_k(m_\bo)\right)
		\ket{\bo}
		>
		1 + \delta
	\end{equation}
	which, by linearity, also satisfies
	\begin{equation}
	\label{eq:benc=>bte-1}
		\E_{k \gets K(\emptystring)}
		\left(
		\bra{\bz}
			\Phi\left(E_k(m_\bz)\right)
		\ket{\bz}
		+
		\bra{\bo}
			\Phi\left(E_k(m_\bo)\right)
		\ket{\bo}
		\right)
		>
		1 + \delta.
	\end{equation}
	Let $\tsf{A} = \C^{\bs{}}$ be the state space of a single qubit and
	$
		A = \CGM(\Phi):
		\mc{L}(\tsf{C})\to\mc{L}(\tsf{C}\tensor\tsf{A})
	$ be the
	coherent gentle measurement channel obtained from the distinguishing
	channel $\Phi$ as described in \cref{df:cgm}. Note that $A$ is a
	tamper attack against the scheme $S$ and that by \cref{th:cgm} we
	have that
	\begin{equation}
	\label{te:eq:te-benc-cgm}
		\frac{1}{2}\norm{
			A(\rho)
			-
			\rho \tensor \ketbra{b}
			}_1
		\leq
		\sqrt{
			\Tr\left(\rho\right)
			-
			\left(\bra{b} \Phi(\rho)\ket{b}\right)^2
		}
	\end{equation}
	for each $b \in \bs{}$ and all $\rho \in \mc{D}_\bullet(\tsf{C})$.

	Before proceeding, we recall that tamper evidence is defined with
	respect to a probability over keys, and not an expectation.
	Thus, we use the concentration inequality of
	\cref{th:concentration} to translate \cref{eq:benc=>bte-1} to this
	setting and obtain that
	\begin{equation}
	\begin{aligned}
		\Pr_{k \gets K(\emptystring)}
		\left[
			\bra{\bz}
				\Phi\left(E_k(m_0)\right)
			\ket{\bz}
			+
			\bra{\bo}
				\Phi\left(E_k(m_1)\right)
			\ket{\bo}
			>
			1 + \delta'
		\right]
		&>
		\frac{(1 + \delta) - (1 + \delta')}{2 - (1 + \delta')}
		\\&=
		1
		-
		\frac{1-\delta}{1-\delta'}
	\end{aligned}
	\end{equation}
	for a yet-to-be-determined value of $\delta'$ satisfying
	$0 \leq \delta' + 1 < 2$, as required by \cref{th:concentration}.
	Similarly, we use the same concentration inequality and the assumed
	$\epsilon$-correctness of the AQECM scheme to obtain
	\begin{equation}
		\Pr_{k \gets K(\emptystring)}\left[
			\Tr\left[\overline{D}_k \circ E_k(m_b)\right]
			\geq
			1
			-
			\sqrt{2\epsilon}
		\right]
		\geq
		1 - \sqrt{\frac{\epsilon}{2}}
	\end{equation}
	for each $b \in \bs{}$. Note that we are only checking here that
	$D_k$ accepted, not that it also obtained the correct message.
	Combining these previous inequalities via the union bound, we obtain
	that
	\begin{equation}
		\Pr_{k \gets K(\emptystring)}\left[
			\begin{array}{c}
				\bra{\bz} \Phi \circ E_k(m_\bz) \ket{\bz}
				+
				\bra{\bo} \Phi \circ E_k(m_\bo) \ket{\bo}
				> 1 + \delta'
				\\
				\land
				\\
				\Tr\left[\overline{D}_k \circ E_k(m_\bz)\right]
				\geq
				1 - \sqrt{2\epsilon}
				\\
				\land
				\\
				\Tr\left[\overline{D}_k \circ E_k(m_\bo)\right]
				\geq
				1 - \sqrt{2\epsilon}
			\end{array}
		\right]
		>
		1
		-
		\frac{1-\delta}{1-\delta'}
		-
		2\sqrt{\frac{\epsilon}{2}}.
	\end{equation}
	Let $\mc{K}_\text{Good} \subseteq \mc{K}$ be the set of keys
	satisfying the event in the previous inequality.
	In other words, these are the keys which satisfy two conditions: the
	encodings of $m_\bz$ and $m_\bo$ under these keys are distinguished
	sufficiently well by $\Phi$ and, in the absence of any adversarial
	attack, the encodings of $m_\bz$ and $m_\bo$ under these keys will
	be accepted with sufficiently high probability by the decryption
	map. Our goal now is to show that the tamper attack $A$ succeeds
	well enough when a key in $\mc{K}_\text{Good}$ is used. We note
	that, by definition,
	\begin{equation}
		\Pr_{k \gets K(\varepsilon)}\left[k\in\mc{K}_\text{Good}\right]
		>
		1 - \frac{1-\delta}{1-\delta'} - \sqrt{2\epsilon}
	\end{equation}
	and so selecting $\delta'$ in such a way that the right-hand side of
	this inequality is at least $0$ is sufficient to ensure that
	$\mc{K}_\text{Good}$ is non-empty.

	Now, for each $b \in \bs{}$ and key $k \in \mc{K}_\text{Good}$,
	define the operators
	\begin{equation}
		\rho_{k,b}
		=
		\left(
			\left(\Tr_\tsf{M} \circ \overline{D}_k\right)
			\tensor
			\Id_\tsf{A}
		\right)
		\circ
		A
		\circ
		E_k(m_b)
		\qq{and}
		\sigma_{k,b}
		=
		\Tr\left[\overline{D}_k \circ E_k(m_b)\right] \cdot \ketbra{b}.
	\end{equation}
	All of these operators are subnormalized density operators on a
	single qubit, \ie~elements of $\mc{D}_\bullet(\C^{\bs{}})$,  and so
	represent a possible state for the system held by an eavesdropping
	adversary who used the tamper attack $A$ against
	the scheme $S$, scaled by the probability that this attack was
	undetected.
	The $\rho$ operators correspond to the ``real''
	scenario where the tamper attack $A$ is actually applied.
	The $\sigma$ operators, on the other hand, correspond to an
	``ideal'' scenario where $A$ is not actually applied, but where the
	adversary is simply given directly the bit corresponding to the
	message which was encoded.

	We pause a moment to emphasize that if we can show that for
	every $k \in \mc{K}_\text{Good}$ we have that
	\begin{equation}
		\frac{1}{2}\norm{\rho_{k,\bz} - \rho_{k,\bo}}
		>
		1
		-
		\frac{1 - \delta}{1 - \delta'} - \sqrt{2\epsilon}
	\end{equation}
	then we will have shown that $(K, E, D)$ is not $
		\left(1 - \frac{1-\delta}{1-\delta'} - \sqrt{2\epsilon}\right)
	$-tamper-evident.
	This is now our goal.

	For any key $k \in \mc{K}_\text{Good}$, we obtain via the triangle
	inequality that
	\begin{equation}
		\norm{\sigma_{k,\bz} - \sigma_{k,\bo}}_1
		\leq
		\norm{\sigma_{k,\bz} - \rho_{k,\bz}}_1
		+
		\norm{\rho_{k,\bz} - \rho_{k,\bo}}_1
		+
		\norm{\rho_{k,\bo} - \sigma_{k,\bo}}_1,
	\end{equation}
	which in turn implies that
	\begin{equation}
		\norm{\rho_{k,\bz} - \rho_{k,\bo}}_1
		\geq
		\norm{\sigma_{k,\bz} - \sigma_{k,\bo}}_1
		-
		\norm{\rho_{k,\bz} - \sigma_{k,\bz}}_1
		-
		\norm{\rho_{k,\bo} - \sigma_{k,\bo}}_1.
	\end{equation}
	Thus, a lower-bound on $\norm{\rho_{k,\bz} - \rho_{k,\bo}}_1$ can be
	found by finding a lower bound on
	$\norm{\sigma_{k,\bz} - \sigma_{k,\bo}}_1$ and an upper bound on
	$\sum_{b \in \bs{}} \norm{\rho_{k,b} - \sigma_{k,b}}_1$.

	By direct computation, we find that
	\begin{equation}
		\norm{
			\sigma_{k,0} - \sigma_{k,1}
		}_1
		=
		\Tr\left[\overline{D}_k \circ E_k(m_\bz)\right]
		+
		\Tr\left[\overline{D}_k \circ E_k(m_\bo)\right]
		\geq
		2 - 2\sqrt{2\epsilon}
	\end{equation}
	where the inequality follows from the assumption that $k \in
	\mc{K}_\text{Good}$.
	Next, by \cref{te:eq:te-benc-cgm}, we have that
	\begin{equation}
	\label{tee:eq:opt}
		\sum_{b \in \bs{}}
		\frac{1}{2}
		\norm{A(E_k(m_b)) - E_k(m_b) \tensor \ketbra{b}}_1
		\leq
		\sum_{b \in \bs{}}
		\sqrt{1 - \left(\bra{b} \Phi(E_k(m_b))
		\ket{b}\right)^2}.
	\end{equation}
	While we do not have much knowledge on the individual values of
	$\sqrt{1 - \left(\bra{b} \Phi(E_k(m_b))\ket{b}\right)^2}$ for each
	value of $b \in \bs{}$, we do know that
	$\sum_{b \in \{0,1\}}\bra{b}\Phi(E_k(m_b))\ket{b} > 1 + \delta'$ by
	virtue of $k \in \mc{K}_\text{Good}$.
	Solving the optimization problem of maximizing the right-hand
	side of \cref{tee:eq:opt} under this constraint, we find that
	\begin{equation}
	\begin{aligned}
		\sum_{b \in \bs{}}
		\frac{1}{2}
		\norm{A(E_k(m_b)) - E_k(m_b) \tensor \ketbra{b}}_1
		&<
		2\sqrt{1 - \left(\frac{1+\delta'}{2}\right)^2}
		\\&=
		\sqrt{4 - (1 + \delta')^2}
		\\&
		<
		2\sqrt{1-\delta'}
	\end{aligned}
	\end{equation}
	where the last inequality follows by algebraic manipulation and
	dropping a few terms,\footnote{%
		Specifically, $
			\sqrt{4 - (1 + \delta')^2}
			=
			\sqrt{4 - (2 + (\delta' -1))^2}
			=
			\sqrt{4(1-\delta') - (\delta'-1)^2}
			<
			\sqrt{4(1-\delta')}$.}
	which we do to simplify later arguments.
	Now, as the trace distance is contractive under the map
	$\left(\Tr_\tsf{M} \circ \overline{D}_k\right) \tensor \Id_\tsf{A}$
	and  applying this map to the states in the above trace distance
	will allows us to recover our $\rho$ and $\sigma$ operators, we have
	that
	\begin{equation}
		\sum_{b \in \bs{}}
		\frac{1}{2}
		\norm{
			\rho_{k,b} - \sigma_{k,b}
		}_1
		\leq
		2\sqrt{1-\delta'}.
	\end{equation}
	Combining all the above, we obtain that
	\begin{equation}
		\frac{1}{2}
		\norm{
			\rho_{k,\bz} - \rho_{k,\bo}
		}_1
		> 1 - 2\sqrt{1-\delta'} - \sqrt{2\epsilon}.
	\end{equation}
	for all keys $k \in \mc{K}_\text{Good}$.

	It now suffices to fix the value of $\delta'$.
	We choose to take
	$\delta' = 1 - \left(\frac{1-\delta}{2}\right)^{2/3}$ as this
	implies that $
		2\sqrt{1-\delta'}
		=
		\frac{1-\delta}{1-\delta'}
		=
		\sqrt[3]{4(1-\delta)}
	$.
	In particular, we see that $\delta < 1$ implies that $\delta' < 1$,
	as required, and that $\mc{K}_\text{Good} \not= \emptyset$ as $
		\Pr_{k \gets K(\emptystring)}[k \in \mc{K}_\text{Good}]
		>
		1 - \sqrt[3]{4(1-\delta)} - \sqrt{2\epsilon}
		\geq
		0
	$.
	Thus,
	\begin{equation}
		\Pr_{k \gets K(\emptystring)}
			\left[
				\frac{1}{2}
				\norm{
					\rho_{k,0} - \rho_{k,1}
				}_1
				>
				1-\sqrt[3]{4(1-\delta)} - \sqrt{2\epsilon}
			\right]
			>
			1-\sqrt[3]{4(1-\delta)} - \sqrt{2\epsilon}
	\end{equation}
	from which we conclude that the scheme $S$ is not
	$(1-\sqrt[3]{4(1-\delta)} - \sqrt{2\epsilon})$-tamper-evident.
\end{proof}

%----------------------------------------------------------------------%
\subsection{On The Parallel Composition of AQECM Schemes}              %
\label{sc:aqecm-repetition}                                            %
%----------------------------------------------------------------------%

We study here the parallel composition of
AQECM schemes and how it affects the correctness and
tamper evidence of such schemes.

The parallel composition of two AQECM schemes $S'$ and $S''$ is a
scheme $S = S' \tensor S''$ which encodes pairs of messages $(m', m'')$
with pairs of independent keys $(k', k'')$ to obtain pairs of
ciphertexts $(\rho', \rho'')$. The first component of
these pairs is treated with the scheme $S'$ and the second with scheme
$S''$. Note that, formally, we will model these pairs as tensor
products, such as $\rho'\tensor\rho''$.
Ciphertexts are also decoded independently, \emph{except} for the
production of the single flag qubit necessary for a well defined AQECM.
The flag qubit produced by the overall scheme $S$ will be the result of
the logical-and of the flag qubits independently produced by the
component schemes $S'$ and $S''$. In other words, $S$ accepts if and
only if both of the schemes $S'$ and $S''$ accept their respective
ciphertexts.

This construction is formalized in \cref{df:aqecm-parallel}. Note that
$\Swap$ channels are included to ensure that keys, plaintexts, and
ciphertexts are correctly ordered and distributed to the component
schemes. The $V$ channel, for its part, takes the two flag qubits
produced by the component schemes and computes their logical-and.

\begin{definition}
\label{df:aqecm-parallel}
	Let $S' = (K', E', D')$ and $S'' = (K'', E'', D'')$ be two
	AQECM schemes as given in \cref{df:aqecm} where the notation for all
	associated alphabets and Hilbert spaces carry an additional prime
	$'$ or double prime $''$ to clarify with which scheme it is
	associated.

	We define the parallel composition of these schemes, denoted
	$S' \tensor S''$, as the triplet of channels $(K, E, D)$ defined by
	\begin{align}
		K
		&=
		K' \tensor K''
		\\
		E
		&=
		\left(E' \tensor E''\right)
		\circ
		\Swap
		_{\tsf{K}',  \tsf{K}'', \tsf{M}', \tsf{M}''}
		^{(1, 3, 2, 4)}
		\\
		D
		&=
		V
		\circ
		\left(D' \tensor D''\right)
		\circ
		\Swap
		_{\tsf{K}', \tsf{K}'', \tsf{C}', \tsf{C}''}
		^{(1, 3, 2, 4)}
	\end{align}
	where $V$ is the channel defined by
	\begin{equation}
		\rho
		\mapsto
		\sum_{b', b'' \in \{\bz,\bo\}}
		\left(
			I_{\tsf{M}'}
			\tensor
			\bra{b'}
			\tensor
			I_{\tsf{M}''}
			\tensor
			\bra{b''}
		\right)
		\rho
		\left(
			I_{\tsf{M}'}
			\tensor
			\ket{b'}
			\tensor
			I_{\tsf{M}''}
			\tensor
			\ket{b''}
		\right)
			\tensor
			(b' \land b'')
	\end{equation}
	where $b' \land b'' \in \bs{}$ and $b' \land b'' = \bo \iff b' = b''
	= \bo$.
	The channel $V$ can be understood as measuring and discarding the
	flag qubits produced by $D'$ and $D''$ in the computational basis
	and outputing a new flag qubit whose state is the logical-and
	of these measurement outcomes.

	Note that this yields an AQECM whose keys are the elements of
	$\mc{K} = \mc{K}' \times \mc{K}''$ and whose messages are the
	elements of $\mc{M} = \mc{M}' \times \mc{M}''$.
	Further note that, for any key $(k',k'') \in \mc{K}'\times\mc{K}''$,
	we have that
	\begin{equation}
		E_{(k',k'')} = E'_{k'} \tensor E''_{k''},
		\qq*{}
		D_{(k',k'')} = V \circ (D'_{k'} \tensor D''_{k''}),
		\qq{and}
		\overline{D}_{(k',k'')} = \overline{D'}_{k'} \tensor
		\overline{D''}_{k''}.
	\end{equation}
\end{definition}

We first show how correctness is preseved under this construction.

\begin{theorem}
\label{th:aqecm-parallel}
	Let $S'$ and $S''$ be two AQECM schemes, as given in
	\cref{df:aqecm-parallel}, which are  $\epsilon'$- and
	$\epsilon''$-correct, respectively.
	Then, $S = S' \tensor S''$ is $(\epsilon' + \epsilon'')$-correct.
\end{theorem}
\begin{proof}
	Let $m = (m',m'') \in \mc{M}' \times \mc{M}''$ be a message for $S$.
	We have that
	\begin{align}
		&
		\E_{k \gets K(\emptystring)}
		\bra{m}
			\overline{D}_{k}
			\circ
			E_{k}
			(m)
		\ket{m}
		\\&=
		\E_{\substack{k' \gets K'(\emptystring) \\ k'' \gets
		K''(\emptystring)}}
		\bra{m', m''}
			\overline{D'}_{k'}
			\circ
			E'_{k'}
			(m')
			\tensor
			\overline{D''}_{k''}
			\circ
			E''_{k''}
			(m'')
		\ket{m', m''}
		\\&=
		\E_{k' \gets K'(\emptystring)}
		\bra{m'}
			\overline{D'}_{k'}
			\circ
			E'_{k'}
			(m')
		\ket{m'}
		\E_{k'' \gets K''(\emptystring)}
		\bra{m^1}
			\overline{D''}_{k''}
			\circ
			E_{k''}
			(m'')
		\ket{m''}
		\\&\geq
		(1 - \epsilon') \cdot (1 - \epsilon'')
		\\&\geq
		1 - (\epsilon' + \epsilon'')
	\end{align}
	which is the desired result.
\end{proof}

We now show how tamper evidence is preserved under parallel composition.

\begin{theorem}
\label{th:te-parallel}
	Let $S'$ and $S''$ be two AQECM schemes, as given in
	\cref{df:aqecm-parallel}, which are $\delta'$- and
	$\delta''$-tamper-evident, respectively.
	Then, $S = S' \tensor S''$ is
	$(\delta' + \delta'')$-tamper-evident.
\end{theorem}
\begin{proof}
	Fix a tamper attack $
		A :
		\mc{L}(\tsf{C})
		\to
		\mc{L}(\tsf{C} \tensor \tsf{A})
	$ against the scheme $S$.
	Now, for any key $k = (k', k'') \in \mc{K}'\times\mc{K}''$
	and any message $m = (m', m'') \in \mc{M}'\times\mc{M}''$,
	let $\rho_{k,m} \in \mc{D}_\bullet(\tsf{A})$ be the
	subnormalized state defined by
	\begin{equation}
		\rho_{k,m}
		=
		\left(\Tr_\tsf{M} \tensor \Id_\tsf{A}\right)
		\circ
		\left(\overline{D}_{k} \tensor \Id_\tsf{A}\right)
		\circ
		A
		\circ
		E_k(m).
	\end{equation}
	For the remainder of the proof, fix two messages $m_0=(m'_0,m''_0)$
	and $m_1 = (m'_1, m''_1)$ in $\mc{M}$.
	It suffices to show that
	\begin{equation}
		\Pr_{k \gets K(\emptystring)}
		\left[
			\frac{1}{2}
			\norm{
				\rho_{k,m_0} - \rho_{k,m_1}
			}_1
			\leq \delta_0 + \delta_1
		\right]
		\geq
		1 - (\delta_0 + \delta_1).
	\end{equation}
	By the triangle inequality, we have for any key $k \in \mc{K}$ that
	\begin{equation}
		\frac{1}{2}
		\norm{
			\rho_{k,(m'_0,m''_0)}
			-
			\rho_{k,(m'_1,m''_1)}
		}_1
		\leq
		\underbrace{
		\frac{1}{2}
		\norm{
			\rho_{k,(m'_0,m''_0)}
			-
			\rho_{k,(m'_1,m''_0)}
		}_1
		}_{= N'_k}
		+
		\underbrace{
		\frac{1}{2}
		\norm{
			\rho_{k,(m'_1,m''_0)}
			-
			\rho_{k,(m'_1,m''_1)}
		}_1
		}_{= N''_k}
	\end{equation}
	where we note the presence of the subnormalized state
	$\rho_{k,(m'_1,m''_0)}$, corresponding to the ``hybrid'' message
	$(m'_1,m''_0)$, in both norms in the right-hand side of the above.
	It now suffices to show that $N'_k$ is at most $\delta'$ with
	probability at least $1 - \delta'$ over the sampling of $k$ and
	that $N''_k$ is at most $\delta''$ with probability at least
	$1 - \delta''$, over the sampling of $k$.
	Indeed, by the union bound, this would then imply that the
	probability over the sampling of the key $k$ that both
	$N'_k \leq \delta'$ and $N''_k \leq \delta''$ is at least
	$1 - (\delta' + \delta'')$, which yields the desired result.

	We show that $
		\Pr_{k \gets K(\emptystring)}\left[N'_k \leq \delta'\right]
		\geq 1 - \delta'
	$ by constructing a certain
	family of tamper attacks $\{A'_{k''}\}_{k'' \in \mc{K}''}$ against
	the scheme $S'$.
	The corresponding inequality for $N''_k$ can be obtained completely
	analogously and will not be explicitly proven here.

	For all keys $k'' \in \mc{K}''$ we define the tamper attack $
		A'_{k''}
		:
		\mc{L}(\tsf{C}')
		\to
		\mc{L}(\tsf{C}' \tensor \tsf{A}')
	$ against the scheme $S'$, where
	$\tsf{A}'=\C^{\bs{}}\tensor\tsf{A}$, by
	\begin{equation}
		\rho
		\mapsto
		\left(
		\left(
			\Id_{\tsf{C}'}
			\tensor
			\left(\Tr_{\tsf{M}''} \tensor \Id_{\C^{\bs{}}}\right)
			\circ D''_{k''}\right)
		\right)
		\circ
		A
		\circ
		\left(
			\rho
			\tensor
			E_{k''}^1(m''_0)
		\right).
	\end{equation}
	In other words, $A'_{k''}$ attacks $S'$ by first generating the
	state $E''_{k''}(m''_0)$, where $k''$ parametrizes the attack and
	$m''_0$ is determined by the fixed message $m_0$ we are considering
	in this proof, and then running
	the attack $A$ against the ciphertext obtained by combining this
	state with the one intercepted from the honest users.
	The attack $A'_{k''}$ then applies the decoding channel $D''_{k''}$
	on the $\tsf{C}''$ portion of the ciphertext outputted by $A$,
	discards the corresponding message register, but \emph{keeps} the
	flag qubit.
	For all keys $k' \in \mc{K}'$ and messages $m' \in \mc{M}'$, let
	$\sigma_{(k',k''),m'} \in \mc{D}_\bullet(\tsf{A}')$ be the
	subnormalized state defined by
	\begin{equation}
		\sigma_{(k',k''),m'}
		=
		\left(\Tr_{\tsf{M}'} \tensor \Id_{\tsf{A}'}\right)
		\circ
		\overline{D'}_{k'}
		\circ
		A'_{k''}
		\circ
		E'_{k'}
		(m').
	\end{equation}
	Since $S'$ is $\delta'$-tamper-evident, we have that
	\begin{equation}
		\Pr_{k' \gets K'(\emptystring)}\left[
			\frac{1}{2}
			\norm{
				\sigma_{(k',k''),m'_0}
				-
				\sigma_{(k',k''),m'_1}
			}_1
			\leq
			\delta'
		\right]
		\geq
		1 - \delta'
	\end{equation}
	for all $k'' \in \mc{K}''$.
	Indeed, the various keys $k''$ simply represents different tamper
	attacks against the scheme $S'$ and, by our assumption on the
	tamper-evident security of this scheme, the above inequality must
	hold for \emph{all} attacks.
	Since sampling $(k', k'')$ from $K(\varepsilon)$ is
	equivalent to sampling~$k'$ from $K'(\varepsilon)$ and $k''$ from
	$K''(\varepsilon)$ independently, the above implies that
	\begin{equation}
		\Pr_{(k',k'') \gets K(\emptystring)}\left[
			\frac{1}{2}
			\norm{
				\sigma_{(k',k''),m'_0}
				-
				\sigma_{(k',k''),m'_1}
			}_1
			\leq
			\delta'
		\right]
		\geq
		1 - \delta'.
	\end{equation}
	To complete the proof, it now suffices to show that
	\begin{equation}
		\frac{1}{2}\norm{
			\rho_{(k',k''),(m'_0,m''_0)}
			-
			\rho_{(k',k''),(m'_1,m''_0)}
		}_1
		\leq
		\frac{1}{2}\norm{
			\sigma_{(k',k''),m'_0}
			-
			\sigma_{(k',k''),m'_1}
		}_1
	\end{equation}
	for all keys $(k',k'') \in \mc{K}$.

	Recall that $\overline{D''}_{k''}$ is defined by $
		\rho
		\mapsto
		(I_{\tsf{M}''} \tensor \bra{\bo}_{\tsf{F}''})
		D''_{k''}(\rho)
		(I_{\tsf{M}''} \tensor \ket{\bo}_{\tsf{F}''})
	$.
	Moreover, note that the only difference between the subnormalized
	states $\rho_{(k',k''),(m',m''_0)}$ and $\sigma_{(k',k''),m'}$ is
	that the latter is obtained by applying
	$\left(\Tr_{\tsf{M}''}\tensor\Id_{\C^{\bs{}}}\right)\circ D''_{k''}$
	to the $\tsf{C}''$ register
	produced by the channel $A$ and the former obtained by applying
	$\Tr_{\tsf{M}''} \circ \overline{D''}_{k''}$ to this same register.
	Thus, for any message $m' \in \mc{M}''$ and any
	key $(k',k'') \in \mc{K}$, we have that
	\begin{equation}
		\rho_{(k',k''),(m',m''_0)}
		=
		(\bra{\bo} \tensor I_{\tsf{A}})
		\sigma_{(k',k''),m'}
		(\ket{\bo} \tensor I_{\tsf{A}})
	\end{equation}
	from which it follows, using a well-known property of the operator
	norm $\norm{\cdot}_\infty$ (\eg~\cite[Eq.~1.175]{Wat18}), that
	\begin{equation}
	\begin{aligned}
		&
		\frac{1}{2}\norm{
			\rho_{(k',k''),(m'_0,m''_0)}
			-
			\rho_{(k',k''),(m'_1,m''_0)}
		}_1
		\\&\leq
		\frac{1}{2}\norm{
			(\bra{\bo} \tensor I_{\tsf{A}})
			\left(
				\sigma_{(k',k''),m'_0}
				-
				\sigma_{(k',k''),m'_1}
			\right)
			(\ket{\bo} \tensor I_{\tsf{A}})
		}_1
		\\&\leq
		\frac{1}{2}
		\norm{
			(\bra{\bo} \tensor I_{\tsf{A}})
		}_\infty
		\cdot
		\norm{
			\left(
				\sigma_{(k',k''),m'_0}
				-
				\sigma_{(k',k''),m'_1}
			\right)
		}_1
		\cdot
		\norm{
			\ket{\bo} \tensor I_{\tsf{A}}
		}_\infty
		\\&=
		\frac{1}{2}
		\norm{
			\left(
				\sigma_{(k',k''),m'_0}
				-
				\sigma_{(k',k''),m'_1}
			\right)
		}_1,
	\end{aligned}
	\end{equation}
	since $
		\norm{
			\bra{\bo} \tensor I_{\tsf{A}}
		}_\infty
		=
		\norm{
			\ket{\bo} \tensor I_{\tsf{A}}
		}_\infty
		=
		1
	$.
\end{proof}

Before proceeding, we note that the repeated parallel
composition of an AQECM scheme with itself is easily defined.

\begin{definition}
\label{df:aqecm-nfold}
	Let $S$ be an AQECM scheme.
	For all $n \in \N^+$, we define the $n$-fold parallel repetition of
	$S$, denoted by $S^{\tensor n}$, by
	\begin{equation}
		S^{\tensor n}
		=
		\begin{cases}
			S & \qq{if} n = 1 \\
			S^{\tensor(n-1)} \tensor S & \qq{if} n > 1.
		\end{cases}
	\end{equation}
\end{definition}

Finally, the results on how correctness and tamper-evident security is
preserved under parallel composition generalizes as expected to the
repeated parallel composition of AQECM schemes with themselves.

\begin{lemma}
\label{th:aqecm-nfold}
	Let $S$ be an AQECM scheme which is $\epsilon$-correct and
	$\delta$-tamper-evident.
	Then, for all $n \in \N^+$, the scheme
	$S^{\tensor n}$ is $(n \epsilon)$-correct and
	$(n \delta)$-tamper-evident.
\end{lemma}
\begin{proof}
	This follows directly from the previous two results by induction.
\end{proof}

%~~~~~~~~~~~~~~~~~~~~~~~~~~~~~~~~~~~~~~~~~~~~~~~~~~~~~~~~~~~~~~~~~~~~~~%
\subsection{Good Tamper Evidence Implies Good Encryption}
\label{sc:gte=>genc}
%~~~~~~~~~~~~~~~~~~~~~~~~~~~~~~~~~~~~~~~~~~~~~~~~~~~~~~~~~~~~~~~~~~~~~~%

We can now prove the main result of this section, namely that good
tamper-evident schemes are good encryption schemes.

\begin{theorem}
\label{th:te=>enc}
	An $\epsilon$-correct $\delta$-tamper-evident AQECM scheme is
	$\sqrt{19(\delta + \sqrt{2\epsilon})}$-encrypting.
\end{theorem}
\begin{proof}
	First, we can assume for this proof that
	$\delta + \sqrt{2\epsilon} < \frac{1}{19}$ as, otherwise, the result
	already holds since every AQECM scheme is $1$-encrypting.

	Let $S = (K, E, D)$ be an $\epsilon$-correct
	$\delta$-tamper-evident
	AQECM scheme and consider the scheme $S^{\tensor n} = (K', E', D')$,
	the $n$-fold parallel repetition of $S$ with itself for some value
	$n \in \N^+$ to be determined later.
	Note that, by \cref{th:aqecm-nfold}, $S^{\tensor n}$ is
	$n\epsilon$-correct.

	Assume that $S$ is not $\alpha$-encrypting for some $\alpha \geq 0$.
	Then, there exists two messages~$m_0, m_1 \in \mc{M}$ such that
	\begin{equation}
		\frac{1}{2}\norm{
			\E_{k \gets K(\emptystring)} E_k(m_0)
			-
			\E_{k \gets K(\emptystring)} E_k(m_1)
		}_1
		>
		\alpha.
	\end{equation}
	Now, let $m'_0 = m_0^{\times n} = (m_0, \ldots, m_0)$ and
	$m'_1 = m_1^{\times n} = (m_1, \ldots, m_1)$ be the $n$-fold
	repetitions of $m_0$ and $m_1$.
	Note that for each $b \in \{0,1\}$ we have that
	\begin{equation}
		\E_{k' \gets K'(\emptystring)}E'_{k'}(m'_b)
		=
		\left(\E_{k \gets K(\emptystring)} E_k(m_b)\right)^{\tensor n}.
	\end{equation}
	Then, \cref{th:trace-distance-copies} concerning the trace
	distance between copies of two states yields that
	\begin{equation}
	\begin{aligned}
		\frac{1}{2}\norm{
			\E_{k' \gets K'(\emptystring)}
			E'_{k'}(m'_0 - m'_1)
		}_1
		&>
		1
		-
		2\exp\left(
			-\frac{n}{2}
			\left(
			\frac{1}{2}
			\norm{
				\E_{k \gets K(\emptystring)}
				E_k(m_0 - m_1)
			}_1\right)^2
		\right)
		\\&>
		1
		-
		2\exp\left(
			-\frac{n \cdot \alpha^2}{2}
		\right)
	\end{aligned}
	\end{equation}
	which is to say that $S^{\tensor n}$ is not $
		\left(1 - 2\exp\left(-\frac{n \cdot \alpha^2}{2}\right)\right)
	$-encrypting.
	It follows, by \cref{th:benc=>bte}, that $S^{\tensor n}$ is not $
		\left(
			1
			-
			2\exp\left(-\frac{n \cdot \alpha^2}{6}\right)
			-
			\sqrt{2n\epsilon}
		\right)
	$-tamper-evident.
	However, by \cref{th:aqecm-nfold}, $S^{\tensor n}$ is
	$n\delta$-tamper-evident.
	Thus, we must have that
	\begin{equation}
		1 - 2\exp\left(-\frac{n \alpha^2}{6}\right)-\sqrt{2n\epsilon}
		<
		n
		\delta
	\end{equation}
	which implies that
	\begin{equation}
		\exp\left(-\frac{n\alpha^2}{6}\right)
		>
		\frac{1-n\delta-\sqrt{2n\epsilon}}{2}
		>
		\frac{1-n(\delta+\sqrt{2\epsilon})}{2}
	\end{equation}
	where the second inequality follows from the fact that
	$\sqrt{2n\epsilon} < n \sqrt{2\epsilon}$.
	We now fix the value of $n$ to be $
		n = \floor{\frac{1}{2(\delta + \sqrt{2\epsilon})}}
	$.
	Note that this quantity is well defined as at least one of $\delta$ or
	$\epsilon$ must be larger than $0$.\footnote{%
		A perfectly correct AQECM scheme cannot be $0$-tamper-evident
		since an adversary can always correctly guess which key was used
		with probability $1/\abs{\mc{K}} > 0$. Against a perfectly
		correct scheme, an adversary who correctly guessed the key can
		use it to recover the plaintext undetected with certainty.}
	Furthermore, our assumption that $\delta + \sqrt{2\epsilon} <
	\frac{1}{19}$ implies that $n \geq 1$, as required for our
	construction of $S^{\tensor n}$.
	Thus, by placing this value for $n$ in the right-hand side of the
	above inequality, we find that
	$\exp(-\frac{n\alpha^2}{6}) > \frac{1}{4}$ which then implies that
	\begin{equation}
		\alpha^2
		<
		\frac{6\ln(4)}{n}
		<
		\frac{6\ln(4)}{\frac{1}{2(\delta + \sqrt{2\epsilon})} - 1}
		=
		\frac{
			12\ln(4)(\delta + 2\sqrt{2\epsilon})
		}{
			1-2(\delta+\sqrt{2\epsilon})
		}
		<
		\frac{
			12\ln(4)(\delta + \sqrt{2\epsilon})
		}{
			\frac{17}{19}
		}
		<
		19(\delta + \sqrt{2\epsilon})
	\end{equation}
	where the second inequality above follows from the fact that $
		n
		>
		\frac{1}{2(\delta + \sqrt{2\epsilon})} - 1
	$, the second to last from the assumption that
	$\delta + \sqrt{2\epsilon} < \frac{1}{19}$, and the last from direct
	computation.
	Taking the square root, we obtain that $\alpha < \sqrt{19(\delta +
	\sqrt{2\epsilon})}$.

	We have thus shown that if $S$ is not $\alpha$-encrypting, then
	$\alpha < \sqrt{19(\delta + \sqrt{2\epsilon})}$.
	By contraposition, if $\alpha \geq \sqrt{19(\delta +
	\sqrt{2\epsilon})}$, then $S$ must be $\alpha$-encrypting.
	In particular we have that $S$ is
	$\sqrt{19(\delta +\sqrt{2\epsilon})}$-encrypting.
\end{proof}

%======================================================================%
\section{Quantum Money From Any Tamper-Evident Scheme}                 %
\label{ch:te=>qm}                                                      %
\label{te:sc:te=>qm}
%======================================================================%

As discussed in \cref{sc:contributions-qm}, quantum money was first
conceived by Wiesner in the founding paper of quantum cryptography
\cite{Wie83}. By leveraging the no-cloning principle, a bank is able to
mint unforgeable banknotes by incorporating into them quantum systems
whose specific states are unknown to the general public. In 2011,
Gottesman sketched how to obtain a secure quantum money scheme from any
secure tamper-evident scheme \cite{Got11se}.

In this section, we formalize Gottesman's construction of a
quantum money scheme from any tamper-evident scheme and we quantify and
prove the level of security it achieves. We do this in
\cref{te:sc:qm-got} after reviewing the relevant definitions pertaining
to quantum money in \cref{sc:qm}.

%~~~~~~~~~~~~~~~~~~~~~~~~~~~~~~~~~~~~~~~~~~~~~~~~~~~~~~~~~~~~~~~~~~~~~~%
\subsection{Defining Quantum Money}                                    %
\label{sc:qm}                                                          %
%~~~~~~~~~~~~~~~~~~~~~~~~~~~~~~~~~~~~~~~~~~~~~~~~~~~~~~~~~~~~~~~~~~~~~~%

Formally, we define a quantum money scheme as a triplet of channels: a
\emph{key generation} channel $K$, a \emph{minting} channel $M$, and a
\emph{verification} channel $V$. To produce a banknote, the bank first
samples a classical secret key $k$ via the key generation channel,
$k \gets K(\emptystring)$, and then uses it with the minting channel
$M$ to produce a quantum state: $\rho_k = M(k)$. The resulting quantum
state $\rho_k$ is the banknote. We then suppose that the bank places
this banknote in circulation, but that it eventually receives it back as
part of a deposit. The bank can then, with the help of the key $k$,
attempt to validate the banknote via the verification channel $V$. This
channel either accepts by outputting $\bo$ or rejects by
outputting~$\bz$. This is captured later in \cref{df:qm}. We say that a
quantum money scheme is secure if it is infeasible for an adversary
having access to a \emph{single} honestly produced banknote to output
\emph{two} quantum systems which would both be accepted by the bank's
verification procedure. This is formalized in \cref{df:qm-s}.

Note that what we have described here is a \emph{symmetric-} or
\emph{private-key} quantum money scheme: the key $k$ used to generate
the banknote $\rho_k$ is the same as the one used to later validate it.
This is in contrast to \emph{public-key} quantum money schemes which
have also been discussed in the literature (\eg: \cite{AC12}). These
schemes admit a publicly known verification key which allows anyone ---
not only the bank --- to verify if a supposed banknote is legitimate or
not, but which does not aid adversaries in producing forgeries.
Public-key quantum money schemes are typically studied in the context of
computational assumptions and we do not consider them in this work.

The syntax we have outlined here and which we formalize in the next
definition is implicit in Wiesner's original work
on quantum money \cite{Wie83}. It is also essentially the one given in
Aaronson and Christiano's more modern work on this topic \cite{AC12},
up to two changes. First, we identify the public and private keys in
Aaronson and Christiano's work to a be a single private key. Second, we
move their definition from a computational setting, \ie~considering
polynomial-time algorithms, to an information-theoretic setting,
\ie~considering arbitrary quantum channels. We also note that our
definition abstracts away the problem of determining which key is
associated to which banknote. The standard solution, without introducing
any computational assumptions, is for the banknote to be actually
composed of a pair $(\rho_k, s)$ where $s$ is a classical serial number.
This serial number is selected at random and the bank maintains a
large database mapping serial numbers to the private keys used to
generate and verify the banknotes. We will not consider the issue of
properly associating keys to banknotes any further.

\begin{definition}
\label{df:qm}
    A \emph{quantum money scheme} is a triplet of channels $(K, M, V)$
    of the form
    \begin{equation}
            K : \C \to \mc{L}(\tsf{K})
            ,\qq{}
            M : \mc{L}(\tsf{K}) \to \mc{L}(\tsf{N})
            ,\qq{and}
            V : \mc{L}(\tsf{K} \tensor \tsf{N}) \to \mc{L}(\tsf{F})
    \end{equation}
    for an alphabet $\mc{K}$ and Hilbert spaces $\tsf{K} = \C^{\mc{K}}$,
    $\tsf{N}$, and $\tsf{F} = \C^{\bs{}}$.
    The elements of $\mc{K}$ are called the \emph{keys} of the scheme and the banknotes are quantum states on the space $\tsf{N}$.

    Moreover, for every key $k \in \mc{K}$ we additionally define the
    channel $V_k : \mc{L}(\tsf{N}) \to \mc{L}(\tsf{F})$ by $\rho \mapsto
    V(k \tensor \rho)$ and the map $\overline{V}_k : \mc{L}(\tsf{N}) \to
    \C$ by $\rho \mapsto \bra{\bo}V_k(\rho)\ket{\bo}$.
\end{definition}

A quantum money scheme is said to be $\epsilon$-correct if \emph{every}
banknote which it can generate, which is to say those produced by any
key with a non-zero chance of being sampled, will be accepted with
probability at least $1 - \epsilon$. This is analogous to the definition
given by Aaronson and Christiano \cite{AC12}.

\begin{definition}
\label{df:qm-c}
    Let $\epsilon \in \R$.
    A quantum money scheme $(K, M , V)$ as given in \cref{df:qm} is
    \emph{$\epsilon$-correct} if for all keys $k \in \mc{K}$ we have
    that
    \begin{equation}
        \bra{k} K(\emptystring) \ket{k} > 0
        \implies
        \overline{V}_k \circ M(k) \geq 1 - \epsilon.
    \end{equation}
\end{definition}

We now define our security notion. This notion quantifies an upper bound
on the ability of any adversary to transform one honest banknote into
two forgeries which would both be accepted, with the same key, by the
bank. This notion is implicit in Wiesner's original work \cite{Wie83}
and is the one with respect to which Molina, Vidick, and Watrous prove
the security of Wiesner's scheme~\cite{MVW13}.
In particular, the adversarial counterfeiter is not able to query the
bank's verification procedure during the course of their attack. These
are called \emph{simple counterfeiting attacks}\footnote{
    For discussions on more involved attack models against quantum money
    schemes, we invite the interested reader to examine the works of
    Lutowski \cite{Lut10arxiv} and of Nagaj,
    Sattath, Brodutch, and Unruh
    \cite{NSBU16}.}
and they are the only ones we consider in this work.

\begin{definition}
\label{df:qm-s}
    Let $\delta \in \R$.
    A quantum money scheme $(K, M, V)$ as given in \cref{df:qm} is
    \emph{$\delta$-secure} if for all channels
    $A : \mc{L}(\tsf{N}) \to \mc{L}(\tsf{N} \tensor \tsf{N})$ we have
    that
    \begin{equation}
        \E_{k \gets K(\emptystring)}
            \left(\overline{V}_k \tensor \overline{V}_k\right)
            \circ
            A
            \circ
            M(k)
        \leq
        \delta.
    \end{equation}
\end{definition}

%~~~~~~~~~~~~~~~~~~~~~~~~~~~~~~~~~~~~~~~~~~~~~~~~~~~~~~~~~~~~~~~~~~~~~~%
\subsection{Gottesman's Construction of Quantum Money}
\label{sc:qm-got}
\label{te:sc:qm-got}
%~~~~~~~~~~~~~~~~~~~~~~~~~~~~~~~~~~~~~~~~~~~~~~~~~~~~~~~~~~~~~~~~~~~~~~%

Now that we have the necessary definitions pertaining to quantum money,
we can formalize Gottesman's construction of such a scheme from
any tamper-evident scheme \cite{Got11se}. We do so in \cref{df:te=>qm},
after which we state the correctness and security of the resulting
scheme in \cref{te:th:te=>qm-c} and \cref{th:te=>qm-s} after proving a
few lemmas.

As discussed in \cref{sc:contributions-qm}, the basic idea of
Gottesman's construction is to start with a tamper-evident scheme
$(K,E,D)$ and to obtain banknotes by encoding a randomly chosen message
with a randomly chosen key. The bank, which knows the key-message pair,
can verify a candidate banknote by decoding it according to the
tamper-evident scheme and verifying that the state was accepted
\emph{and} that the correct message was obtained. In this
construction, the keys of the resulting quantum money scheme are pairs
composed of a key and a message from the underlying tamper-evident
scheme.

A non-trivial issue which was not discussed by Gottesman but which
must be considered in a formal construction occurs when the underlying
tamper-evident scheme is not perfectly correct. If a tamper-evident
scheme $(K, E, D)$ is not perfectly correct, then there may exist
key-message pairs $(k,m)$ such that the decryption map $D_k$ only
accepts and correctly decodes $E_k(m)$ with arbitrarily small, perhaps
even $0$, probability. This follows from the definition of correctness
for AQECM schemes (\cref{df:aqecm-correct}) which only requires each
message to be accepted and correctly recovered with sufficient high
probability \emph{over the expected choice of the key}. However,
correctness of a quantum money scheme (\cref{df:qm-c}) requires that
banknotes produced with \emph{any valid key} be accepted with high
probability. Thus, naively instantiating Gottesman's construction with
an AQECM scheme which is not perfectly correct may yield a quantum money
scheme without a significant level of correctness.

To solve this issue, we require that the key generation map of the
resulting quantum money scheme only samples from key-message pairs
$(k,m)$ with a sufficiently high probability that $E_k(m)$ will be
correctly decoded and accepted by $D_k$. More precisely, we parametrize
our construction by a real value $\gamma \in \R$ and only sample
key-message pairs $(k,m)$ for which it holds that
$\bra{m} \overline{D}_k \circ E_k(m) \ket{m} \geq 1 - \gamma$.
Our construction must also account for the possibility that the value of
$\gamma$ is such that no key-message pair $(k,m)$ satisfies this
condition. In this case, we fall back to an essentially trivial scheme
which will be correct, but offer no security. All of this is captured in
\cref{df:te=>qm}. Recall that our construction is in the
information-theoretic setting and so we need not worry about the
possible complexity of identifying and sampling from this subset of
key-message pairs. We discuss this point in more detail at the end of
this section.

\begin{definition}
\label{df:te=>qm}
    Let $S = (K, E, D)$ be an AQECM scheme as given in
    \cref{df:aqecm}, let~$\gamma \in \R$ be a real number, and let
    $\mc{G} \subseteq \mc{K} \times \mc{M}$ be the set of key-message
    pairs for $S$ satisfying
    \begin{equation}
        \bra{k}K(\emptystring)\ket{k} > 0
        \qq{and}
        \bra{m} \overline{D}_k \circ E_k(m) \ket{m} \geq 1 - \gamma.
    \end{equation}
    We define the quantum money scheme $\text{QM}_\gamma(S)=(K',M',V')$
    in one of two ways, depending on if $\mc{G}$ is empty or not. In
    both cases, the keys will be elements of
    $\mc{K}' = \mc{K} \times \mc{M}$, implying that the Hilbert space of
    the key space of $\text{QM}_\gamma(S)$ is $\tsf{K}' = \tsf{K}
    \tensor \tsf{M}$, and the banknotes will be states on
    the Hilbert space $\tsf{C}$ corresponding to the ciphertexts of $S$,
    implying that the state space of the banknotes is
    $\tsf{N}' = \tsf{C}$.

    If $\mc{G}$ is non-empty, we define $\text{QM}_\gamma(S)=(K',M',V')$
    as follows:
    \begin{itemize}
        \item
            The key generation channel $K'$ for the quantum money scheme
            $\text{QM}_{\gamma}(S)$ is obtained by sampling a message
            $m$ uniformly at random from the message space $\mc{M}$ of
            $S$ and a key $k$ according to the key generation procedure
            $K$ of $S$, conditioned on the resulting pair $(k,m)$ being
            in $\mc{G}$.

            More formally, consider a key-message pair $(k,m)$ where $k$
            is sampled according to the AQECM's key generation channel
            and where $m$ is sampled uniformly at random. Let $
                p_\mc{G}
                =
                \sum_{(k,m) \in \mc{G}}
                \bra{k} K(\varepsilon) \ket{k}
                \cdot
                \frac{1}{\abs{\mc{M}}}
            $ be the probability that a key-message pair generated in
            this way is in $\mc{G}$. Note that $p_\mc{G} > 0$ as
            $\mc{G}$ is non-empty.
            We then define~$K'$ by
            \begin{equation}
                c
                \mapsto
                c \cdot
                \frac{1}{p_\mc{G}}
                \sum_{(k,m) \in \mc{G}}
                \frac{\bra{k}K(\varepsilon)\ket{k}}{\abs{\mc{M}}}
                \ketbra{k,m}.
            \end{equation}
            In particular, we see that $
                \bra{k,m}K'(\varepsilon)\ket{k,m}
                =
                \frac{
                    \bra{k}K(\varepsilon)\ket{k}
                }{
                    p_\mc{G} \cdot\abs{\mc{M}}
                }
            $ for any pair $(k,m) \in \mc{G}$.
        \item
            On input of a key-message pair $(k,m)$, the minting channel
            $M'$ encrypts $m$ with the key $k$ using the
            encryption map $E$ of $S$. In other words, $M' = E$.
        \item
            On input of a key-message pair $(k,m)$ and a candidate
            banknote $\rho$, the verification channel $V'$
            decrypts $\rho$ using the key $k$ with the AQECM's
            decryption channel $D$. It then accepts if $D$ accepts
            \emph{and} produces $m$. Otherwise, it rejects.

            Formally, and in full generality, $V'$ is given by
            \begin{equation}
            \begin{aligned}
                \rho
                \mapsto
                &
                \left(
                    \sum_{m \in \mc{M}}
                    \bra{m}
                        \overline{D}\left(
                            \left(
                                I_\tsf{K}
                                \tensor
                                \bra{m}_\tsf{M}
                                \tensor
                                I_\tsf{C}
                            \right)
                            \rho
                            \left(
                                I_\tsf{K}
                                \tensor
                                \ket{m}_\tsf{M}
                                \tensor
                                I_\tsf{C}
                            \right)
                        \right)
                    \ket{m}
                \right)
                \left(\ketbra{\bo} - \ketbra{\bz}\right)
                \\&+
                \Tr(\rho) \ketbra{\bz}.
            \end{aligned}
            \end{equation}
            This can be more easily parsed in the case that $\rho$ is
            the tensor of a well formed key for the quantum money scheme
            and some candidate banknote: $
                \rho
                =
                \ketbra*{\tilde{k}}
                \tensor
                \ketbra{\tilde{m}}
                \tensor
                \sigma
            $. In this case, we see that
            \begin{equation}
                \begin{aligned}
                    V'(
                        \ketbra*{\tilde{k}}
                        \tensor
                        \ketbra{\tilde{m}}
                        \tensor
                        \sigma
                    )
                    &=
                    \bra{\tilde{m}}
                    \overline{D}(
                        \ketbra*{\tilde{k}} \tensor \sigma
                    )
                    \ket{\tilde{m}}\left(\ketbra{\bo} -
                    \ketbra{\bz}\right)
                    +
                    \Tr(\sigma)\ketbra{\bz}
                    \\&=
                    \bra*{\tilde{m}}
                    \overline{D}_{\tilde{k}}(\sigma)
                    \ket*{\tilde{m}}
                    \ketbra{\bo}
                    +
                    \left(
                        \Tr(\sigma)
                        -
                        \bra*{\tilde{m}}
                        \overline{D}_{\tilde{k}}(\sigma)
                        \ket*{\tilde{m}}
                    \right)
                    \ketbra{\bz}
                    \\&=
                    p_\text{Accept}     \cdot \ketbra{\bo} +
                    (1-p_\text{Accept}) \cdot \ketbra{\bz}
                \end{aligned}
            \end{equation}
            where $
                p_\text{Accept}
                =
                \bra*{\tilde{m}}
                    \overline{D}_{\tilde{k}}(\sigma)
                \ket*{\tilde{m}}
            $ is the probability that the ciphertext is accepted and the
            message $\tilde{m}$ is obtained.\footnote{%
                Note that the $\Tr(\rho)$ term, which reduces to
                $\Tr(\sigma) = 1$ in our example, is necessary to
                ensure that $V'$ remains trace-preserving in the case
                that $\Tr(\rho) \not= 1$.}
    \end{itemize}

    If $\mc{G}$ is empty, we define $\text{QM}_\gamma(S) = (K', M', V')$
    as in the previous case, but with the following two differences:
    \begin{itemize}
        \item
            We take $K' : \C \to \mc{L}(\tsf{K}')$ to be defined by
            \begin{equation}
                c
                \mapsto
                c
                \cdot
                K(\varepsilon)
                \tensor
                \frac{I_\tsf{M}}{\dim\tsf{M}},
            \end{equation}
            which is to say that a key-message pair $(k,m)$ is obtained
            by sampling $m$ uniformly at random from $\mc{M}$ and $k$
            according to $K$.
        \item
            We take $V'$ to be defined by $\rho \mapsto \Tr(\rho)
            \cdot \ketbra{\bo}$, \ie, $V'$ always accepts.
    \end{itemize}
\end{definition}

We define this construction even in the case where
$\mc{G}$ is empty so that $\text{QM}_\gamma(S)$ is well defined for any
AQECM $S$ and real value $\gamma$. In particular, the map
$S \mapsto \text{QM}_\gamma(S)$ taking AQECM schemes to quantum money
schemes is well defined for all values $\gamma \in \R$.

Note that our definition of $\text{QM}_\gamma(S)$ when $\mc{G}$ is empty
has no chance of offering any non-trivial security since the
verification procedure always accepts. On the other hand, this yields a
scheme which is perfectly correct. Thus, we obtain the following
theorem concerning the correctness of $\text{QM}_\gamma(S)$ and
emphasize that it is independent of the correctness of $S$.

\begin{theorem}
    \label{te:th:te=>qm-c}
    Let $S = (K, E, D)$ be an AQECM scheme and let $\gamma \in \R^+_0$
    be a non-negative real.
    Then, the quantum money scheme $\text{QM}_\gamma(S) = (K', V', M')$
    is $\gamma$-correct.
\end{theorem}
\begin{proof}
    Let $\mc{G}$ be defined as in \cref{df:te=>qm}. If $\mc{G}$ is
    empty, then $V'$ always accepts by construction and so
    \begin{equation}
        \bra{\bo} V'_{k'} \circ M'(k') \ket{\bo}
        =
        1
        \geq
        1 - \gamma
    \end{equation}
    for any key $k' \in \mc{K} \tensor \mc{M}$. This implies that
    $\text{QM}_\gamma(S)$ is $\gamma$-correct.

    If $\mc{G}$ is not empty, then any key $k' = (k,m)$ for the quantum
    money scheme which satisfies
    $\bra{k,m} K'(\varepsilon) \ket{k,m} > 0$ must be in $\mc{G}$. For
    such keys, we have by construction that
    \begin{equation}
        \bra{\bo} V'_{k'} \circ M'(k') \ket{\bo}
        =
        \bra{m} \overline{D}_k \circ E_k(m) \ket{m}
        \geq 1 - \gamma
        .
    \end{equation}
    Thus, $\text{QM}_\gamma(S)$ is $\gamma$-correct.
\end{proof}

We now move on to demonstrating the security of
$\text{QM}_{\sqrt{\epsilon}}(S)$ if the underlying AQECM scheme
$S = (K, E, D)$ is $\epsilon$-correct and $\delta$-tamper-evident.
We begin by sketching the ideas involved in this proof, emphasizing
that these were first presented by Gottesman \cite{Got11se}.

Consider an attack $A$ against the quantum money scheme
$\text{QM}_{\sqrt{\epsilon}}(S)$.
For this attack to succeed, it must, on input of a single ciphertext
$E_k(m)$, produce a bipartite state on two registers such that both
will be accepted and both will produce a correct message when decoded
by the channel $D_k$.
Clearly, an upper bound on the probability of this occurring is the
probability that the first instance of $D_k$ accepts, but may not
have recovered the correct message, and that the second instance of
$D_k$ recovers the correct message, but may not have accepted.
Our proof consists of upper-bounding this probability by leveraging the
fact that $S$ is tamper-evident.
Indeeed, tamper evidence means that any eavesdropping
attempt --- roughly corresponding to a counterfeiting attack against the
resulting quantum money scheme --- will either be detected
(\ie~the first instance of the $D_k$ map will reject) with high
probability or will not provide much information on the encoded message
(\ie~the second instance of the $D_k$ map will not recover $m$ with high
probability).

Our goal is now to formalize this argument.
There are two technical points which need to be considered to obtain a
proper proof:
\begin{enumerate}
    \item
        The definition of tamper evidence, \cref{df:te-security},
        considers the ability of an adversary to distinguish,
        undetected, between any two specific plaintexts.
        To prove the security of $\text{QM}_{\sqrt{\epsilon}}(S)$, we
        need a statement about the (in)ability of an adversary to
        determine, undetected, which \emph{specific} message was sent
        among many possibilities.
    \item
        Recall that the way key-message pairs $(k,m)$ are sampled in
        $\text{QM}_\gamma(S)$ may be slightly different from the way
        they would be sampled when evaluating the security of $S$ as a
        tamper-evident scheme. This occurs, roughly, when
        $\mc{G} \not= \mc{K} \times \mc{M}$. We will need to account for
        this difference in our proof.
\end{enumerate}
We give two lemmas to overcome these two points.

The first point is covered by \cref{th:te-random-message}.
In this lemma, we consider an adversary characterized by two parts which
act sequentially. First, there is an initial tamper attack $A$ which,
as usual, is understood as trying to obtain some information on the
encoded plaintext without being detected. Second, there is a set of
channels $\{A_k\}_{k \in \mc{K}}$ indexed by the keys of the AQECM
scheme which attempt to explicitly guess the plaintext based on what the
initial attack $A$ gathered.\footnote{%
    In the context of our proof of security for the quantum money scheme
    $\text{QM}_\gamma(S)$ these will be taken to be specifically
    $A_k = \left(\Id_\tsf{M} \tensor \Tr_\tsf{F}\right) \circ D_k$,
    which is to say the honest decryption channel which discards the
    flag output. However, our lemma is more general.}
The fact that these channels are indexed by the keys models the fact
that, at this point, we will assume that the adversary knows the key.
The lemma bounds the ability of this two-part adversary to remain
undetected \emph{and} correctly guess the value of the message encoded
by the tamper-evident scheme after they have learned the key.
To prove \cref{th:te-random-message}, we will make use of
another technical lemma, \cref{th:set-discrimination}. This second lemma
gives a bound on how well any quantum channel can distinguish
(subnormalized) states from a given set. In this sense, it can be
thought of as a generalization of \cref{th:trace-channel} which
considered only the case of sets given by precisely two states. Note
that, unlike \cref{th:trace-channel}, we do not claim that
\cref{th:set-discrimination} always gives an optimal bound.

The second point is covered by \cref{th:te=>qm-prob}, which is a
statement purely of probability theory giving an upper bound on the
change in the expectation of a non-negative function when going from one
random variable to another which is simply the first conditioned on some
particular event occurring.

We now state and prove these lemmas.

\begin{lemma}
\label{th:set-discrimination}
    Let $\mc{M}$ be an alphabet and, for each $m \in \mc{M}$, let
    $\rho_m \in \mc{D}_\bullet(\tsf{H})$ be a subnormalized density
    operator on a Hilbert space $\tsf{H}$.
    Further, let $\delta \in \R$ be a real value such that for all
    $m_0, m_1 \in \mc{M}$ we have that
    $\frac{1}{2}\norm{\rho_{m_0} - \rho_{m_1}}_1 \leq \delta$.
    Then, for any random variable $\mu$ distributed on $\mc{M}$ and any
    channel $\Phi:\mc{L}(\tsf{H})\to\mc{L}(\C^\mc{M})$, we have that
    \begin{equation}
        \E_{m \gets \mc{\mu}}
        \bra{m} \Phi(\rho_m) \ket{m}
        \leq
        \left(\max_{m \in \mc{M}} \Pr[\mu = m]\right)\cdot(1 - 2\delta)
        + 2\delta.
    \end{equation}
\end{lemma}
\begin{proof}
    Let $p_m = \Pr[\mu = m]$ for all $m \in\mc{M}$ and let $m' \in
    \mc{M}$ be an element that satisfies
    $p_{m'} = \max_{m \in \mc{M}}\Pr[\mu = m]$.
    By singling out the $m'$ contribution to the expectation and
    recalling that $\Tr(\rho_{m'}) = \sum_{m \in
    \mc{M}}\bra{m}\rho_{m'}\ket{m}$, we have that
    \begin{equation}
    \begin{aligned}
        \E_{m \gets \mu}
        \bra{m} \Phi(\rho_m) \ket{m}
        &=
        p_{m'}\bra{m'} \Phi(\rho_{m'})\ket{m'}
        +
        \sum_{m \in \mc{M} \setminus \{m'\}}
        p_m
        \cdot
        \bra{m} \Phi(\rho_m) \ket{m}
        \\&=
        p_{m'}\Tr\left(\rho_{m'}\right)
        +
        \sum_{m \in \mc{M}\setminus\{m'\}}
        \left(
            p_{m}\bra{m}\Phi(\rho_{m})\ket{m}
            -
            p_{m'}\bra{m'}\Phi(\rho_{m'})\ket{m'}
        \right)
    \end{aligned}
    \end{equation}
    As $p_{m'} \geq p_m$ for all $m \in \mc{M}$, we can replace $p_{m'}$
    by $p_m$ in each term of the summation at the cost of introducing an
    inequality. At the same time, we recall that $\Tr(\rho_{m'})\leq1$.
    Thus,
    \begin{equation}
    \begin{aligned}
        \E_{m \gets \mu}
        \bra{m} \Phi(\rho_m) \ket{m}
        &\leq
        p_{m'}
        +
        \sum_{\mc{M} \setminus \{m'\}}
            p_m
            \cdot
            \left(
                \bra{m}\Phi(\rho_m)\ket{m}
                -
                \bra{m}\Phi(\rho_{m'})\ket{m}
            \right)
        \\&\leq
        p_{m'}
        +
        \sum_{\mc{M} \setminus \{m'\}}
            p_m
            \cdot
            \abs{
                \bra{m}\Phi(\rho_m - \rho_{m'})\ket{m}
            }
    \end{aligned}
    \end{equation}
    By invoking a slight generalization of \cref{th:trace-channel} which
    we will not formally prove,\footnote{%
        Recall that \cref{th:trace-channel} implies that $
            \abs{\bra{\bz} \Phi(\rho - \sigma) \ket{\bz}}
            \leq
            \norm{\rho - \sigma}_1
        $ for any two subnormalized density operators
        $\rho, \sigma \in \mc{D}_\bullet(\tsf{H})$ and any channel
        $\Phi : \mc{L}(\tsf{H}) \to \mc{L}(\C^{\bs{}})$. The slight
        generalizing here is changing the codomain
        of $\Phi$ to $\mc{L}(\C^\mc{M})$  and arguing that the
        inequality does not depend on the element we pick in $\mc{M}$.}
    we can give upper bounds to the absolute
    values in the summands in terms of the trace distance. Further
    recalling that the trace distance is contractive under the action of
    channels, we obtain
    \begin{equation}
    \begin{aligned}
        \E_{m \gets \mu}
        \bra{m} \Phi(\rho_m) \ket{m}
        &\leq
        p_{m'}
        +
        \sum_{m \in \mc{M} \setminus {m'}}
        p_{m} \cdot \norm{\rho_m - \rho_{m'}}_1
        \\&\leq
        p_{m'} + (1-p_{m'})\cdot2\delta
    \end{aligned}
    \end{equation}
    where we invoke our assumption of the pair-wise trace distance
    between the subnormalized states we consider and that $\sum_{m \in
    \mc{M} \setminus \{m'\}} p_m = 1 - p_{m'}$. Rearranging the terms
    yields the desired result.
\end{proof}

Note that if $\Tr(\rho_m) = \Tr(\rho_{m'})$ for all $m \in \mc{M}$, then
we can obtain a better bound. Indeed, under this assumption
\cref{th:trace-channel} yields $
    \abs{\bra{m} \Phi(\rho_m - \rho_{m'}) \ket{m}}
    \leq
    \frac{1}{2}\norm{\rho_m - \rho_{m'}}_1
$ instead of $
    \abs{\bra{m} \Phi(\rho_m - \rho_{m'}) \ket{m}}
    \leq
    \norm{\rho_m - \rho_{m'}}_1
$, implying a final bound of $p_{m'} + (1-p_{m'})\delta$. Moreover, this
matches the bound of \cref{th:trace-channel} if $\mc{M} = \{\bz, \bo\}$,
$\mu$ is uniformly random, and both traces are $1$.

We can now state and prove \cref{th:te-random-message}.

\begin{lemma}
\label{th:te-random-message}
    Let $(K, E, D)$ be a $\delta$-tamper-evident AQECM as given in
    \cref{df:aqecm} for~$\delta \leq 1$ and let $\mu$ be a random
    variable on the messages of this scheme.
    For any tamper attack
    $A : \mc{L}(\tsf{C}) \to \mc{L}(\tsf{C} \tensor \tsf{A})$ against
    this scheme and any set
    $\{A'_k : \mc{L}(\tsf{A}) \to \mc{L}(\tsf{M})\}_{k \in \mc{K}}$
    of channels indexed by the keys of the scheme,
    we have that
    \begin{equation}
    \label{te:eq:te-random-message}
        \E_{\substack{k \gets K(\emptystring)\\m \gets \mu}}
        \bra{m}
            \left(
                \left(\Tr_\tsf{M} \circ \overline{D}_k\right)
                \tensor
                A'_k
            \right)
            \circ
            A
            \circ
            E_k(m)
        \ket{m}
        \leq
        \left(\max_{m \in \mc{M}} \Pr\left[\mu = m\right]\right)
        (1 - 4\delta)
        +
        4\delta.
    \end{equation}
\end{lemma}
\begin{proof}
    For any key $k$ and message $m$ of the AQECM scheme $(K, E, D)$,
    define the two operators
    \begin{equation}
        \rho_{k,m}
        =
        \left(
            \left(\Tr_\tsf{M}\circ\overline{D}_k\right)
            \tensor
            \Id_\tsf{A}
        \right)
        \circ
        A
        \circ
        E_k(m)
        \qq{and}
        \sigma_{k,m}
        =
        A'_k(\rho_{k,m}).
    \end{equation}
    More precisely, $\rho_{k,m} \in \mc{D}_\bullet(\tsf{A})$ is the
    (subnormalized) state held by the adversary after the verifier has
    verified the ciphertext but before they learn the key, and
    $\sigma_{k,m} \in \mc{D}_\bullet(\tsf{M})$ can be understood as the
    adversary's subsequent guess of $m$ after having been given the key
    $k$. Note that the $\rho_{k,m}$ operators are precisely those in the
    left-hand side of \cref{te:eq:te-random-message}.

    Now, by linearity and \cref{th:set-discrimination} with
    $\Phi = \Id_\tsf{M}$, it suffices for us to show that the inequality
    $
        \frac{1}{2}\norm{
            \E_{k \gets K(\varepsilon)}\sigma_{k,m}
            -
            \E_{k \gets K(\varepsilon)}\sigma_{k,m'}
        }
        \leq
        2\delta
    $
    holds for all $m, m' \in \mc{M}$ to obtain the desired result.
    To obtain this bound, we observe that the triangle inequality allows
    us to take the expectation out of the norm and the fact that the
    trace distance is contractive under the action of channels allows us
    to pass from the $\sigma$ operators to the $\rho$ operators. This
    yields that
    \begin{equation}
    \begin{aligned}
        \frac{1}{2}\norm{
            \E_{k \gets K(\varepsilon)}\sigma_{k,m}
            -
            \E_{k \gets K(\varepsilon)}\sigma_{k,m'}
        }
        &\leq
        \E_{k \gets K(\varepsilon)}
        \frac{1}{2}\norm{
            \sigma_{k,m}
            -
            \sigma_{k,m'}
        }
        \\&\leq
        \E_{k \gets K(\varepsilon)}
        \frac{1}{2}\norm{
            \rho_{k,m}
            -
            \rho_{k,m'}
        }.
    \end{aligned}
    \end{equation}
    It now suffices to show that $
        \E_{k \gets K(\varepsilon)}
        \frac{1}{2}\norm{
                \rho_{k,m}
                -
                \rho_{k,m'}
        }_1
        \leq 2\delta
    $. However, this follows directly from the fact that $S$ is assumed
    to be $\delta$-tamper-evident and \cref{tee:th:te-ns} which states
    that $\delta$-tamper-evidence implies a $2\delta$ upper-bound on
    this expectation.
\end{proof}

We now state and prove \cref{th:te=>qm-prob} which will help us
overcome the second technical hurdle we had identified. As
mentioned, this is a result purely of probability theory.

\begin{lemma}
\label{th:te=>qm-prob}
    Let $\mc{A}$ be an alphabet, $A$ be a random variable distributed
    on $\mc{A}$, $\mc{B} \subseteq \mc{A}$ be a susbset such that
    $\Pr\left[A \in \mc{B}\right] > 0$, and $f : \mc{A} \to \R^+_0$ be
    any map. Finally, let $B$ be the random variable distributed on
    $\mc{B}$ defined by $\Pr[B = b] = \Pr[A = b]/\Pr[A \in \mc{B}]$,
    \ie~$B$ is simply $A$ conditioned on $A \in \mc{B}$. Then
    \begin{equation}
        \abs{\E_{a \gets A} f(a) - \E_{a \gets B} f(a)}
        \leq
        \max_{a \in \mc{A}}f(a) \cdot \Pr[A \notin \mc{B}].
    \end{equation}
\end{lemma}
\begin{proof}
    By definition of the expectation and our random variables, we find
    \begin{equation}
    \begin{aligned}
        \abs{\E_{a \gets A} f(a) - \E_{a \gets B} f(a)}
        &=
        \abs{
            \sum_{a \in \mc{A}} \Pr[A = a] \cdot f(a)
            -
            \sum_{a \in \mc{B}} \frac{\Pr[A = a]}{\Pr[A \in \mc{B}]}
            \cdot f(a)
        }
        \\&=
        \abs{
            \sum_{a \in \mc{A} \setminus \mc{B}}
            \Pr[A = a] \cdot f(a)
            -
            \sum_{a \in \mc{B}}
            \left(\frac{\Pr[A = a]}{\Pr[A \in \mc{B}]} - \Pr[A=a]\right)
            \cdot
            f(a)
        }
        \\&=
        \abs{
            \sum_{a \in \mc{A} \setminus \mc{B}}
            \Pr[A = a] \cdot f(a)
            -
            \sum_{a \in \mc{B}}
            \left(1 - \Pr[A \in \mc{B}]\right)\frac{\Pr[A = a]}{\Pr[A
            \in \mc{B}]}
            \cdot
            f(a)
        }.
    \end{aligned}
    \end{equation}
    Since the codomain of $f$ is $\R^+_0$, the last of these absolute
    values is the absolute value of the difference between two
    non-negative real values.
    Hence, the absolute value of the difference is upper-bounded by the
    maximum of these two values.
    Noting that
    \begin{equation}
    \begin{aligned}
        \sum_{a \in \mc{A}\setminus \mc{B}}
        \Pr[A = a] \cdot f(a)
        &\leq
        \sum_{a \in \mc{A} \setminus \mc{B}}
        \Pr[A = a] \cdot \max_{a \in \mc{A}} f(a)
        \\&=
        \Pr[A \not\in \mc{B}] \cdot \max_{a \in \mc{A}} f(a)
    \end{aligned}
    \end{equation}
    and
    \begin{equation}
    \begin{aligned}
        \sum_{a \in \mc{B}}
            \left(1 - \Pr[A \in \mc{B}]\right)\frac{\Pr[A = a]}{\Pr[A
            \in \mc{B}]}
            \cdot
            f(a)
        &\leq
        \Pr[A \not\in\mc{B}]
        \frac{\sum_{a \in \mc{B}} \Pr[A = a]}{\Pr[A \in \mc{B}]}
        \cdot
        \max_{a \in \mc{A}} f(a)
        \\&=
        \Pr[A \not\in\mc{B}]
        \cdot
        \max_{a \in \mc{A}} f(a)
    \end{aligned}
    \end{equation}
    yields the desired result.
\end{proof}

With \cref{th:te-random-message} and \cref{th:te=>qm-prob} in hand, we
can now state and prove the security of the quantum money
scheme $\text{QM}_{\sqrt{\epsilon}}(S)$.

\begin{theorem}
\label{th:te=>qm-s}
    Let $S = (K, E, D)$ be an $\epsilon$-correct $\delta$-tamper-evident
    AQECM scheme with messages~$\mc{M}$.
    Then, the quantum money scheme
    $\text{QM}_{\sqrt{\epsilon}}(S)$ is
    $\left(\abs{\mc{M}}^{-1} + 4\delta + \sqrt{\epsilon}\right)$-secure.
\end{theorem}
\begin{proof}
    First, we assume that $\epsilon < 1$. Otherwise, the result
    trivially holds since any quantum money scheme is $\delta'$-secure
    for any $\delta' \geq 1$.

    Now, note that the set
    \begin{equation}
        \mc{G}
        =
        \left\{
            (k,m) \in \mc{K} \times \mc{M}
            \;:\;
            \bra{k} K(\emptystring) \ket{k} > 0
            \land
            \bra{m} \overline{D}_k \circ E_k(m) \ket{m}
            \geq
            1 - \sqrt{\epsilon}
        \right\}
    \end{equation}
    used in the construction of $\text{QM}_{\sqrt{\epsilon}}(S)$ is
    non-empty.
    Indeed, as the AQECM scheme $S$ is assumed to be $\epsilon$-correct,
    we have that
    \begin{equation}
        \E_{m \gets \mc{M}}
        \E_{k \gets K(\emptystring)}
            \bra{m}
            \overline{D}_k \circ E(m)
            \ket{m}
        \geq 1 - \epsilon
    \end{equation}
    which, by the concentration inequality of \cref{th:concentration},
    implies that
    \begin{equation}
    \label{tee:eq:qm-sec}
        \Pr_{\substack{k \gets K(\emptystring)\\m \gets \mc{M}}}\left[
            (k,m) \in \mc{G}
    \right]
        =
        \Pr_{\substack{k \gets K(\emptystring)\\m \gets \mc{M}}}\left[
            \bra{m} \overline{D}_k \circ E(m)\ket{m}
            \geq
            1 - \sqrt{\epsilon}
        \right]
        \geq
        1 - \sqrt{\epsilon}
        >
        0.
    \end{equation}
    It follows that the quantum money scheme
    $\text{QM}_{\sqrt{\epsilon}}(S) = (K', M', V')$ is constructed
    according to the non-empty $\mc{G}$ case of  \cref{df:te=>qm}.
    Recalling this definition, we see that for any counterfeiting attack
    $A' : \mc{L}(\tsf{C}) \to \mc{L}(\tsf{C} \tensor \tsf{C})$ and any
    $(k,m) \in \mc{K} \times \mc{M} = \mc{K}'$, we have that
    \begin{equation}
    \begin{aligned}
            \left(\overline{V'}_{(k,m)} \tensor \overline{V'}_{(k,m)}\right)
                \circ
                A' \circ M'(k \tensor m)
            &=
            \bra{m,m}
                \left(\overline{D}_k \tensor \overline{D}_k\right)
                \circ
                A'
                \circ
                E_k(m)
            \ket{m,m}
            \\&\leq
            \bra{m}
            \left(
                \left(\Tr_{\tsf{M}} \circ \overline{D}_k\right)
                \tensor
                \tilde{D}_{k}
            \right)
            \circ
            A'
            \circ
            E_k(m)
            \ket{m}
    \end{aligned}
    \end{equation}
    where $\tilde{D}_k = \left(\Id_\tsf{M} \tensor \Tr_\tsf{F}\right)
    \circ D_k$ is the decoding channel which discards, without
    verifying, the accept/reject flag.
    At a high-level this inequality is obtained by
    noting that the probability that both instances of the map $D_k$
    accept \emph{and} decode the correct message is no greater than the
    probability that the first accepts, possibly with the wrong message,
    and that the second decodes the correct message, possibly while
    rejecting.
    Mathematically, it follows from the definition of the trace and the
    fact that for any positive semidefinite operator $\rho$ on the
    appropriate spaces, such as
    $\rho = \left(D_k \tensor D_k\right) \circ A \circ E_k(m)$, it holds
    that
    \begin{equation}
    \begin{aligned}
        \bra{m,\bo,m,\bo} \rho \ket{m,\bo,m,\bo}
        &\leq
        \sum_{m' \in \mc{M}}
        \sum_{b' \in \bs{}}
        \bra{m',\bo,m,b'} \rho \ket{m',\bo,m,b'}
        \\&=
        \bra{\bo,m}
            \left(
                \Tr_\tsf{M}
                \tensor
                \Id_{\tsf{F}}
                \tensor
                \Id_\tsf{M}
                \tensor
                \Tr_\tsf{F}
            \right)
            (\rho)
        \ket{\bo,m}.
    \end{aligned}
    \end{equation}

    By \cref{th:te-random-message} and the fact that $S$ is
    $\delta$-tamper-evident, we have that
    \begin{equation}
    \label{eq:te=>qm-s}
    \begin{aligned}
        \E_{\substack{k \gets K(\emptystring)\\m \gets \mc{M}}}
        \bra{m}
            \left(
                \left(\Tr_{\tsf{M}} \circ \overline{D}_k\right)
                \tensor
                \tilde{D}_{k}
            \right)
            \circ
            A
            \circ
            E_k(m)
        \ket{m}
        &\leq
        \frac{1}{\abs{\mc{M}}}(1 - 4\delta) + 4\delta
        \\&\leq
        \frac{1}{\abs{\mc{M}}} + 4\delta
    \end{aligned}
    \end{equation}
    where, for simplicity, we neglect the $4\delta\abs{\mc{M}}^{-1}$
    term for the remainder of the proof.

    We must now account for the possible deviation of the key generation
    procedures $K'$ of $\text{QM}_{\sqrt{\epsilon}}(S)$ from simply
    sampling $m$ uniformly at random and $k$ from $K(\varepsilon)$.
    For every key-message pair
    $(k,m) \in \mc{K}' = \mc{K} \times \mc{M}$, we have that
    \begin{equation}
        \bra{k,m} K'(\emptystring) \ket{k,m}
        =
        \begin{cases}
            \frac{
                \bra{k,m} K(\emptystring) \ket{k,m}
            }{
                \sum_{(k,m) \in \mc{G}}
                \bra{k,m} K(\emptystring) \ket{k,m}
            }
            &
            \qq{if} (k,m) \in \mc{G}
            \\
            0 & \qq{else.}
        \end{cases}
    \end{equation}
    Hence, by \cref{th:te=>qm-prob}, we have that
    \begin{equation}
    \begin{aligned}
        &
        \E_{(k,m) \gets K'(\emptystring)}
        \bra{m}
            \left(
                \left(\Tr_{\tsf{M}} \circ \overline{D}_k\right)
                \tensor
                \tilde{D}_{k}
            \right)
            \circ
            A
            \circ
            E_k(m)
        \ket{m}
        \\&\leq
        \E_{\substack{k \gets K(\emptystring)\\m \gets \mc{M}}}
        \bra{m}
            \left(
                \left(\Tr_{\tsf{M}} \circ \overline{D}_k\right)
                \tensor
                \tilde{D}_{k}
            \right)
            \circ
            A
            \circ
            E_k(m)
        \ket{m}
        +
        \left(
            1
            -
            \Pr_{\substack{k \gets K(\emptystring)\\m \gets \mc{M}}}
            \left[(k,m)\in\mc{G}\right]
        \right)
        \\&\leq
        \frac{1}{\mc{M}}
        +
        4\delta
        +
        \left(
            1
            -
            \Pr_{\substack{k \gets K(\emptystring)\\m \gets \mc{M}}}
            \left[(k,m)\in\mc{G}\right]
        \right)
        \\&\leq
        \frac{1}{\mc{M}}
        +
        4\delta
        +
        \sqrt{\epsilon}
    \end{aligned}
    \end{equation}
    where the last inequality is obtained by \cref{tee:eq:qm-sec}. This
    is the desired result.
\end{proof}

\paragraph{On avoiding sampling from $\mc{G}$.}
As we have previously mentioned, we do not consider here the complexity
of sampling key-message pairs from $\mc{G}$. If the complexity of this
sampling would be problematic, we can modify our construction of
$\text{QM}_\gamma(S)$ to avoid it. However, this comes at the cost of
only achieving a weaker notion of correctness. Instead of achieving the
notion of correctness we gave in \cref{df:qm-c}, namely that
\begin{equation}
    \bra{k} K(\varepsilon) \ket{k} > 0
    \implies
    \overline{V}_k \circ M(k) \geq 1 - \epsilon,
\end{equation}
which is to say that \emph{all} banknotes produced by the bank are
accepted with high probability, we will only achieve
\begin{equation}
    \label{eq:qm-c-weak}
    \E_{k \gets K(\varepsilon)} \overline{V}_k \circ M(k)
    \geq
    1 - \epsilon,
\end{equation}
which is to say that banknotes are accepted, \emph{on average}, with
high probability. In particular, it may be possible for the bank to
produce faulty banknotes which will not be accepted with any significant
probability.

We denote our modified construction of a quantum money scheme from an
AQECM scheme $S$ by $\text{QM}_\text{weak}(S)$. This is exactly as the
$\mc{G} \not= \varnothing$ case of $\text{QM}_\gamma(S)$ in
\cref{df:te=>qm}, except that key-message pairs $(k,m)$ are always
generated by sampling $k$ according to the key generation channel $S$
and sampling $m$ uniformly at random from the messages of $S$. If $S$ is
an $\epsilon$-correct $\delta$-tamper-evident AQECM scheme,
then it is straightforward to see that the quantum money scheme
$\text{QM}_\text{weak}(S)$ is $\epsilon$-correct according with respect
to the weaker definition of correctness given in \cref{eq:qm-c-weak} and
that it is $\left(\abs{\mc{M}}^{-1} + 4\delta\right)$-secure by applying
the proof of \cref{th:te=>qm-s} but stopping at \cref{eq:te=>qm-s}.

%======================================================================%
\section{Tamper Evidence, Revocation, and Certified Deletion}
\label{te:sc:te<=>rev}
%======================================================================%

In this section, we explore the relation between tamper evidence and
revocation schemes. A high-level overview of this section can be found
in \cref{te:sc:contributions-rev}, our initial discussion in the
introduction.

In \cref{sc:revocation-syntax} we review the syntax of a revocation
scheme and we define the security of such schemes in
\cref{sc:revocation-security-definition}. This includes our novel
definition, the existing game-based definition, and proofs of the near
equivalence of these two definitions. We give our construction of
tamper-evident schemes from revocation schemes in \cref{te:sc:rev=>te}
and our construction of revocation schemes from tamper-evident schemes
in \cref{te:sc:te=>rev}. Finally, in
\cref{sc:revocation=>encryption} we highlight that our results yield as
an easy corollary that revocation implies encryption.

%~~~~~~~~~~~~~~~~~~~~~~~~~~~~~~~~~~~~~~~~~~~~~~~~~~~~~~~~~~~~~~~~~~~~~~%
\subsection{The Syntax of Revocation and Certified Deletion Schemes}
\label{sc:revocation-syntax}
%~~~~~~~~~~~~~~~~~~~~~~~~~~~~~~~~~~~~~~~~~~~~~~~~~~~~~~~~~~~~~~~~~~~~~~%

Our first task in defining revocation and certified-deletion is to
lay out the syntax for schemes aiming to achieve this security notion.
We formalize this syntax in \cref{df:qecmr} as \emph{quantum encryptions
of classical messages with revocation} (QECMR) schemes, or simply
\emph{revocation schemes}, and note that it is essentially the same
definition as the one given for \emph{certified deletion encryption}
schemes by Broadbent and Islam \cite{BI20}, modulo the fact that we move
the definition from a computational setting to an information-theoretic
setting.

A QECMR $(K, E, D, R, V)$ is a QECM $(K, E, D)$ equipped
with two additional channels $R$ and $V$.
The first, $R$, is the \emph{revocation} channel and is unkeyed.
This is the channel which the receiver of the original message should
use if the message is to be revoked or forfeited. It processes a
ciphertext into a state which we will call a \emph{revocation token}.
The second channel,~$V$, is the \emph{verification} channel and it is
keyed.
This channel processes a revocation token with the help of a key and
produces a single qubit indicating whether it accepts or rejects the
revocation.

\begin{definition}
\label{df:qecmr}
	A \emph{quantum encoding of classical messages with revocation}
	(QECMR) is a tuple $(K, E, D, R, V)$ of channels such that $(K,E,D)$
	is a QECM as given in \cref{df:qecm} and $V$ and $R$ are of the form
	\begin{equation}
	\begin{aligned}
		R : \mc{L}(\tsf{C}) \to \mc{L}(\tsf{R})
		\qq{and}
		V : \mc{L}(\tsf{K} \tensor \tsf{R}) \to \mc{L}(\tsf{F})
	\end{aligned}
	\end{equation}
	for an alphabet $\mc{R}$ and Hilbert spaces $\tsf{R} = \C^\mc{R}$
	and $\tsf{F} = \C^{\bs{}}$.
	For all keys $k \in \mc{K}$, we also define the map
	$\overline{V}_k : \mc{L}(\tsf{R}) \to \C$ by
	$\rho \mapsto \bra{\bo} V(k \tensor \rho) \ket{\bo}$.
\end{definition}

The correctness of a QECMR is quantified by a single value
$\epsilon \in \R$ which is an upper bound on both (1) the probability
over the key generation process that any message is not correctly
recovered after decryption and (2) the probability over the key
generation process that any well formed revocation token is not
accepted.\footnote{%
	One could quantify correctness for a
	QECMR with a \emph{pair} of values, one bounding the
	probability of bad decryption and the other the probability of
	bad revocation. We do not need this level of granularity here.}
This is given in \cref{df:qecmr-c}.

\begin{definition}
\label{df:qecmr-c}
	Let $\epsilon \in \R$. A QECMR as given in \cref{df:qecmr} is
    $\epsilon$-correct if the QECM $(K, E, D)$ is $\epsilon$-correct
    and, for all messages $m \in \mc{M}$, we have that
			\begin{equation}
				\E_{k \gets K(\emptystring)}
				\overline{V}_k \circ R \circ E_k(m)
				\geq
				1
				-
				\epsilon.
			\end{equation}
\end{definition}

\paragraph{Certified deletion schemes as a subset of revocation
schemes.}
We briefly comment on the distinction between a revocation scheme and a
certified deletion scheme, but emphasize that this distinction is
inconsequential for our work.
Note that the previous definition does not require the revocation
channel $R$ to produce a classical output. While this property was not
formally encoded in Broadbent and Islam's original definition
\cite{BI20}, it was explicitly emphasized in the surrounding discussion.

A natural way to enforce the requirement that the revocation map
produces a classical output would be to require that
$R = \Delta_{\mc{R}} \circ R$, which is to say that the output of $R$
never changes if it is immediately measured in
the computational basis.
This is \emph{not} what we will do here.
Instead, we make a similar requirement on the verification map $V$.

\begin{definition}
\label{df:cd}
	A QECMR scheme as given in \cref{df:qecmr} is a \emph{certified
	deletion scheme} if
	$V = V \circ \left(\Id_\tsf{K} \tensor \Delta_{\mc{R}}\right)$.
\end{definition}

We impose this condition on the \emph{verification} map instead of
the \emph{revocation} map because, of the two, only the verification map
appears in the upcoming security definitions.
Hence, the condition ``revocation tokens in a certified deletion scheme
must be classical'' needs to be encoded in the verification map for it
to be relevant in the discussion of security. However, for honest users,
there is no distinction between these two options.
Indeed, if the verification map always begins with a computational basis
measurement, the revocation map may as well output a classical state.
Similarly, in this case, adversaries attempting to fool the verification
map will gain no advantage in submitting a non-classical state.
Mathematically, this is due to the fact that the computational basis
measurement
channel $\Delta_\mc{R}$ is idempotent, yielding, for any $k \in \mc{K}$,
that
$
	(V_k \circ \Delta_\mc{R})
	\circ
	(\Delta_\mc{R} \circ R)
	=
	(V_k \circ \Delta_\mc{R}) \circ R
	=
	V_k \circ (\Delta_\mc{R} \circ R)
$.

%~~~~~~~~~~~~~~~~~~~~~~~~~~~~~~~~~~~~~~~~~~~~~~~~~~~~~~~~~~~~~~~~~~~~~~%
\subsection{Defining Revocation Security}
\label{sc:revocation-security-definition}
%~~~~~~~~~~~~~~~~~~~~~~~~~~~~~~~~~~~~~~~~~~~~~~~~~~~~~~~~~~~~~~~~~~~~~~%

We give two definitions of revocation security for QECMR schemes and
show that they are conceptually equivalent. The first,
\cref{df:revocation-a}, is a novel formulation which is well suited for
the information-theoretic regime and closely resembles Gottesman's
definition of tamper evidence. The second, \cref{df:game-revocation}, is
essentially the one given by Broadbent and Islam \cite{BI20}.

\paragraph{Our new definition inspired by tamper evidence.}
We begin by defining a \emph{revocation} attack against an QECMR
scheme. This is a channel which, on input of a ciphertext, aims to
produce a revocation token which will be accepted while also keeping
some information on the ciphertext.

\begin{definition}
\label{df:revocation-a}
	A \emph{revocation attack} against a QECMR scheme as given in
	\cref{df:qecmr} is a channel of the
	form $A : \mc{L}(\tsf{C}) \to \mc{L}(\tsf{R} \tensor \tsf{A})$ for
	some Hilbert space $\tsf{A}$.
\end{definition}

This is quite similar to the definition of a tamper
attack $A: \mc{L}(\tsf{C}) \to \mc{L}(\tsf{C} \tensor \tsf{A})$. The
only difference is that the ciphertext space in the codomain is replaced
with the revocation token space. Conceptually, there is also a close
similarity: in both cases the adversary attempts to have an honest party
accept a state while trying to obtain or keep information.

We now give our novel definition for revocation security, inspired by
the definition of tamper evidence given in \cref{df:te-security}. The
only differences are that we substitute the decryption map of an AQECM
with the verification map of a QECMR and that we bound the
\emph{expectation} of the trace distance over the choice of keys
instead of requiring that it is small with high probability over the
choice of the keys.

\begin{definition}
\label{df:revocation}
	Let $\delta \in \R$ and let $(K, E, D, R, V)$ be a QECMR scheme as 
	given in \cref{df:qecmr}.
	This scheme has $\delta$-revocation security if for any
	two messages $m,m' \in \mc{M}$ and any revocation attack $A$ against
	it we have that
	\begin{equation}
	\label{eq:revocation}
		\E_{k \gets K(\emptystring)}
		\frac{1}{2}\norm{
			\left(\overline{V}_k \tensor \Id_\tsf{A}\right)
			\circ
			A
			\circ E_k(m-m')
		}_1
		\leq
		\delta.
	\end{equation}
\end{definition}

\paragraph{The existing game-based definition of revocation security.}
We now recall the game-based definition of security for revocation,
which is stated in \cref{df:game-revocation}. This will formalize the
idea of the following game played between an adversary and a referee:

\begin{enumerate}
	\item
		The adversary prepares a tripartite state over the space
		$\tsf{M} \tensor \tsf{M} \tensor \tsf{A}$. This should be
		understood as the adversary preparing two candidate messages
		$m_\bz$ and $m_\bo$ on the space $\tsf{M} \tensor \tsf{M}$ as
		well as keeping a memory state on the space $\tsf{A}$.
		The adversary then sends the $\tsf{M} \tensor \tsf{M}$ subsystem
		to the referee.
	\item
		The referee samples a key $k \gets K(\emptystring)$, samples
		uniformly at random a bit $b \gets \bs{}$, and measures the
		received states on the $\tsf{M} \tensor \tsf{M}$ space in the
		computational basis to obtain a result $(m_\bz, m_\bo)$.
		Then, they return $E_k(m_b)$ to the adversary.
	\item
		Acting on their memory state and the $E_k(m_b)$ state received
		from the referee, the adversary produces a bipartite state on
		$\tsf{R} \tensor \tsf{A}$ and sends the $\tsf{R}$ subsystem to
		the referee.
		This should be understood as the adversary attempting to
		convince the referee that the ciphertext has been revoked while
        keeping some information.
	\item
		The referee verifies the state returned by the adversary by
		applying $V_k$ and measuring the resulting qubit in the
		computational basis. Call the result $v$.
	\item
		The referee divulges the key $k$ to the adversary, after which
		the adversary produces a guess $g \in \bs{}$ for the value of
		$b$.
\end{enumerate}

Conceptually, we understand the adversary to win if and only if they
correctly guess the value of $b$ \emph{and} the referee accepted the
revocation.
In other words, the adversary wins if and only if $v = \bo$ and $g = b$.
We could then define the security of this scheme as the advantage of the
adversary in winning this game beyond the trivially achievable
probability of $\frac{1}{2}$.

Instead, we take a slightly different approach and simply
require that the adversary's output $g$ does not depend much on $b$
when $v = \bo$. Formally, the scheme
achieves $\delta$-game-revocation security if the absolute value between
    $\E_{k\gets K(\varepsilon)}\Pr\left[v=\bo\land g=\bo |b=\bz\right]$
and~$\E_{k\gets K(\varepsilon)}\Pr\left[v=\bo\land g=\bo |b=\bo\right]$,
where the probability is taken over one execution of this game with the
sampled key $k$, is at most~$\delta$. This is represented in
\cref{tb:revocation}.

\begin{table}[H]
	\begin{center}
	\begin{tabular}{c c c c c}
	\toprule
		&
		\multicolumn{2}{c}{\footnotesize Ref.~Sampled $b = \bz$} &
		\multicolumn{2}{c}{\footnotesize Ref.~Sampled $b = \bo$} \\
	\cmidrule(lr){2-3} \cmidrule(lr){4-5}
		&
		{\footnotesize Ref.~Rej. ($v = \bz$)} &
		{\footnotesize Ref.~Acc. ($v = \bo$)} &
		{\footnotesize Ref.~Rej. ($v = \bz$)} &
		{\footnotesize Ref.~Acc. ($v = \bo$)} \\
	\midrule
		{\footnotesize Adv.~Guessed $g = \bz$}\\
		{\footnotesize Adv.~Guessed $g = \bo$} &
		&
		$p_{\bz}$
		&
		&
		$p_{\bo}$
		\\
	\bottomrule
\end{tabular}
	\end{center}
	\caption{\label{tb:revocation}%
		Eight possible scenarios in a revocation game. For both values
        of $b' \in \{\bz,\bo\}$, we let
        $p_{b'} = \Pr[g = 1 \land v = 1 | b = b']$ and identify these
		events in the table.}
\end{table}

We formalize this game and security statement in the following
two definitions. In order to differentiate from our previously defined
notions of revocation attacks and revocation security, we will use the
terms \emph{game-revocation attack} and \emph{game-revocation security}
here.

\begin{definition}
\label{df:game-revocation-a}
	Let $(K, E, D, V, R)$ be a QECMR as given in \cref{df:qecmr}.
	A \emph{game-revocation attack} against this scheme is a triplet of
	channel $(A^0, A^1, A^2)$ of the form
	\begin{equation}
		\begin{aligned}
			&
			A^0 :
			\mc{L}(\C) \to \mc{L}(\tsf{M}\tensor\tsf{M}\tensor\tsf{A})
			\\&
			A^1 :
			\mc{L}(\tsf{C}\tensor\tsf{A})
			\to
			\mc{L}(\tsf{R}\tensor\tsf{A})
			\\&
			A^2 :
			\mc{L}(\tsf{K} \tensor \tsf{A}) \to \mc{L}(\C^{\bs{}})
		\end{aligned}
	\end{equation}
	for some Hilbert space $\tsf{A}$.
	
	For every key $k \in \mc{K}$, we also define the map $A^2_k :
	\mc{L}(\tsf{A}) \to \mc{L}(\C^{\bs{}})$ by $\rho \mapsto A^2(k \tensor
	\rho)$.
\end{definition}

Note that $A^0$, $A^1$, and $A^2$ model, respectively, the actions taken
by the adversary in steps $1$, $3$, and $5$ of the game we described.

\begin{definition}
\label{df:game-revocation}
	Let $S = (K, E, D, V, R)$ be a QECMR scheme as given in
	\cref{df:qecmr}.
	For all keys $k \in \mc{K}$, we define the channels
	\begin{equation}
		E^\bz_k
		=
		E_k \circ \left(\Delta_\mc{M} \tensor \Tr_\tsf{M}\right)
		\qq{and}
		E^\bo_k
		=
		E_k
		\circ\left(\Tr_\tsf{M} \tensor \Delta_\mc{M}\right)
	\end{equation}
	which is to say that for both $b \in \bs{}$, $E^b_k$ measures its
	input in the computational basis to obtain an outcome
	$(m_\bz, m_\bo)$ and then outputs $E_k(m_b)$.

	The scheme $S$ has \emph{$\delta$-game-revocation security} if for
	all game-revocation attacks $(A^0, A^1, A^2)$ against it we have
	that
	\begin{equation}
		\abs{
			\E_{k \gets K(\emptystring)}
			\bra{\bo}
			A^2_k
			\circ
			\left(\overline{V}_k \tensor \Id_\tsf{A}\right)
			\circ
			A^1
			\circ
			\left(\left(E^\bz_k - E^\bo_k\right) \tensor \Id_A\right)
			\circ
			A^0(\emptystring)
			\ket{\bo}
		}
		\leq
		\delta.
	\end{equation}
\end{definition}

This is essentially the definition of certified deletion given by
Broadbent and Islam \cite{BI20} and of revocation given by Unruh
\cite{Unr15b}. The main difference is that we do not
state this definition in an asymptotic regime; instead of
requiring the adversaries to have an advantage at most negligible in a
security parameter, we say that the scheme achieves $\delta$-security if
their advantage is at most $\delta$. We also differ from Broadbent and
Islam as we require the adversary to initially submit \emph{two}
messages. Their definition only considered cases where the message
space was the set of bit strings of a fixed length and fixed the second
message which the referee could use to be the all-zero bit string. In
our setting, where we do not impose such restrictions on the message
space, it is more natural to let the adversary submit two messages than
to have a definition which depends on the referee's alternative message.

\paragraph{Another game-based definition.}
A slightly different game-based definition of revocation security has
also been given in the work of  Hiroka, Morimae, Nishimaki, and
Yamakawa~\cite{HMNY21} on certified deletion.
In their definition, the adversary does not need the referee to accept
the revocation token to win the game. However, the adversary is given
the key in step 5 if and only if the referee accepts this token. If the
referee does not accept, then the adversary must generate their guess
$g$ without knowledge of the key $k$. Evidently, a QECMR scheme
achieving a good level of revocation security with respect to the game
of Hiroka \emph{et al}.~must also be a good encryption scheme. Indeed,
if the scheme was not a good encryption scheme, then an adversary could
distinguish the two ciphertexts without needing the key, and hence would
be able to win the game without even trying to obtain the key. Thus, the
Hiroka \emph{et al}.~definition has the advantages of covering two
security notions, encryption and revocation, in a single security game
and may better reflect certain use cases. However, it does have certain
theoretical drawbacks.

One such drawback is that analysing the revocation security of a QECMR
scheme with respect to this definition is a bit more involved than with
respect to \cref{df:game-revocation}. Indeed, we must now consider the
adversary's guessing strategy in two very distinct cases: in the case
where they do get the key, and in the case where they do not.

A second drawback is that this definition ties
together two security notions, encryption and revocation, which
\emph{a priori} are distinct. For example, not every encryption scheme
allows revocation. This grouping of encryption and revocation hinders a
more fine-grained study of these security notions. For example, is it
possible to achieve revocation without encryption? With respect to their
definition, the answer is trivially ``no''. On the other hand, this is a
non-trivial question with respect to \cref{df:game-revocation}, even if
we show later in \cref{sc:revocation=>encryption} that the answer
remains ``no''.
We believe that it is theoretically more interesting to know that
revocation security implies encryption by necessity, and not simply by
definition.

\paragraph{The near equivalence of these two definitions.}
As we have introduced in \cref{df:revocation} a new formulation of
revocation security, we have a responsibility to show how it relates to
the existing definition. We show that there is at most a factor of two
difference between the revocation security and game-revocation security
achieved by any given QECMR scheme.
More precisely, we show in \cref{th:itrev=>gamerev} that
$\delta$-revocation implies $2\delta$-game-revocation and in
\cref{th:gamerev=>itrev} that $\delta$-game-revocation implies
$\delta$-revocation. In other words, if a scheme has
$\delta_\text{r}$-revocation and $\delta_\text{gr}$-game-revocation
security, then
\begin{equation}
	\delta_\text{r}
	\leq
	\delta_\text{gr}
	\leq
	2\delta_\text{r}
	\qq{or, equivalently,}
	\frac{1}{2}\delta_\text{gr}
	\leq
	\delta_\text{r}
	\leq
	\delta_\text{gr}.
\end{equation}

This factor of two between upper and lower bounds may be a proof
artefact. However, it seems to be closely related to the fact that the
inequality $
	\abs{\bra{\bo}\Phi(\rho - \sigma)\ket{\bo}}
	\leq
	\frac{1}{2}\norm{\rho - \sigma}_1
$ for a channel $\Phi$ may not hold if $\rho$ or~$\sigma$ are
subnormalized operators. In the discussion surrounding
\cref{th:trace-channel}, we noted that the correct inequality in
this case is
$\abs{\bra{\bo}\Phi(\rho-\sigma)\ket{\bo}} \leq \norm{\rho - \sigma}_1$,
namely that the upper bound must be multiplied by a factor of two to
hold.

\begin{theorem}
\label{th:itrev=>gamerev}
	A QECMR with $\delta$-revocation security has
	$2\delta$-game-revocation security.
\end{theorem}
\begin{proof}
	Let $(K, E, D, V, R)$ be a QECMR as given in \cref{df:qecmr} and let
	$(A^0, A^1, A^2)$ be a game-revocation attack against this scheme as
	given in \cref{df:game-revocation-a}.
	The first step of this proof is to eliminate the roles
	of the $A^0$ and $A^2$ channels and identify the remaining $A^1$
    channel with a revocation attack as considered in
	\cref{df:revocation}.

	Recall that the quantity to bound is
	\begin{equation}
		\abs{
			\E_{k \gets K(\emptystring)}
			\bra{\bo}
				A^2_k
				\circ
				(\overline{V}_k \tensor \Id_\tsf{A})
				\circ
				A^1
				\circ
				\left((E^0_k - E^1_k) \tensor \Id_\tsf{A}\right)
				\circ
				A^0(\emptystring)
			\ket{\bo}
		}.
	\end{equation}
	We begin by eliminating $A^2$.
	By the triangle inequality and
	\cref{th:trace-channel}, which tells us
	that~$\abs{\bra{\bo} \Phi(A - B)\ket{\bo}} \leq \norm{A - B}_1$ if
	$\Phi$ is a channel and~$A-B$ is Hermitian, we have that
	\begin{equation}
	\label{eq:itrev=>gamerev-1}
	\begin{aligned}
		&
		\abs{
			\E_{k \gets K(\emptystring)}
			\bra{\bo}
				A^2_k
				\circ
				(\overline{V}_k \tensor \Id_\tsf{A})
				\circ
				A^1
				\circ
				\left((E^0_k - E^1_k) \tensor \Id_\tsf{A}\right)
				\circ
				A^0(\emptystring)
			\ket{\bo}
		}
		\\&\leq
		\E_{k \gets K(\emptystring)}
		\abs{
			\bra{\bo}
				A^2_k
				\circ
				(\overline{V}_k \tensor \Id_\tsf{A})
				\circ
				A^1
				\circ
				\left((E^0_k - E^1_k) \tensor \Id_\tsf{A}\right)
				\circ
				A^0(\emptystring)
			\ket{\bo}
		}
		\\&\leq
		\E_{k \gets K(\emptystring)}
		\norm{
			(\overline{V}_k \tensor \Id_\tsf{A})
			\circ
			A^1
			\circ
			\left((E^0_k - E^1_k) \tensor \Id_\tsf{A}\right)
			\circ
			A^0(\emptystring)
		}_1.
	\end{aligned}
	\end{equation}
	Next, we eliminate $A^0$. By definition, we have that
	\begin{equation}
		\left(\left(E^0_k - E^1_k\right) \tensor \Id_\tsf{A}\right)
		\circ
		A^0
		=
		E_k
		\circ
		\left(
			\left(\Id_\tsf{M} \tensor \Tr_\tsf{M}\right)
			-
			\left(\Tr_\tsf{M} \tensor \Id_\tsf{M}\right)
		\right)
		\circ
		\left(\Delta_\mc{M} \tensor \Delta_\mc{M}\right)
		\circ
		A^0
	\end{equation}
	and we can write
	\begin{equation}
		\left(\Delta_\mc{M} \tensor \Delta_\mc{M}\right)
		\circ
		A^0
		(\emptystring)
		=
		\sum_{m,m' \in \mc{M}} p_{m,m'} \cdot
		m \tensor m' \tensor \sigma_{m,m'}
	\end{equation}
	where the $p_{m,m'}$ terms are non-negative reals which sum to $1$
	and $\sigma_{m,m'} \in \mc{D}(\tsf{A})$ is the state of the
	adversarial memory register $\tsf{A}$ if $(m,m')$ is the result of
	the measurement. Thus,
	\begin{equation}
	\begin{aligned}
		\left(\left(E^0_k - E^1_k\right) \tensor \Id_\tsf{A}\right)
		\circ
		A^0(\emptystring)
		=
		\sum_{m,m' \in \mc{M}}
		p_{m,m'}
		\cdot
		E_k(m - m') \tensor \sigma_{m,m'}
	\end{aligned}
	\end{equation}
	and so
	\begin{equation}
	\begin{aligned}
		&
		\E_{k \gets K(\emptystring)}
		\norm{
			(\overline{V}_k \tensor \Id_\tsf{A})
			\circ
			A^1
			\circ
			\left((E^0_k - E^1_k) \tensor \Id_\tsf{A}\right)
			\circ
			A^0(\emptystring)
		}_1
		\\&=
		\E_{k \gets K(\emptystring)}
		\norm{
			\sum_{m,m' \in \mc{M}}
			p_{m,m'}
			(\overline{V}_k \tensor \Id_\tsf{A})
			\circ
			A^1
			\left(
				\left(E_k(m - m')\right) \tensor \sigma_{m,m'}
			\right)
		}_1.
    \end{aligned}
    \end{equation}
    Using the triangle inequality, we can upper bound the values of the
    previous equation by
    \begin{equation}
    \begin{aligned}
        &
		\sum_{m,m' \in \mc{M}}
		p_{m,m'}
		\E_{k \gets K(\emptystring)}
		\norm{
			(\overline{V}_k \tensor \Id_\tsf{A})
		\circ
			A^1
			\left(
				\left(E_k(m - m')\right) \tensor \sigma_{m,m'}
			\right)
		}_1
		\\&=
		\sum_{m,m' \in \mc{M}}
		p_{m,m'}
		\E_{k \gets K(\emptystring)}
		\norm{
			(\overline{V}_k \tensor \Id_\tsf{A})
			\circ
			A'_{m,m'}
			\circ
			E_k(m - m')
		}_1
	\end{aligned}
	\end{equation}
	where, for every $m,m' \in \mc{M}$, we define the channel
	$A'_{m,m'}$ by $\rho \mapsto A^1(\rho \tensor \sigma_{m,m'})$. Note
	that the channel $A'_{m,m'}$ is a revocation attack.
	As we assume that the QECMR scheme has~$\delta$-revocation security,
    we have by definition that
	\begin{equation}
		\E_{k \gets K(\varepsilon)}
		\frac{1}{2}
		\norm{
			(\overline{V}_k \tensor \Id_\tsf{A})
			\circ
			A'_{m,m'}
			\circ
			E_k(m - m')
		}_1
		\leq \delta
	\end{equation}
	for all $m,m' \in \mc{M}$ and so
	\begin{equation}
		\sum_{m,m'}
		p_{m,m'}
		\E_{k \gets K(\emptystring)}
		\norm{
			(\overline{V}_k \tensor \Id_\tsf{A})
			\circ
			A'_{m,m'}
			\circ
			E_k(m - m')
		}_1
		\leq
		\sum_{m,m'} p_{m,m'} 2 \delta
		=
		2\delta
	\end{equation}
	which, with \cref{eq:itrev=>gamerev-1}, yields the desired result.
\end{proof}

We now move on to show that $\delta$-game-revocation
security implies $\delta$-revocation security. We do this by
transforming an revocation attack $A$ and pair of messages $(m,m')$ into
a game-revocation attack $(A^0, A^1, A^2)$. Essentially, $A^0$ simply
generates the state $m \tensor m'$ as the candidate messages for the
referee, $A^1$ is taken to be $A$, and $A^2$ is an optimal
distinguishing channel for the two possible states held by the
adversary. For technical reasons, namely to saturate the first bound in
\cref{th:trace-channel}, we also need to consider a closely related
game-revocation attack $(A^0, A^1, A^2_\text{Flip})$ which simply flips
the guess of $A^2$.

\begin{theorem}
\label{th:gamerev=>itrev}
	A QECMR with $\delta$-game-revocation security has
	$\delta$-revocation security.
\end{theorem}
\begin{proof}
	Let $(K, E, D, V, R)$ be a QECMR as given in \cref{df:qecmr}.
	For any revocation attack $A : \mc{L}(\tsf{C}) \to \mc{L}(\tsf{R}
	\tensor \tsf{A})$ and any two messages $m,m' \in \mc{M}$, we define
	the game-revocation attack $(A^0, A^1, A^2)$ as follows:
	\begin{itemize}
		\item
			The channel $A^0 : \C \to \mc{L}(\tsf{C} \tensor \tsf{C}
			\tensor \tsf{A})$ is defined by $c \mapsto c\left(m\tensor m'
			\tensor \frac{I_\tsf{A}}{\dim\tsf{A}}\right)$.

			The final term in this tensor product is there merely
            because the syntax requires
			$A^0$ to include the $\tsf{A}$ space in its output. It will
			be promptly deleted and replaced by $A^1$.
		\item
			The channel $A^1 : \mc{L}(\tsf{C} \tensor \tsf{A}) \to
			\mc{L}(\tsf{R} \tensor \tsf{A})$ is defined by
			$A^1 = A\circ\left(\Id_\tsf{C} \tensor \Tr_\tsf{A}\right)$,
			which is to say that it simply applies the revocation
			attack $A$.
		\item
			The channel $
				A^2
				:
				\mc{L}(\tsf{K} \tensor \tsf{A})
				\to
				\mc{L}(\C^{\bs})
			$
			is any channel which satisfies
			\begin{equation}
			\begin{aligned}
				&
				\sum_{b \in \bs}
				\abs{
					\bra{b}
					A^2\left(
					\E_{k \gets K(\emptystring)}
					k
					\tensor
					\overline{V}_k
					\circ
					A
					\circ
					E_k(m - m')
					\right)
					\ket{b}
				}
				\\&\hspace{10em}=
				\norm{
					\E_{k \gets K(\emptystring)}
					k
					\tensor
					\overline{V}_k
					\circ
					A
					\circ
					E_k(m - m')
				}_1
			\end{aligned}
			\end{equation}
            which is to say that $A^2$ is an optimal channel to
            distinguish between the possible outputs of $A^1$ when
            given the key $k$.
			Such a channel exists by \cref{th:trace-channel}.
	\end{itemize}
	We also define $A^2_\text{Flip} : \mc{L}(\tsf{K} \tensor \tsf{C}) \to
	\mc{L}(\C^{\bs})$ by $
		\rho
		\mapsto
		(\ketbra{\bo}{\bz} + \ketbra{\bz}{\bo})
		A^2(\rho)
		(\ketbra{\bo}{\bz} + \ketbra{\bz}{\bo})
	$, which is to say that $A_\text{Flip}^2$ is simply $A^2$ followed
    by a bit flip operation.
	In particular, it holds that $
		\bra{\bo}{A}^2_\text{Flip}(\rho)\ket{\bo}
		=
		\bra{\bz}A^2(\rho)\ket{\bz}$
	for any linear operator $\rho$.
	Note that $(A^0, A^1, A_\text{Flip}^2)$ is also a game-revocation
	attack.
	We examine the performance of both these attacks.

	By the definition of $A^0$ and $A^1$, we have that
	\begin{equation}
		A^1
		\circ
		\left((E^0_k - E^1_k)\tensor\Id_\tsf{A}\right)
		\circ
		A^0(\emptystring)
		=
		A \circ E_k(m-m')
	\end{equation}
	and so, summing over both attacks, we see that
	\begin{equation}
	\begin{aligned}
		&
		\sum_{C \in \{A^2,A_\text{Flip}^2\}}
		\abs{
			\E_{k \gets K(\emptystring)}
			\bra{\bo}
				C_k
				\circ
				\left(\overline{V}_k \tensor \Id_\tsf{A}\right)
				\circ
				A^1
				\circ
				\left(
					(E_k^0 - E_k^1) \tensor \Id_\tsf{A}
				\right)
				\circ
				A^0(\emptystring)
			\ket{\bo}
		}
		\\&=
		\sum_{C \in \{A^2,A_\text{Flip}^2\}}
		\abs{
			\E_{k \gets K(\emptystring)}
			\bra{\bo}
				C_k
				\circ
				\left(\overline{V}_k \tensor \Id_\tsf{A}\right)
				\circ
				A
				\circ
				E_k(m - m')
			\ket{\bo}
		}
		\\&=
		\sum_{C \in \{A^2,A_\text{Flip}^2\}}
		\abs{
			\bra{\bo}
				C
				\left(
				\E_{k \gets K(\emptystring)}
				k
				\tensor
				\left(\overline{V}_k \tensor \Id_\tsf{A}\right)
				\circ
				A
				\circ
				E_k(m - m')
				\right)
			\ket{\bo}
		}
		\\&=
		\sum_{b \in \bs}
		\abs{
			\bra{b}
				A^2
				\left(
				\E_{k \gets K(\emptystring)}
				k
				\tensor
				\left(\overline{V}_k \tensor \Id_\tsf{A}\right)
				\circ
				A
				\circ
				E_k(m - m')
			\right)
			\ket{b}
		}
		\\&=
		\norm{
				\E_{k \gets K(\emptystring)}
				k
				\tensor
				\left(\overline{V}_k \tensor \Id_\tsf{A}\right)
				\circ
				A
				\circ
				E_k(m-m')
		}_1
		\\&=
		\E_{k \gets K(\emptystring)}
		\norm{
			\left(\overline{V}_k \tensor \Id_\tsf{A}\right)
				\circ
				A
				\circ
				E_k(m - m')
		}_1
	\end{aligned}
	\end{equation}
	where the first three equalities follow from the definitions of the
	various maps, the fourth from the defining property of $A^2$, and
	the fifth from \cref{th:trace-orthogonal} which states a technical
    property of the Schatten $1$-norm.
	
    As the scheme is assumed to be $\delta$-game-revocation secure,
	both of the initial absolute values is upper-bounded by $\delta$.
	Hence,
	$
		\E_{k \gets K(\emptystring)}
		\norm{
			\left(\overline{V}_k \tensor \Id_\tsf{A}\right)
				\circ
				A
				\circ
				E_k(m - m')
		}_1
		\leq
		2\delta
	$
    which yields the desired result after dividing both sides by two.
\end{proof}

%~~~~~~~~~~~~~~~~~~~~~~~~~~~~~~~~~~~~~~~~~~~~~~~~~~~~~~~~~~~~~~~~~~~~~~%
\subsection{From Revocation to Tamper Evidence and Back}
\label{sc:revocation-tamper-evidence-construction}
\label{sc:te<=>rev}
%~~~~~~~~~~~~~~~~~~~~~~~~~~~~~~~~~~~~~~~~~~~~~~~~~~~~~~~~~~~~~~~~~~~~~~%

We show here how to construct a revocation scheme from a tamper-evident
scheme (\cref{te:sc:te=>rev}) and vice-versa (\cref{te:sc:rev=>te}).
The intuition for these constructions can be conveyed by
recalling and directly comparing the inequalities defining the two
relevant security guarantees.

As stated in \cref{df:te-security}, an AQECM scheme is
$\delta$-tamper-evident if and only if for all channels
$A : \mc{L}(\tsf{C}) \to \mc{L}(\tsf{C} \tensor \tsf{A})$ and all
messages $m, m' \in \mc{M}$ we have that
\begin{equation}
	\Pr_{k \gets K(\emptystring)}
	\left[
		\frac{1}{2}\norm{
			\left(
				\left(\Tr_\tsf{M} \circ \overline{D}_k\right)
				\tensor
				\Id_\tsf{A}
			\right)
			\circ
			A
			\circ
			E_k(m - m')
		}_1
		\leq \delta
		\right]
		\geq
		1-\delta.
\end{equation}
As stated in \cref{df:revocation}, a QECMR scheme has
$\delta$-revocation security if and only if for all channels
$A : \mc{L}(\tsf{C}) \to \mc{L}(\tsf{R} \tensor \tsf{A})$ and all
messages $m, m' \in \mc{M}$ we have that
\begin{equation}
	\E_{k \gets \kappa}
	\frac{1}{2}
	\norm{
		\left(\overline{V}_k \tensor \Id_\tsf{A}\right)
		\circ
		A
		\circ
		E_k(m - m')
	}_1
	\leq
	\delta.
\end{equation}
The striking similarity between these two security guarantees is that
they are both statements on the trace distance between subnormalized
states held by an adversary, conditioned on an honest receiver accepting
what the adversary gave them.

We acknowledge one difference, namely that revocation security is
expressed as a bound on the \emph{expected} trace distance while tamper
evidence states that with \emph{high probability over the choice of
keys} the trace distance is small. However, using the concentration
inequality given in \cref{th:concentration}, we can shift between both
types of statements with ease.

%----------------------------------------------------------------------%
\subsubsection{Revocation From Tamper Evidence}
\label{te:sc:te=>rev}
%----------------------------------------------------------------------%

We show here how to construct a QECMR which is $\epsilon$-correct and
has $2\delta$-revocation security from any $\epsilon$-correct
$\delta$-tamper-evident AQECM.

Our construction, which we denote $S \mapsto \text{Rev}(S)$, is quite
simple. It suffices to note that in an AQECM
scheme the decoding and verification of a ciphertext are done by the
same channel, but that in a QECMR these tasks are accomplished by
different channels. Thus, we can obtain the decoding and revocation maps
of the QECMR by simply using the decoding map of the AQECM and
keeping only the appropriate output. Finally, the revocation map is the
identity, which is to say that the original recipient simply returns the
ciphertext unmodified during the revocation protocol. This is formalized
in \cref{df:aqecm->qecmr}.

\begin{definition}
\label{df:aqecm->qecmr}
\label{df:te=>rev}
	Let $S = (K, E, D)$ be an AQECM scheme as given in \cref{df:aqecm}.
	We define $\text{Rev}(S) = (K', E', D', R', V')$ by $K' = K$,
    $E' = E$, $R = \Id_\tsf{C}$, as well as
	\begin{equation}
		D' = \left(\Id_\tsf{M} \tensor \Tr_\tsf{F}\right) \circ D,
		\qq{and}
		V' = \left(\Tr_\tsf{M} \tensor \Id_\tsf{F}\right) \circ D.
	\end{equation}
	Note that $\text{Rev}(S)$ is a QECMR with the same keys and messages
	as $S$.
\end{definition}

We show in the following theorem that $\text{Rev}(S)$ inherits the
correctness of $S$.

\begin{theorem}
	If $S$ is an $\epsilon$-correct AQECM, then $\text{Rev}(S)$ is an
	$\epsilon$-correct QECMR.
\end{theorem}
\begin{proof}
	Let $S$ be an AQECM as given in \cref{df:aqecm} and let
	$\text{Rev}(S)$ be a QECMR as given in \cref{df:aqecm->qecmr}.
	For any message $m \in \mc{M}$, we have that
	\begin{equation}
	\begin{aligned}
		\E_{k\gets K'(\emptystring)}
			\bra{m}D'_k\circ E'_k(m)\ket{m}
		&=
		\E_{k\gets K(\emptystring)}
			\bra{m}
			\left(\Id_\tsf{M} \circ \Tr_\tsf{F}\right)
			\circ
			D_k
			\circ
			E_k(m)
			\ket{m}
		\\&\geq
		\E_{k\gets K(\emptystring)}
			\bra{m} \overline{D}_k \circ E_k(m) \ket{m}
		\\&\geq
		1 - \epsilon
	\end{aligned}
	\end{equation}
	where the equality follows from the definition of $\text{Rev}(S)$,
	the first inequality by replacing the trace with the map
	$\rho \mapsto \bra{\bo} \rho \ket{\bo}$, and the second by the
	assumption that $S$ is $\epsilon$-correct.
	Similarly, for any message $m \in \mc{M}$, we have that
	\begin{equation}
	\begin{aligned}
		\E_{k\gets K'(\emptystring)}
			\overline{V'}_k \circ R \circ E'_k(m)
		&=
		\E_{k\gets K(\emptystring)}
			\Tr_\tsf{M}
			\circ
			\overline{D}_k
			\circ
			E_k(m)
		\\&\geq
		\E_{k\gets K(\emptystring)}
			\bra{m} \overline{D}_k \circ E_k(m) \ket{m}
		\\&\geq
		1 - \epsilon.
	\end{aligned}
	\end{equation}
	Thus, we can conclude that $\text{Rev}(S)$ is an $\epsilon$-correct
	QECMR scheme.
\end{proof}

Finally, we show that if $S$ is $\delta$-tamper evident, then
$\text{Rev}(S)$ has $2\delta$-revocation security. The factor of two
loss in security is incurred by passing from a concentration inequality
in the definition of tamper evidence to an expectation in the definition
of revocation security, as illustrated in \cref{tee:th:te-ns}.

\begin{theorem}
	$\text{Rev}(S)$ has $2\delta$-revocation security if $S$ is a
	$\delta$-tamper-evident AQECM.
\end{theorem}
\begin{proof}
	We can assume that $0 \leq \delta < \frac{1}{2}$ as the result
	result is trivially true if $\delta \geq \frac{1}{2}$ and there are
    no $\delta$-tamper-evident schemes if $\delta < 0$.

	By definition of $\text{Rev}(S)$, specifically that $R' =
	\Id_\tsf{C}$, revocation tokens are simply ciphertexts from $S$.
	This implies that the Hilbert space $\tsf{R}$ is equal to $\tsf{C}$.
	Thus, any revocation attack
	$A : \mc{L}(\tsf{C}) \to \mc{L}(\tsf{R} \tensor \tsf{A})$ against
	$\text{Rev}(S)$ is also a tamper attack against $S$.
	
	Since $S$ is $\delta$-tamper-evident, for any two
	messages $m, m' \in \mc{M}$ it holds by \cref{tee:th:te-ns} that
	\begin{equation}
		\E_{k \gets K(\emptystring)}
		\frac{1}{2}\norm{
			\left(
				\left(\Tr_\tsf{M} \circ \overline{D}_k\right)
				\tensor
				\Id_\tsf{A}
			\right)
			\circ
			A
			\circ
			E_k(m - m')
		}_1
		\leq
		2\delta - \delta^2
	\end{equation}
	which implies, after substituting in the channels of
	$\text{Rev}(S)$, that
	\begin{equation}
		\E_{k \gets K'(\emptystring)}
		\frac{1}{2}\norm{
			\left(
				\overline{V'}_k
				\tensor
				\Id_\tsf{A}
			\right)
			\circ
			A
			\circ
			E'_k(m - m')
		}_1
		\leq
		2\delta - \delta^2
	\end{equation}
	which is sufficient to obtain that $\text{Rev}(S)$ has
	$2\delta$-revocation security.
\end{proof}

%-----------------------------------------------%
\subsubsection{Tamper Evidence from Revocation} %
\label{te:sc:rev=>te}
%-----------------------------------------------%

We show here how to construct an AQECM which is
$2\sqrt[4]{\epsilon}$-correct and $\sqrt{\delta}$-tamper evident from
any QECMR which is $\epsilon$-correct and has $\delta$-revocation
security. We denote the construction $S \mapsto \text{TE}(S)$.

The main idea of this construction is for the recipient in the
tamper-evident setting to try to detect meaningful eavesdropping by
running the complete revocation protocol of the original revocation
scheme, namely $V \circ (\Id_\tsf{K} \tensor R)$, on their side
upon reception of the ciphertext. They accept the ciphertext if and only
if this revocation protocol accepts.
This leads to an obvious problem: how can they then recover the
plaintext if they have run the revocation protocol?

The solution is for the recipient to actually use the \emph{coherent
gentle measurement} version of
$V\circ (\Id_\tsf{K} \tensor R)$, namely
$\text{CGM}(V \circ (\Id_\tsf{K} \tensor R))$ as described in
\cref{df:cgm}.
The resulting channel produces as output a state on the
ciphertext space as well as a single qubit flag where the resulting
ciphertext state should be quite close to the original one, provided
that it was very likely for the revocation protocol to accept.
The recipient can then subsequently recover the plaintext by applying
the decryption map of the QECMR scheme to this ciphertext state.

This construction is formalized in the following definition.

\begin{definition}
\label{df:qecmr=>aqecm}
\label{df:rev=>te}
	Let $S = (K, E, D, V, R)$ be a QECMR scheme as given in
	\cref{df:qecmr}. We define the AQECM $\text{TE}(S) = (K', E', D')$ by
	\begin{equation}
		K' = K,
		\qq{}
		E' = E,
		\qq{and}
		D' = \left(D \tensor \Id_\tsf{F}\right)
		\circ \CGM(V \circ \left(\Id_\tsf{K} \tensor R\right)).
	\end{equation}
\end{definition}

We now show the correctness of $\text{TE}(S)$.

\begin{theorem}
    \label{te:th:rev=>te-c}
	Let $S$ be an $\epsilon$-correct QECMR.
	Then, $\text{TE}(S)$ is a $2\sqrt[4]{\epsilon}$-correct AQECM
\end{theorem}
\begin{proof}
	Let $S = (K, E, D, V, R)$ be an $\epsilon$-correct QECMR as given in
	\cref{df:qecmr}.
	We assume that $0\leq\epsilon<\frac{1}{16}$ as there is no QECMR
    which is $\epsilon'$-correct if $\epsilon' < 0$ and the result is
	trivially true if $\epsilon \geq \frac{1}{16}$. Let $\text{TE}(S) =
	(K', E', D')$ be as given in \cref{df:rev=>te}.

	Recall that correctness for a QECMR considers the correctness of
	decoding the proper message and the correctness of accepting an
	honest revocation separately.
	In particular, there is \emph{a priori} no guarantee that a message
	$m$ which correctly encodes and decodes with the key $k$ will have
	its revocation with the same key accepted with high probability.
	However, correctness for an AQECM considers both the correct
	decoding and the acceptance simultaneously.

	The fact that the QECMR $S$ is $\epsilon$-correct means that
	\begin{equation}
		\E_{k \gets K(\emptystring)}
		\bra{\bo} V_k \circ R \circ E_k(m) \ket{\bo}
		\geq
		1 - \epsilon
		\qq{and}
		\E_{k \gets K(\emptystring)}
		\bra{m}D_k \circ E_k(m)\ket{m} \geq 1 - \epsilon
	\end{equation}
	for any message $m \in \mc{M}$.
	With the concentration inequality of \cref{th:concentration} and
	the union bound, this implies that
	for any $m \in \mc{M}$ it holds that
	\begin{equation}
		\Pr_{k \gets K(\emptystring)}
		\begin{bmatrix}
			\bra{\bo}V_k\circ R \circ E_k(m)\ket{\bo}
			\geq
			1-\sqrt{\epsilon}
			\\
			\land
			\\
			\bra{m} D_k \circ E_k(m) \ket{m}
			\geq
			1 - \sqrt{\epsilon}
		\end{bmatrix}
		\geq
		1 - 2\sqrt{\epsilon}.
	\end{equation}
	Let $\mc{G}_m \subseteq \mc{K}$ be the set of keys satisfying the
	conditions in this probability statement.
	Note that $1 - 2\sqrt{\epsilon} > 0$ as $\epsilon < 1/4$ and so
	$\mc{G}_k \not= \varnothing$.
	Note that this set depends, in general, on $m$.
	
	By definition of $E'$ and $D'$, we have for any message $m\in\mc{M}$
    and key $k \in \mc{K}$ that
	\begin{equation}
	\label{eq:te=>qecmr-c-1}
	\begin{aligned}
		&
		\bra{m}\overline{D'}_k \circ E'_{k}(m) \ket{m}
		=
		\left(\bra{m} \tensor \bra{\bo}\right)
		\left(D \tensor \Id_\tsf{F}\right)
		\circ \CGM(V \circ \left(\Id_\tsf{K} \tensor R\right))
		\left(k \tensor E_k (m)\right)
		\left(\ket{m} \tensor \ket{\bo}\right).
	\end{aligned}
	\end{equation}
	We begin by examining
	$\CGM(V \circ (\Id_\tsf{K} \tensor R))(k \tensor E_k(m))$ under the
	assumption that $k \in \mc{G}_m$.
	Under this assumption, we have that
    \begin{equation}
		\bra{\bo}
			(V \circ (\Id_\tsf{K} \tensor R))(k \tensor E_k(m))
		\ket{\bo}
		=
		\bra{\bo}V_k \circ R \circ E_k(m)\ket{\bo}
		\geq
		1 - \sqrt{\epsilon}
    \end{equation}
	and so, by the properties of the $\CGM$ map established in
	\cref{th:cgm}, it follows that
	\begin{equation}
		\frac{1}{2}
		\norm{
			\CGM(V \circ (\Id_\tsf{K} \tensor R))(k \tensor E_k(m))
			-
			k \tensor E_k(m) \tensor \ketbra{\bo}
		}_1
		\leq
		\sqrt{1 - (1-\sqrt{\epsilon})^2}
		=
		\sqrt{2\sqrt{\epsilon} - \epsilon}.
	\end{equation}
    We can substitute
	$\CGM(V \circ (\Id_K \tensor R))(k \tensor E_k(m))$ with the
    state $k \tensor E_k(m) \tensor \ketbra{\bo}$ in the rigth-hand side
	of \cref{eq:te=>qecmr-c-1} and its value will change by at most
	$\sqrt{2\sqrt{\epsilon} - \epsilon}$.\footnote{
        That the change in the inner product following such a
        substitution can be controlled by the trace distance is well
        known. It can be recovered as a corollary to the
        second remark after \cref{th:trace-channel} and considering, in
        this case, the
        channel $\Phi = \Ms(\mu)$ where
        $\mu = \bs{} \to \Pos(\tsf{M} \tensor \C^{\bs})$ is the
        measurement defined by~$\mu(\bo)=\ketbra{m}\tensor\ketbra{\bo}$
        and $\mu(\bz) = I_{\tsf{M} \tensor \C^{\bs{}}} - \mu(\bz)$.}
	Thus,
	\begin{equation}
	\begin{aligned}
		&
		\left(\bra{m} \tensor \bra{\bo}\right)
		\left(D \tensor \Id_\tsf{F}\right)
		\circ \CGM(V \circ \left(\Id_\tsf{K} \tensor R\right))
		\left(k \tensor E_k (m)\right)
		\left(\ket{m} \tensor \ket{\bo}\right)
		\\&\geq
		\left(\bra{m} \tensor \bra{\bo}\right)
		\left(D \tensor \Id_\tsf{F}\right)
		(k \tensor E_k(m) \tensor \ketbra{\bo})
		\left(\ket{m} \tensor \ket{\bo}\right)
		-
		\sqrt{2\sqrt{\epsilon} - \epsilon}
		\\&=
		\bra{m} D(k \tensor E_k(m)) \ket{m}
		-
		\sqrt{2\sqrt{\epsilon} - \epsilon}
		\\&\geq
		1 - \sqrt{\epsilon} - \sqrt{2\sqrt{\epsilon} - \epsilon}
	\end{aligned}
	\end{equation}
	where the final inequality is obtained by the assumption that
	$\bra{m} D_k \circ E_k(m) \ket{m} \geq 1 - \sqrt{\epsilon}$ since
	$k \in \mc{G}_m$.

	Finally, we have to account for the cases where $k\not\in\mc{G}_m$.
	We have that
	\begin{equation}
	\begin{aligned}
		\E_{k \gets K'(\emptystring)}
		\bra{m}
		\overline{D'}_k \circ E'_k(m)
		\ket{m}
		&\geq
		\Pr_{k \gets K(\emptystring)}[k \in \mc{G}_m]
		\cdot
		\left(1-\sqrt{\epsilon}-\sqrt{2\sqrt{\epsilon}-\epsilon}\right)
		\\&\geq
		(1 - 2\sqrt{\epsilon})
		\left(1-\sqrt{\epsilon}-\sqrt{2\sqrt{\epsilon}-\epsilon}\right)
		\\&=
		1
		-
		\sqrt[4]{\epsilon}
		\underbrace{\left(
			(1-2\sqrt{\epsilon})
			\sqrt{2-\sqrt{\epsilon}}
			-
			2\epsilon^{3/4}
			+
			3\sqrt[4]{\epsilon}
		\right)}_{=g(\epsilon)}
		\\&\geq
		1-2\sqrt[4]{\epsilon}
	\end{aligned}
	\end{equation}
	where the factor of $2$ results from maximizing $g(\epsilon)$ over
	all values $0 \leq \epsilon \leq \frac{1}{16}$.
	Noting that this holds for all $m \in \mc{M}$ completes the proof.
\end{proof}

We now show that if $S$ has $\delta$-revocation security, then
$\text{TE}(S)$ is $\sqrt{\delta}$-tamper evident. We begin with a
technical lemma making explicit a relation between the maps of these
schemes.

\begin{lemma}
\label{te:th:rev=>te-s-lemma}
	Let $S = (K, E, D, V, R)$ be a QECMR scheme as given in
	\cref{df:qecmr} and let $S' = \text{TE}(S) = (K', E', D')$ be an
	AQECM as given in \cref{df:qecmr=>aqecm}. Then, for any key
	$k \in \mc{K}$, we have that
	\begin{equation}
		\Tr_\tsf{M} \circ \overline{D'}_k = \overline{V}_k \circ R.
	\end{equation}
\end{lemma}
\begin{proof}
	Let $\rho \in \mc{L}(\tsf{C})$ be a linear operator.
	Then, by definition of $D'$, we have that
	\begin{align}
	\begin{split}
		\Tr_\tsf{M} \circ \overline{D'}_k(\rho)
		&=
		\Tr_\tsf{M}\left(
			\left(I_\tsf{M} \tensor \bra{\bo}\right)
			D'(k \tensor \rho)
			\left(I_\tsf{M} \tensor \ket{\bo}\right)
		\right)
		\\&=
		\Tr_\tsf{M}\left(
			\left(I_\tsf{M} \tensor \bra{\bo}\right)
			\left(D \tensor \Id_\tsf{F}\right)
			\circ
			\CGM(V \circ(\Id_\tsf{K} \tensor R))
			(k \tensor \rho)
			\left(I_\tsf{M} \tensor \ket{\bo}\right)
		\right)
	\end{split}
		\\
		\intertext{
			and if we let $C_{\bra{\bo}} : \mc{L}(\tsf{F}) \to \C$ be
			defined by $\sigma \mapsto \bra{\bo} \sigma \ket{\bo}$,
            we find that $$
				(\Tr_\tsf{M} \tensor C_{\bra{\bo}})
				\circ
				(D \tensor \Id_\tsf{F})
				=
				(\Tr_\tsf{M} \circ D) \tensor C_{\bra{\bo}}
				=
				\Tr_{\tsf{K} \tensor \tsf{C}} \tensor C_{\bra{\bo}}
				=
				C_{\bra{\bo}}
				\circ
				\left(
					\Tr_{\tsf{K} \tensor \tsf{C}} \tensor \Id_\tsf{F}
				\right)
			$$
            which allows us to continue the previous equation as
		}
		&=
		\bra{\bo}
		\left(\Tr_{\tsf{K} \tensor \tsf{C}} \tensor \Id_\tsf{F}\right)
		\circ\CGM(V \circ (\Id_\tsf{K} \tensor R))(k \tensor \rho)
		\ket{\bo}
		\\
		\intertext{which, by properties of the $\CGM$ map exhibited in
        \cref{th:cgm}, allows us to further continue as}
	\begin{split}
		&=
		\bra{\bo}
		V \circ \left(\Id_k \tensor R\right)(k \tensor \rho)
		\ket{\bo}
		\\&=
		\bra{\bo} V_k \circ R(\rho) \ket{\bo}
	\end{split}
	\end{align}
	which is precisely $\overline{V}_k \circ R(\rho)$.
	Since this equality holds for all linear operators $\rho$, we obtain
	the desired result that the maps are equal.
\end{proof}

With this lemma in hand, we now proceed to the theorem proper.

\begin{theorem}
    \label{te:th:rev=>te-s}
	If $S$ is a $\delta$-revocation secure QECMR, then $\text{TE}(S)$ is
	$\sqrt{\delta}$-tamper evident.
\end{theorem}
\begin{proof}
	Let $S = (K, E, D, V, R)$ be a QECMR scheme as given in
	\cref{df:qecmr} which has $\delta$-revocation security and let
	$\text{TE}(S) = (K',E',D')$ be as given in \cref{df:qecmr=>aqecm}.

	Let $A' : \mc{L}(\tsf{C}) \to \mc{L}(\tsf{C} \tensor \tsf{A})$ be a
	tamper attack against $\text{TE}(S)$.
	Then, $A = (R \tensor \Id_\tsf{A}) \circ A'$ is a revocation attack
	against $S$.
	Since $S$ has $\delta$-revocation security, for any two messages
	$m, m' \in \mc{M}$ we have that
	\begin{equation}
		\E_{k \gets K(\emptystring)}
		\frac{1}{2}
		\norm{
			\left(\overline{V}_k \tensor \Id_\tsf{A}\right)
			\circ
			A
			\circ
			E_k(m - m')
		}_1
		\leq
		\delta
	\end{equation}
	which, as $\Tr_\tsf{M}\circ\overline{D'}_k=\overline{V}_k \circ R$,
	implies that
	\begin{equation}
	\begin{aligned}
		&
		\E_{k \gets K(\emptystring)}
		\frac{1}{2}
		\norm{
			\left(
				(\Tr_\tsf{M} \circ \overline{D'}_k)
				\tensor
				\Id_\tsf{A}
			\right)
			\circ
			A'
			\circ
			E_k(m - m')
		}_1
		\\&=
		\E_{k \gets K(\emptystring)}
		\frac{1}{2}
		\norm{
			\left(\overline{V}_k \tensor \Id_\tsf{A}\right)
			\circ
			\left(R \tensor \Id_\tsf{A}\right)
			\circ
			A'
			\circ
			E_k(m - m')
		}_1
		\\&=
		\E_{k \gets K(\emptystring)}
		\frac{1}{2}
		\norm{
			\left(\overline{V}_k \tensor \Id_\tsf{A}\right)
			\circ
			A
			\circ
			E_k(m - m')
		}_1
		\\&\leq
		\delta.
	\end{aligned}
	\end{equation}
	Thus, $\text{TE}(S)$ is $\sqrt{\delta}$-tamper evident by
	application of \cref{tee:th:te-ns}.
\end{proof}

%~~~~~~~~~~~~~~~~~~~~~~~~~~~~~~~~~~~~~~~~~~~~~~~~~~~~~~~~~~~~~~~~~~~~~~%
\subsection{Revocation Implies Encryption}
\label{sc:revocation=>encryption}
%~~~~~~~~~~~~~~~~~~~~~~~~~~~~~~~~~~~~~~~~~~~~~~~~~~~~~~~~~~~~~~~~~~~~~~%

An interesting and direct corollary to the work we have done so far 
is that every revocation scheme must be an encryption
scheme. This follows from three facts: First, as discussed in
\cref{sc:te=>enc}, every tamper-evident scheme is an encryption scheme.
Second, as discussed in \cref{te:sc:rev=>te}, every revocation scheme
can be transformed into a tamper-evident scheme via the mapping
$S \mapsto \text{TE}(S)$. And, third, this mapping leaves the key
generation and encryption channel unchanged. These observations yield
the following.

\begin{corollary}
\label{th:cd=>enc}
	Let $S = (K, E, D, V, R)$ be an $\epsilon$-correct QECMR with
	$\delta$-revocation security.
	Then, $(K, E, D)$ is a $\sqrt{19(\sqrt{\delta} +
	2\sqrt[8]{\epsilon})}$-encrypting QECM.
\end{corollary}
\begin{proof}
    By \cref{te:th:rev=>te-c,te:th:rev=>te-s}, we have that 
    $\text{TE}(S)=(K', E', D')$
	is a $2\sqrt[4]{\epsilon}$-correct $\sqrt{\delta}$-tamper-evident
	AQECM. By \cref{th:te=>enc}, which tells us that tamper evidence
	implies encryption, $\text{TE}(S)$ is
	$\sqrt{19(\sqrt{\delta} + 2\sqrt[8]{\epsilon})}$-encrypting. Recall that
    the property of being encrypting, given in  \cref{tee:df:enc},
    depends only on the key generation and the encoding channels.
	Moreover, since $\text{TE}(S)$ and $(K, E, D)$ have the same key
	generation and encoding channels by construction
	(see \cref{df:qecmr=>aqecm}), it follows that
	$(K, E, D)$ is a
	$\sqrt{19(\sqrt{\delta}+2\sqrt[8]{\epsilon})}$-encrypting QECM.
\end{proof}

%======================================================================%
\section{What Tamper Evidence Is Not}
\label{te:sc:separations}
%======================================================================%

We have so far been concerned with ``positive'' implications of
tamper evidence. We complete this work by changing our perspective
and collecting in this section a few shortcomings, or non-implications,
of tamper evidence.
We have four results in this category:
\begin{itemize}
    \item
		In \cref{sc:double}, we show that tamper evidence does not
		imply uncloneable encryption. In other words, a single
		ciphertext can be split into two shares, both of which cause the
		honest decryption procedure to yield the correct message, but
		neither accepting.
    \item
		In \cref{te:sc:ast}, we show that tamper evidence may
		permit an adversary to split a single ciphertext into two states
		which are both accepted by the honest decryption procedure, but
		neither yielding the correct message. In other words,
		tamper evidence does not imply a unique accepting instance.
    \item
		In \cref{te:sc:te-not-auth}, we show that tamper evidence does
		not imply authentication.
	\item
		In \cref{te:sc:te-not-qenc}, we show that tamper evidence does
		not imply the encryption of arbitrary quantum states.
\end{itemize}
In all cases, we obtain optimal separations. We can pick any
$\delta > 0$ and instantiate the result for a perfectly correct
$\delta$-tamper-evident AQECM where the other security notion is, in
some sense, perfectly broken.

The last three results are obtained from the same generic
construction using a simple group-based secret sharing scheme
described in \cref{te:sc:ast}. Even further, we obtain the
last two with precisely the same AQECM scheme described in
\cref{te:sc:otp-te}: a classical one-time pad
where only the key is encoded with a tamper-evident scheme.

%~~~~~~~~~~~~~~~~~~~~~~~~~~~~~~~~~~~~~~~~~~~~~~~~~~~~~~~~~~~~~~~~~~~~~~%
\subsection{Tamper Evidence Does Not Imply Uncloneable Encryption}
\label{sc:double}
\label{te:sc:double}
%~~~~~~~~~~~~~~~~~~~~~~~~~~~~~~~~~~~~~~~~~~~~~~~~~~~~~~~~~~~~~~~~~~~~~~%

Uncloneable encryption was first formalized by Broadbent and Lord
\cite{BL20} following ideas and questions initially raised during
Gottesman's work on tamper-evident schemes \cite{Got03}.
We eschew stating a formal definition of uncloneable encryption in this
work. It suffices to recall the general principle that a symmetric key
encryption scheme is uncloneable if it is infeasible for an initial
adversary ignorant of the key used to encrypt a message to split the
resulting ciphertext between two subsequent non-communicating
adversaries such that both can recover information on the plaintext,
even if \emph{they} know the key. In general, the adversaries
attempting to recover information on the plaintext are not restricted to
using the decryption procedure prescribed by the scheme.

The exact relation between tamper evidence and uncloneable encryption
has remained unclear since Gotesman first highlighted both of these
notions \cite{Got03}. In one direction, it is easy to see that
uncloneability does not imply tamper evidence as we can equip any
uncloneable encryption scheme with a decryption procedure which
\emph{always} accepts. Such a scheme cannot offer a non-trivial level of
security as a tamper-evident scheme. However, this modification
preserves any notion of security as an uncloneable encryption scheme as
this security guarantee is independent of the decryption process.

Here, we complete the picture by showing that there is no implication in
the other direction either. We do so by explicitly constructing
tamper-evident schemes which do not offer any non-trivial security
as uncloneable encryption schemes. Specifically, we define a
construction taking any tamper-evident scheme $S$ to a new scheme
denoted $\text{Double}(S)$ which is the parallel repetition of $S$,
namely $S \tensor S$, encoding the same message \emph{twice} with
independent keys.

Note that $\text{Double}(S)$ differs from $S \tensor S$ in two
respects: it adds to the encryption map a mechanism to copy classical
plaintexts before encoding them and adds a mechanism to select which
of the two recovered messages should ultimately be produced after
decoding both ciphertexts. This ``message selection'' mechanism takes
the form of an extra qubit added to the ciphertexts that the decryption
map will measure to determine which message to output. By showing that
these modifications do not negatively change the correctness or security
of $\text{Double}(S)$ compared to $S \tensor S$, we obtain that
$\text{Double}(S)$ is $2\epsilon$-correct and~$2\delta$-tamper-evident
if $S$ is $\epsilon$-correct and $\delta$-tamper-evident.

Clearly, $\text{Double}(S)$ offers no non-trivial security as an
uncloneable encryption scheme. Indeed, by splitting the two encryptions
between two different parties, both will be able to
recover the plaintext once they learn the key. Even further, due to the
way we construct the ``message selection'' mechanism of the decryption
channel, both parties will be able to use the honest decryption map to
recover the message, something that is stronger than what is needed to
break the uncloneable encryption security guarantee.

We formalize this discussion in what follows.

\begin{definition}
\label{df:double}
	Let $S = (K, E, D)$ be an AQECM scheme as given in \cref{df:aqecm}
	and let $S \tensor S = (K',E',D')$ be the parallel repetition of
	this scheme with itself, as given in \cref{df:aqecm-parallel}.
	Finally, let $V \in \mc{U}(\tsf{M}, \tsf{M} \tensor \tsf{M})$ be the
	isometry defined by $V = \sum_{m \in \mc{M}}\ketbra{m,m}{m}$.
	Note that $V$ perfectly copies all elements of $m \in \mc{M}$, which
	is to say that $V m V^\dag = m \tensor m$.

	We define $\text{Double}(S)=(K'', E'', D'')$ to be the triplet
	given by the following channels:
	\begin{itemize}
		\item
			The key generation channel $K''$ is identical to the key
			generation channel of $S \tensor S$,
			\ie~$K'' = K' = K \tensor K$.
		\item
			The encoding channel $
				E'' :
				\mc{L}(\tsf{K} \tensor \tsf{K} \tensor \tsf{M})
				\to
				\mc{L}(\C^{\bs{}} \tensor \tsf{C} \tensor \tsf{C})
			$ is given by
			\begin{equation}
				\rho \mapsto
				\ketbra{\bz}
				\tensor
				E'\left(
					\left(I_{\tsf{K}\tensor\tsf{K}}\tensor V\right)
					\rho
					\left(I_{\tsf{K}\tensor\tsf{K}}\tensor V\right)^\dag
				\right)
			\end{equation}
			which is to say that $E''$ first copies the plaintext,
			assuming it is in the computational basis, via the isometry
			$V$ and then simply applies the encoding channel $E'$ from
			$S \tensor S$, followed by prepending a single qubit in
			the $\bz$ state.
		\item
			Let $
				\Psi :
				\mc{L}(
					\C^{\bs{}}
					\tensor
					\tsf{M}
					\tensor
					\tsf{M}
				)
				\to
				\mc{L}(
					\tsf{M}
				)
			$ be a channel which measures, in the computational basis,
			the first qubit it is given and then either discards the
			second $\tsf{M}$ register if the measurement outcome is
			$\bz$ or discards the first $\tsf{M}$ register if the
			outcome is $\bo$.
			In other words, the first qubit indicates which plaintext
			register is kept.
			Formally, this is given by the map
			\begin{equation}
			\begin{aligned}
				\rho
				\mapsto
				&
				\left(\Id_\tsf{M} \tensor \Tr_\tsf{M}\right)\left(
					\left(
						\bra{\bz}
						\tensor
						I_{\tsf{M} \tensor \tsf{M}}
					\right)
					\rho
					\left(
						\ket{\bz}
						\tensor
						I_{\tsf{M} \tensor \tsf{M}}
					\right)
				\right)
				\\&+
				\left(\Tr_\tsf{M} \tensor \Id_\tsf{M}\right)\left(
					\left(
						\bra{\bo}
						\tensor
						I_{\tsf{M} \tensor \tsf{M}}
					\right)
					\rho
					\left(
						\ket{\bo}
						\tensor
						I_{\tsf{M} \tensor \tsf{M}}
					\right)
				\right).
			\end{aligned}
			\end{equation}

			The decoding channel $
				D'' :
				\mc{L}(
					\tsf{K}\tensor\tsf{K}
					\tensor
					\C^{\bs{}}\tensor\tsf{C}\tensor\tsf{C}
				)
				\to
				\mc{L}(\tsf{M} \tensor \tsf{F})
			$ is then given by
			\begin{equation}
				D'' = 
				\left(\Psi \tensor \Id_{\tsf{F}}\right)
				\circ
				\left(\Id_{\C^{\bs}} \tensor D'\right)
				\circ \left(\Swap_{\tsf{K}, \tsf{K},
				\C^{\bs{}}}^{(2,3,1)} \tensor \Id_{\tsf{C} \tensor
				\tsf{C}}\right)
			\end{equation}
			which is to say that $D''$ begins by applying $D'$, the
			decoding map from $S \tensor S$, on all but the first qubit,
			then checks the first qubit via a computational basis
			measurement to determine which message register $\tsf{M}$ to
			keep and which to discard, all the while keeping the flag
			qubit produced by $D'$ untouched. The swap channel is merely
			there to correctly route the key material.
			Finally, note that $
				\overline{D''}_{(k_0, k_1)}
				=
				\Psi
				\circ
				\left(
					\Id_{\C^{\{0,1\}}} 
					\tensor
					\overline{D'}_{(k_0, k_1)}
				\right)
			$ for all keys $(k_0, k_1) \in \mc{K} \times \mc{K}$.
	\end{itemize}
\end{definition}

We now prove that the mapping $S \mapsto \text{Double}(S)$
preserves the correctness and security of the scheme, up
to the factor of $2$ introduced by passing via the $S \tensor S$ scheme.

For correctness, this is essentially trivial. The probability of
accepting and correctly recovering both messages, \ie~the
$S \tensor S$ correctness criterion, cannot be less than the
probability of accepting both messages and correctly recovering the
first, which is precisely the $\text{Double}(S)$ correctness criterion.

\begin{theorem}
\label{th:double-c}
	If $S$ is an $\epsilon$-correct AQECM scheme, then
	$\text{Double}(S)$ is $2\epsilon$-correct.
\end{theorem}
\begin{proof}
	This follows essentially immediately from the fact that
	$S \tensor S$ is $2\epsilon$-correct, as shown in
	\cref{th:aqecm-parallel}. We go through the details for
	completeness.

	Let $S$, $S \tensor S$, and $\text{Double}(S)$ be defined and
	denoted as in \cref{df:double}.
	By definition, for every key $(k_0,k_1)\in\mc{K}\times\mc{K}$ and
	every message $m \in \mc{M}$ we have that
	\begin{equation}
		E''_{(k_0,k_1)}(m)
		=
		\ketbra{\bz}
		\tensor
		E'_{(k_0,k_1)}(m\tensor m).
	\end{equation}
	It then follows that
	\begin{equation}
	\begin{aligned}
		&
		\E_{(k_0,k_1) \gets K''(\emptystring)}
		\bra{m}
			\overline{D''}_{(k_0,k_1)} \circ E''_{(k_0,k_1)}(m)
		\ket{m}
		\\&=
		\E_{(k_0,k_1) \gets K'(\emptystring)}
		\bra{m}
			\left(\Id_\tsf{M} \tensor \Tr_{\tsf{M}}\right)
			\circ
			\overline{D'}_{(k_0, k_1)}
			\circ
			E'_{k_0,k_1}(m \tensor m)
		\ket{m}
		\\&\geq
		\E_{(k_0,k_1) \gets K'(\emptystring)}
		\bra{m,m}
			\overline{D'}_{(k_0,k_1)}
			\circ
			E'_{(k_0,k1)}(m \tensor m)
		\ket{m,m}
		\\&\geq
		1 - 2\epsilon
	\end{aligned}
	\end{equation}
	where the first inequality holds as the event of 
	measuring both plaintexts to be $m$ cannot be more likely than the
	event of measuring the first plaintext to be $m$ and the second
	inequality follows from the fact that $S \tensor S$ is
	$2\epsilon$-correct. Thus, $\text{Double}(S)$ is
	$2\epsilon$-correct.
\end{proof}

For security, we note that the choice to accept or reject
the ciphertext is independent of the extra qubit used for the message
selection mechanism and that this qubit does not depend on the original
message or key. Informally, this means it does not provide any useful
information to the adversary. Formally, this means that it can be
simulated by the adversary.

\begin{theorem}
	\label{te:th:double-s}
	If $S$ is a $\delta$-tamper-evident AQECM, then
	$\text{Double}(S)$ is $2\delta$-tamper-evident.
\end{theorem}
\begin{proof}
	This follows quickly from the fact that
	$S \tensor S$ is $2\delta$-tamper-evident, as shown in
	\cref{th:aqecm-parallel}. We demonstrate that any attack against
	$\text{Double}(S)$ implies an attack against $S \tensor S$ with
	the same performance by showing that an adversary can simulate all
	differences between these two schemes.

	Let $S$, $S \tensor S$, and $\text{Double}(S)$ be defined and
	denoted as in \cref{df:double}. Next, let~$
		A'' :
		\mc{L}(\C^{\bs{}} \tensor \tsf{C} \tensor \tsf{C})
		\to
		\mc{L}(\C^{\bs{}}\tensor\tsf{C}\tensor\tsf{C}\tensor\tsf{A})
	$ be a tamper attack against $\text{Double}(S)$. Consider now the
	tamper attack $
		A' :
		\mc{L}(\tsf{C} \tensor \tsf{C})
		\to
		\mc{L}(\tsf{C}\tensor\tsf{C}\tensor\tsf{A})
	$ against $S \tensor S$ given by
	\begin{equation}
		\rho
		\mapsto
		\left(
			\Tr_{\C^{\bs{}}}
			\tensor
			\Id_{\tsf{C} \tensor \tsf{C} \tensor \tsf{A}}
		\right)
		\circ
		A''\left(\ketbra{\bz}\tensor\rho\right)
	\end{equation}
	In essence, $A'$
	simulates the attack $A''$ by first providing it with the extra
	qubit it expects, and then removing it before returning a ciphertext
	to the decoding map $D'$ of $S' = S \tensor S$. So, for every
	message $m \in \mc{M}$ and key
	$(k_0, k_1) \in \mc{K} \times
	\mc{K}$, we have by definition that
	\begin{equation}
		A' \circ E'_{(k_0, k_1)}(m \tensor m)
		=
		\left(
			\Tr_{\C^{\bs{}}}
			\tensor
			\Id_{\tsf{C} \tensor \tsf{C} \tensor \tsf{A}}
		\right)
		\circ
		A'' \circ E''_{(k_0, k_1)}(m).
	\end{equation}
	Next, we claim that
	\begin{equation}
		\Tr_{\tsf{M}} \circ \overline{D''}_{(k_0, k_1)}
		=
		\Tr_{\tsf{M}\tensor\tsf{M}} \circ \overline{D'}_{(k_0, k_1)}
		\circ
		\left(
			\Tr_{\C^{\{0,1\}}}
			\tensor
			\Id_{\tsf{C} \tensor \tsf{C}}
		\right).
	\end{equation}
	Intuitively, this is simply stating that discarding the single
	plaintext ultimately produced by $\overline{D''}_{(k_0, k_1)}$ is
	equivalent to discarding both plaintexts produced by
	$\overline{D'}_{(k_0, k_1)}$, ignoring for the moment the extra
	qubit in the ciphertext. Formally, the trace-preserving property of
	channels implies that $
		\Tr_\tsf{M} \circ \Psi
		=
		\Tr_{\C^{\bs{}}\tensor\tsf{M} \tensor \tsf{M}}
		=
		\Tr_{\C^{\bs{}}}
		\tensor
		\Tr_{\tsf{M} \tensor \tsf{M}}
		$.
	We can then write
	\begin{equation}
	\begin{aligned}
		\Tr_\tsf{M} \circ \overline{D''}_{(k_0,k_1)}
		&=
		\Tr_\tsf{M}
		\circ
		\Psi
		\circ
		\left(
			\Id_{\C^{\bs{}}}
			\tensor
			\overline{D'}_{(k_0, k_1)}
		\right)
		\\&=
		\Tr_{\C^{\bs{}} \tensor \tsf{M} \tensor \tsf{M}}
		\left(
			\Id_{\C^{\bs{}}}
			\tensor
			\overline{D'}_{(k_0, k_1)}
		\right)
		\\&=
		\Tr_{\C^{\bs{}}}
		\tensor
		\left(
			\Tr_{\tsf{M} \tensor \tsf{M}}
			\circ
			\overline{D'}_{(k_0, k_1)}
		\right)
		\\&=
		\left(
			\Tr_{\tsf{M} \tensor \tsf{M}}
			\circ
			\overline{D'}_{(k_0, k_1)}
		\right)
		\circ
		\left(
			\Tr_{\C^{\bs{}}}
			\tensor
			\Id_{\tsf{C} \tensor \tsf{C}}
		\right)
	\end{aligned}
	\end{equation}
	which is the desired claim.

	This, and a final remainder that $K'' = K'$, allows us to
	write that for any two messages $m, m' \in \mc{M}$ it holds that
	\begin{equation}
	\begin{aligned}
		&
		\Pr_{(k_0,k_1) \gets K''(\emptystring)}\left[
			\frac{1}{2}\norm{
					\left(
						\left(
							\Tr_\tsf{M}
							\circ
							\overline{D''}_{(k_0,k_1)}
						\right)
						\tensor
						\Id_{\tsf{A}}
					\right)
					\circ
					A''
					\circ
					E''_{(k_0,k_1)}
					(m - m')
				}_1
				\leq
				2\delta
		\right]
		\\&=
		\Pr_{(k_0,k_1) \gets K'(\emptystring)}\left[
			\frac{1}{2}\norm{
					\left(
						\left(
							\Tr_{\tsf{M}\tensor \tsf{M}}
							\circ
							\overline{D'}_{(k_0,k_1)}
						\right)
						\tensor
						\Id_{\tsf{A}}
					\right)
					\circ
					A'
					\circ
					E'_{(k_0,k_1)}
					(m \tensor m - m' \tensor m')
				}_1
				\leq
				2\delta
		\right]
		\\&\geq
		1 - 2\delta
	\end{aligned}
	\end{equation}
	where the last inequality follows from the fact that $S \tensor S$
	is $2\delta$-tamper evident. Thus, the scheme $\text{Double}(S)$ is
	also $2\delta$-tamper evident.
\end{proof}

As mentioned, the end result of this construction is
that it yields an optimal separation between tamper-evident
schemes and uncloneable encryption schemes. This is the
content of the next theorem.

\begin{theorem}
\label{th:te=/>ue}
	For any strictly positive $\delta \in \R^+$ and any alphabet
	$\mc{M}$, there exists
			a perfectly correct $\delta$-tamper-evident scheme
			$S' = (K', E', D')$ with messages $\mc{M}$ and ciphertext
			space $\tsf{C}'$ as well as
			a channel
			$
			\tilde{A} :
			\mc{L}(\tsf{C}' \tensor \tsf{C}')
			\to
			\mc{L}(\tsf{C}' \tensor \tsf{C}')
			$
	such that for any message $m \in \mc{M}$ and any key $k'$ with
	$\bra{k'}K'(\emptystring)\ket{k'} > 0$, we have that
	\begin{equation}
		\bra{m,m}
		\left(\Id_\tsf{M} \tensor \Tr_{\tsf{F}} \tensor \Id_{\tsf{M}}
		\tensor \Tr_{\tsf{F}}\right)
		\circ
		\left(D'_{k'} \circ D'_{k'}\right)
		\circ
		\tilde{A}
		\circ
		E'_k(m)
		\ket{m,m}
		=
		1.
	\end{equation}
\end{theorem}
\begin{proof}
	By \cref{te:th:te-exist}, there exists a perfectly correct
	$\frac{\delta}{2}$-tamper-evident AQECM scheme $S$ with messages
	$\mc{M}$. Then, $S' = \text{Double}(S)$ is a perfectly correct
	$\delta$-tamper-evident scheme. It then suffices to consider the
	channel $\tilde{A}$ given by
	\begin{equation}
		\rho
		\mapsto
		\Swap_{
				\C^{\bs{}}, \C^{\bs{}},
				\tsf{C},\tsf{C},\tsf{C}, \tsf{C}
			}^{
				(1,4,3,5,2,6)
			}
		\left(
		\ketbra{\bz} \tensor \ketbra{\bo}
		\tensor
		\frac{\Id_\tsf{C}}{\dim \tsf{C}}
		\tensor
		\frac{\Id_\tsf{C}}{\dim \tsf{C}}
		\tensor
		\left(
			\Tr_{\C^{\bs{}}}
			\tensor \Id_{\tsf{C} \tensor \tsf{C}}
		\right)(\rho)
		\right)
	\end{equation}
	so that for any key $k' = (k_0, k_1)$ which has a non-zero
	probability of being produced by $K'$ and any message $m$ we have
	\begin{equation}
		\tilde{A} \circ E'_{(k_0,k_1)}(m)
		=
		\left(
			\ketbra{\bz}
			\tensor
			E_{k_0}(m)
			\tensor
			\frac{\Id_\tsf{C}}{\dim \tsf{C}}
		\right)
		\tensor
		\left(
			\ketbra{\bo}
			\tensor
			\frac{\Id_\tsf{C}}{\dim \tsf{C}}
			\tensor
			E_{k_1}(m)
		\right).
	\end{equation}
	Both of these are valid ciphertexts for the scheme
	$\text{Double}(S)$ and it is clear that the decryption map
	$D'_{(k_0,k_1)}$ of $\text{Double}(S)$ will recover the correct
	plaintext $m$ in both cases due to how the plaintext selection qubit
	is fixed and the perfect correctness of the scheme $S$.
\end{proof}

%~~~~~~~~~~~~~~~~~~~~~~~~~~~~~~~~~~~~~~~~~~~~~~~~~~~~~~~~~~~~~~~~~~~~~~%
\subsection{Tamper Evidence with a Simple Secret Sharing
Scheme}
\label{te:sc:gss}
\label{te:sc:ast}
%~~~~~~~~~~~~~~~~~~~~~~~~~~~~~~~~~~~~~~~~~~~~~~~~~~~~~~~~~~~~~~~~~~~~~~%

We describe in this section a method to construct a tamper-evident
scheme by combining two tamper-evident schemes with a secret sharing
scheme (\eg: \cite[Sec.~13.3]{KL15}). This construction will be useful
to obtain various separations.

We conjecture that our construction works with \emph{any}
information theoretic 2-out-of-2, or even any $k$-out-of-$n$ secret
sharing scheme, but we restrict ourselves to considering the following
simple group-based secret sharing scheme similar to the one-time pad.
Consider an alphabet $\mc{M}$ equipped with a binary operation
$\ast : \mc{M} \times \mc{M}\to\mc{M}$ such that $(\mc{M}, \ast)$ forms
a group. The secret sharing scheme we use maps any element
$m \in \mc{M}$ to the couple~$(p, p \ast m)$ for a uniformly random
$p \in \mc{M}$. Clearly, neither $p$ nor $p \ast m$, independently,
depends on $m$ when~$p$ is sampled uniformly at random. This holds as
$p \mapsto p \ast m$ is a permutation for all $m \in \mc{M}$. Hence,
neither pieces of information alone is sufficient to obtain any
information on $m$. But, knowledge of both $p$ and $p \ast m$ is
clearly sufficient to recover $m$.

We combine this secret sharing scheme with any two AQECM schemes
$S=(K,E,D)$ and $S' = (K', E', D')$ which share the same messages
$\mc{M}$ to obtain a new AQECM we will denote $S \ast S'$. On input of a
message $m$ and a pair of keys $(k, k') \in \mc{K} \times \mc{K}'$,
the encryption procedure of this new scheme will sample a uniformly
random $p \in \mc{M}$ and output $E_{k}(p)\tensor E'_{k'}(p\ast m)$.
The decryption procedure of $S \ast S'$ will proceed as expected: both
ciphertexts will be decrypted using the decryption maps $D$ and $D'$,
respectively, and the appropriate group action will be applied to
recover the message $m$. However, the decryption map of
$S \ast S'$ will reject if and only if \emph{both} ciphertexts were
rejected by $D$. Otherwise, \ie~if at least one of $D$ or
$D'$ accepts, the decryption map of $S \ast S'$ will accept.

Assume that $S$ and $S'$ exhibit sufficiently good correctness and
security as tamper-evident schemes. Correctness of $S \ast S'$ follows
naturally from the correctness of $S$ and $S'$. The security of
$S \ast S'$ as a tamper-evident scheme essentially follows from
the fact that if no tampering is detected on at least one of the two
ciphertexts produced by this scheme, then an eavesdropper is highly
unlikely to have learned information on both shares of the secret
sharing scheme and hence is unlikely to have learned any information on
the message.

We formalize this discussion in the following. In particular, we will
require that the group operation be computed coherently. This will be a
useful property later.

\begin{definition}
\label{te:df:gss}
\label{te:df:ast}
	Let $S = (K, E, D)$ and $S' = (K', E', D')$ be two AQECM schemes as
	given in \cref{df:aqecm} which share the same message space
	$\mc{M}$. Let this message space be equipped with a group operation
	denoted $\ast$ and let $U_\ast \in \mc{U}(\tsf{M} \tensor \tsf{M})$
	be the unitary operator which coherently computes this operation,
	which is to
	say that~$U(\ket{a}\tensor\ket{b}) = \ket{a}\tensor\ket{a \ast b}$
	for all $a,b \in \mc{M}$.\footnote{%
		Such a group operation can be found for any alphabet $\mc{M}$
		by considering any bijection $\phi:\mc{M}\to\Z/\abs{\mc{M}}\Z$,
		where $\Z/\abs{\mc{M}}\Z$ is the additive group of integers
		modulo $\abs{\mc{M}}$. Simply define
		$m \ast m' = \phi^{-1}(\phi(m) + \phi(m'))$.}
	We define the AQECM scheme $S \ast S' = (K^\ast, E^\ast, D^\ast)$ as
	follows:
	\begin{itemize}
		\item
			The key generation channel $
				K^\ast:
				\C
				\to \mc{L}(\tsf{K} \tensor \tsf{K}')
			$ is simply the parallel application of the key generation
			channels of $S$ and $S'$, \ie~$K^\ast = K \tensor K'$.
		\item
			The encryption channel $
				E^\ast :
				\mc{L}(\tsf{K} \tensor \tsf{K}' \tensor \tsf{M})
				\to
				\mc{L}(\tsf{C} \tensor \tsf{C}')
			$ is given by
			\begin{equation}
				\rho
				\mapsto
				\E_{p \gets \mc{M}}
				\left(E \tensor E'\right)
				\circ
				\Swap_{\tsf{K}, \tsf{K}', \tsf{M}, \tsf{M}}^{(1,3,2,4)}
				\left(
					\left(
						I_{\tsf{K} \tensor \tsf{K}'}
						\tensor
						U_\ast
					\right)
					(p \tensor \rho)
					\left(
						I_{\tsf{K} \tensor \tsf{K}'}
						\tensor
						U_\ast
					\right)^\dag
				\right).
			\end{equation}
			In other words, $E^\ast$ samples uniformly at random an
			element $p \in \mc{M}$, coherently computes the result of
			the group action between the supplied message and $p$ and
			encrypts the resulting state on the message spaces
			with $E \tensor E'$ after having properly routed the keys.
			For a state of the expected form
			$\rho = k \tensor k' \tensor m$, the
			result of this channel is
			\begin{equation}
				E^\ast_{(k,k')}(m)
				=
				\E_{p \gets \mc{M}}
				E_k(p) \tensor E'_{k'}(p \ast m).
			\end{equation}
		\item
			The decryption channel $
				D^\ast
				:
				\mc{L}\left(
					\tsf{K} \tensor \tsf{K}'
					\tensor
					\tsf{C} \tensor \tsf{C}'
				\right)
				\to
				\mc{L}(\tsf{M} \tensor \tsf{F})
			$
			is given by
			\begin{equation}
				\rho
				\mapsto
				W
				\circ
				\left(D \tensor D'\right)
				\circ
				\Swap_{\tsf{K}, \tsf{K}', \tsf{C}, \tsf{C}'}^{(1,3,2,4)}
				(\rho)
			\end{equation}
			where $
				W
				:
				\mc{L}(
					\tsf{M} \tensor \tsf{F}
					\tensor
					\tsf{M} \tensor \tsf{F}
				)
				\to
				\mc{L}(\tsf{M} \tensor \tsf{F})
			$ is a channel which (1) measures both flag qubits produced
			by $D$ and $D'$ and outputs the logical ``or'' of the
			result, (2) coherently undoes the group operation on the
			message spaces, and (3) discards the final unnecessary
			message space supposed to hold the $p$ element sampled by
			the encryption channel. Formally, this is given by
			\begin{equation}
				\rho
				\mapsto
				\left(\Tr_\tsf{M} \tensor \Id_{\tsf{M} \tensor
				\tsf{F}}\right)
				\left(
				\sum_{b,b' \in \bs{}}
				\left(U_\ast^\dag\tensor\ketbra{b \lor b'}{b,b'}\right)
				\Swap_{\tsf{M},\tsf{F},\tsf{M},\tsf{F}}^{(1,3,2,4)}
					(\rho)
				\left(U_\ast\tensor\ketbra{b,b'}{b \lor b'}\right)
				\right).
			\end{equation}
	\end{itemize}
\end{definition}

The following theorem states that if $S$ and $S'$ are
$\epsilon$-correct and $\epsilon'$-correct, then
$S \ast S'$ is $(\epsilon + \epsilon')$-correct. This follows
from the fact that the probability that the decryption procedure of
$S \ast S'$ both produces the correct message and accepts is at
least the probability that both instances of the decryption procedure
from $S$ accept and produce the correct messages. As the keys are
independently sampled, this is at least
$(1 - \epsilon)(1-\epsilon') \geq 1 - (\epsilon + \epsilon')$. Note that
this is identical to the correctness of $S \tensor S'$. In particular,
we are disregarding the possibility that only one decryption procedure
accepts, or that these decryption maps both output wrong messages which,
by coincidence, still lead to the correct message being recovered after
application of the group operation.

\begin{theorem}
\label{te:th:gss-c}
\label{te:th:ast-c}
	Let $S$ and $S'$ be $\epsilon$- and $\epsilon'$-correct AQECMs,
	respectively, with the same message space $\mc{M}$ which
	admits a group operation denoted $\ast$. Then, $S \ast S'$ is
	$(\epsilon+\epsilon')$-correct.
\end{theorem}
\begin{proof}
	For all keys $(k, k') \in \mc{K} \times \mc{K}'$ and all messages
	$m \in \mc{M}$ we have by definition that
	\begin{equation}
	\begin{aligned}
		\bra{m}
		\overline{D^\ast}_{(k,k')} \circ E^\ast_{(k,k')}(m)
		\ket{m}
		&=
		\E_{p \gets \mc{M}}
		\bra{m}
		\overline{D^\ast}_{(k,k')}
		\left(E_{k}(p) \tensor E_{k'}(p \ast m)\right)
		\ket{m}
		\\&\geq
		\E_{p \gets \mc{M}}
		\bra{p, p \ast m}
		\overline{D}_{k} \circ E_{k}(p)
		\tensor
		\overline{D'}_{k'} \circ E'_{k'}(p \ast m)
		\ket{p, p \ast m}
		\\&=
		\E_{p \gets \mc{M}}
		\bra{p}
		\overline{D}_{k} \circ E_{k}(p)
		\ket{p}
		\cdot
		\bra{p \ast m}
		\overline{D'}_{k'} \circ E'_{k'}(p \ast m)
		\ket{p \ast m}
	\end{aligned}
	\end{equation}
	Now, recalling that $K^\ast = K \tensor K'$ so that both elements of
	the key pair are generated independently, we have that
	\begin{equation}
	\begin{aligned}
		&
		\E_{(k, k') \gets K^\ast(\emptystring)}
		\bra{m}
		\overline{D^\ast}_{(k_0,k_1)} \circ E^\ast_{(k_0,k_1)}(m)
		\ket{m}
		\\&\geq
		\E_{p \gets \mc{M}}
		\left(
		\E_{k \gets K(\emptystring)}
		\bra{p}
		\overline{D}_{k} \circ E_{k}(p)
		\ket{p}
		\right)
		\left(
		\E_{k' \gets K'(\emptystring)}
		\bra{p \ast m}
		\overline{D'}_{k'} \circ E'_{k'}(p \ast m)
		\ket{p \ast m}
		\right)
		\\&\geq
		\E_{p \gets \mc{M}}
		(1-\epsilon)
		(1-\epsilon')
		\\&\geq
		1 - (\epsilon + \epsilon')
	\end{aligned}
	\end{equation}
	which is the desired result.
\end{proof}

We now move on to proving the security of $S \ast S'$. We begin with a
technical lemma providing us with a useful decomposition of the map
$\Tr_\tsf{M} \circ \overline{D^\ast}_{(k, k')}$ obtained by application
of the inclusion-exclusion principal. Note that this channel computes
the probability that at least one of the two underlying ciphertexts is
accepted. This is equal
to the probability that the first is accepted, plus the probability that
the second is accepted, minus the probability that both were accepted.

\begin{lemma}
	\label{te:th:gss-t}
	In the notation of \cref{te:df:gss}, we have for
	every $(k, k') \in \mc{K} \times \mc{K}'$ that
	\begin{equation}
		\Tr_\tsf{M} \circ \overline{D'}_{(k,k')}
		=
		\left(\Tr_\tsf{M} \circ \overline{D}_{k}\right)
		\tensor
		\Tr_\tsf{C}
		+
		\Tr_\tsf{C}
		\tensor
		\left(\Tr_\tsf{M} \circ \overline{D}_{k'}\right)
		-
			\left(\Tr_\tsf{M} \circ \overline{D}_{k}\right)
		\tensor
		\left(\Tr_\tsf{M} \circ \overline{D}_{k'}\right)
		.
	\end{equation}
\end{lemma}
\begin{proof}
	By definition, for every key pair $(k, k') \in \mc{K} \times
	\mc{K}'$ we have that $\overline{D^\ast}_{(k,k')}(\rho)$ is given by
	\begin{equation}
			\left(
				\Tr_\tsf{M} \tensor \Id_{\tsf{M}}
			\right)
			\left(
			\sum_{\substack{b,b' \in \bs{}\\(b,b') \not= (\bz,\bz)}}
			\left(U_\ast^\dag \tensor \bra{b,b'}\right)
			\Swap_{\tsf{M},\tsf{F},\tsf{M},\tsf{F}}^{(1,3,2,4)}
			\circ (D_k \tensor D'_{k'})(\rho)
			\left(U_\ast\tensor \ket{b,b'}\right)
			\right)
	\end{equation}
	It follows that $\Tr_\tsf{M} \circ \overline{D^\ast}_{(k,k')}(\rho)$
	is given by
	\begin{equation}
			\Tr_{\tsf{M} \tensor \tsf{M}}
			\left(
			\sum_{\substack{b,b' \in \bs{}\\(b,b') \not= (\bz,\bz)}}
			\left(U_\ast^\dag \tensor \bra{b,b'}\right)
			\Swap_{\tsf{M},\tsf{F},\tsf{M},\tsf{F}}^{(1,3,2,4)}
			\circ (D_k \tensor D'_{k'})(\rho)
			\left(U_\ast\tensor \ket{b,b'}\right)
			\right).
	\end{equation}
	Tracing out the message spaces after undoing the group operation via
	conjugation by $U_\ast^\dag$ is equivalent to tracing them out
	before this conjugation. We can even take this trace prior to the
	$\Swap$ operation. Hence, the previous expression is equivalent to
	\begin{equation}
			\sum_{\substack{b,b' \in \bs{}\\(b,b') \not= (\bz,\bz)}}
			\bra{b,b'}
			\left(
				\left(
					\left(\Tr_\tsf{M} \tensor \Id_\tsf{F}\right)
					\circ
					D_k
				\right)
				\tensor
				\left(
					\left(\Tr_\tsf{M} \tensor \Id_\tsf{F}\right)
					\circ
					D'_{k'}
				\right)
			\right)
			(\rho)
			\ket{b,b'}
	\end{equation}
	which is, in turn, equivalent to
	\begin{equation}
		\left(
			\left(
				\Tr_\tsf{M} \circ \overline{D}_{k}
			\right)
			\tensor
			\Tr_{\tsf{C}'}
		+
			\Tr_{\tsf{C}}
			\tensor
			\left(
				\Tr_\tsf{M} \circ \overline{D'}_{k'}
			\right)
		-
			\left(
				\Tr_\tsf{M} \circ \overline{D}_{k}
			\right)
			\tensor
			\left(
				\Tr_\tsf{M} \circ \overline{D'}_{k'}
			\right)
		\right)
		(\rho)
	\end{equation}
	where the first summand counts the cases of
	$(b_0,b_1) \in \{(\bo,\bz),(\bo,\bo)\}$, the second the cases of
	$(b_0,b_1) \in \{(\bz,\bo),(\bo,\bo)\}$, and the third accounts for
	the fact that the case of $(b_0, b_1) = (\bo,\bo)$ would otherwise
	be double counted.
	Note that the $\Tr_\tsf{C}$ term arises as a
	result of the fact that $\sum_{b \in \bs{}}
	\bra{b} \left(\Tr_\tsf{M} \tensor \Id_{\tsf{F}}\right) D_k(\rho)
	\ket{b} = \left(\Tr_\tsf{M} \tensor \Tr_{\tsf{F}}\right) \circ
	D_k(\rho) = \Tr_\tsf{C}(\rho)$ where the first equality
	is obtained by a characterization of the trace and the second by the
	fact that $D_k$ is trace preserving. The $\Tr_{\tsf{C}'}$ arises
	similarly.
\end{proof}

With the previous lemma in hand, we can now prove the security of
$S \ast S'$.

\begin{theorem}
\label{te:th:ast-s}
	Let $S$ and $S'$ be $\delta$- and $\delta'$-tamper-evident AQECMs,
	respectively, with the same message space $\mc{M}$ which admits
	a group operation denoted $\ast$. Then, the AQECM scheme $S \ast S'$
	is $2\sqrt{\delta + \delta'}$-tamper evident.
\end{theorem}
\begin{proof}
	Fix two messages $m, \tilde{m} \in \mc{M}$ and a tamper attack $
		A^\ast :
		\mc{L}(\tsf{C} \tensor \tsf{C}')
		\to
		\mc{L}(\tsf{C} \tensor \tsf{C}' \tensor \tsf{A}^\ast)
	$ against the AQECM
	scheme~$S \ast S' = (K^\ast, E^\ast, D^\ast)$.
	For any two elements $p, m \in \mc{M}$ and any key
	$(k,k') \in \mc{K} \times \mc{K}' = \mc{K}^\ast$, we define the
	states
	\begin{equation}
		\rho^{p,m}_{k,k'}
		=
		A^\ast \left(E_{k}(p) \tensor E'_{k'}(p \ast m)\right)
	\end{equation}
	and
	\begin{equation}
		\rho^m_{k,k'}
		=
		A^\ast \circ E^\ast_{(k,k')}(m).
	\end{equation}
	Note that $\rho^{p,m}_{k,k'}$ is the state after the adversary's
	attack --- but before the application of the decryption map --- when
	the message $m$ is encrypted with the key $(k,k')$ and the secret
	sharing scheme picked $p$ as the additional element.
	Further note that
	$\rho^m_{k,k'} = \E_{p \gets \mc{M}} \rho^{p,m}_{k,k'}$ is the
	same state, but averaged over all choices of $p$.

	Recall that \cref{tee:th:te-ns} essentially restates the property of
	being tamper evident from a concentration inequality on
	\begin{equation}
	\frac{1}{2}
		\norm{
			\left(
				\left(
					\Tr_\tsf{M} \circ \overline{D^\ast}_{(k,k')}
				\right)
				\tensor
				\Id_{\tsf{A}^\ast}
			\right)
			\left(
				\rho^{m}_{k,k'} - \rho^{\tilde{m}}_{k,k'}
			\right)
		}_1
	\end{equation}
	over the sampling of the key $(k,k')$ to an expectation of this
	trace distance over the same sampling. This is done via Markov's
	inequality (\cref{th:markov}).
	Our goal for the remainder of the
	proof will be to give an upper bound for the expected value of this
	quantity over sampling of the key $(k,k')$.
	We note that working with the expectations will simplify our
	accounting of the random sampling of the $p$ value by the encryption
	channel $E^\ast$.

	Using the result and notation of \cref{te:th:gss-t}, linearity, and
	the fact that norms are subadditive, for any key
	$(k, k') \in \mc{K}^\ast$ and any two messages
	$m, \tilde{m} \in \mc{M}$, we have that
	\begin{equation}
	\label{te:eq:gcc-s}
	\begin{aligned}
		&
		\E_{(k,k') \gets K^\ast(\varepsilon)}
		\frac{1}{2}
		\norm{
			\left(
				\left(\Tr_\tsf{M}\circ\overline{D^\ast}_{(k,k')}\right)
		\tensor
		\Id_{\tsf{A}^\ast}
		\right)
		\left(
			\rho^{m}_{k,k'} - \rho^{\tilde{m}}_{k,k'}
		\right)
		}_1
		\\&\leq
		\E_{(k,k') \gets K^\ast(\varepsilon)}
		\E_{p \gets \mc{M}}
		\frac{1}{2}
		\norm{
			\left(
				\left(\Tr_\tsf{M} \circ \overline{D^\ast}_{(k,k')}\right)
		\tensor
		\Id_{\tsf{A}^\ast}
		\right)
		\left(
			\rho^{p,m}_{k,k'} - \rho^{p,\tilde{m}}_{k,k'}
		\right)
		}_1
		\\&\leq\phantom{+}
		\E_{(k,k') \gets K^\ast(\varepsilon)}
		\E_{p \gets \mc{M}}
		\frac{1}{2}
		\norm{\left(
			\left(\Tr_\tsf{M} \circ \overline{D}_{k}\right)
			\tensor
			\Tr_{\tsf{C}'}
			\tensor
			\Id_\tsf{A}
		\right)
		(\rho^{p,m}_{k,k'} - \rho^{p,\tilde{m}}_{k,k'})
		}_1
		\\&\hspace{2pt}\phantom{=}+ % HACK
		\E_{(k,k') \gets K^\ast(\varepsilon)}
		\E_{p \gets \mc{M}}
		\frac{1}{2}
		\norm{\left(
			\Tr_\tsf{C}
			\tensor
			\left(\Tr_\tsf{M} \circ \overline{D'}_{k'}\right)
			\tensor
			\Id_\tsf{A}
		\right)
		(\rho^{p,m}_{k,k'} - \rho^{p,\tilde{m}}_{k,k'})
		}_1
		\\&\hspace{2pt}\phantom{=}+ % HACK
		\E_{(k,k') \gets K^\ast(\varepsilon)}
		\E_{p \gets \mc{M}}
		\frac{1}{2}
		\norm{\left(
			\left(\Tr_\tsf{M} \circ \overline{D}_{k}\right)
			\tensor
			\left(\Tr_\tsf{M} \circ \overline{D'}_{k'}\right)
			\tensor
			\Id_\tsf{A}
		\right)
		(\rho^{p,m}_{k,k'} - \rho^{p,\tilde{m}}_{k,k'})
		}_1.
	\end{aligned}
	\end{equation}
	Looking at the last of these inequalities, we will give an
	upperbound on the first two summands by using the fact that $S$ and
	$S'$ are $\delta$- and $\delta'$-tamper-evident, respectively, and
	we will upperbound the third summand by using the fact that
	$S \tensor S'$ is $(\delta + \delta')$-tamper-evident by
	\cref{th:te-parallel}.

	We begin with the third summand.
	Let $S'' = S \tensor S' = (K'', E'', D'')$  and note the following
	facts. First, the channel $A^\ast$ is also a tamper attack against
	$S''$, second we have that
	$
		(\Tr_\tsf{M} \circ \overline{D}_{k})
		\tensor
		(\Tr_\tsf{M} \circ \overline{D'}_{k'})
		=
		\Tr_{\tsf{M} \tensor \tsf{M}} \circ \overline{D''}_{(k,k')}
	$, and third by construction we have that $K^\ast = K''$.
	These, with the fact that $S''$ is $(\delta+\delta')$-tamper
	evident, imply via \cref{tee:th:te-ns} that
	\begin{equation}
		\E_{(k,k) \gets K^\ast(\varepsilon)}
		\frac{1}{2}
		\norm{\left(
			\left(\Tr_\tsf{M} \circ \overline{D}_{k}\right)
			\tensor
			\left(\Tr_\tsf{M} \circ \overline{D}_{k'}\right)
			\tensor
			\Id_{\tsf{A}^\ast}
		\right)
		(\rho^{p,m}_{k,k'} - \rho^{p,\tilde{m}}_{k,k'})
		}_1
		\leq
		2(\delta + \delta') - (\delta + \delta')^2
	\end{equation}
	for all $p \in \mc{M}$.
	As this inequality holds for all $p \in \mc{M}$, we can also bound
	the expectation over the sampling of $p$:
	\begin{equation}
		\E_{(k,k') \gets K^\ast(\varepsilon)}
		\E_{p \gets \mc{M}}
		\frac{1}{2}
		\norm{\left(
			\left(\Tr_\tsf{M} \circ \overline{D}_{k}\right)
			\tensor
			\left(\Tr_\tsf{M} \circ \overline{D'}_{k'}\right)
			\tensor
			\Id_{\tsf{A}^\ast}
		\right)
		(\rho^{p,m}_{k,k'} - \rho^{p,\tilde{m}}_{k,k'})
		}_1
		\leq
		2(\delta + \delta') - (\delta + \delta')^2.
	\end{equation}

	We now move on to the first two summands. We will only explicitly
	treat the second as the bound for first second is obtained
	completely analogously. For every $p \in \mc{M}$ and $k \in \mc{K}$,
	let $A'_{k,p}:\mc{L}(\tsf{C}') \to \mc{L}(\tsf{C}' \tensor \tsf{A}^\ast)$
	be defined by
	\begin{equation}
		\rho
		\mapsto
		\left(
			\Tr_{\tsf{C}}
			\tensor
			\Id_{\tsf{C}' \tensor \tsf{A}^\ast}
		\right)
		\circ
		A^{\ast}\left(E_{k}(p) \tensor \rho \right).
	\end{equation}
	Note that $A'_{k, p}$ is a tamper attack against the underlying
	AQECM scheme $S'$ which is assumed to be $\delta'$-tamper-evident.
	Thus, for every $m, \tilde{m} \in \mc{M}$ \cref{tee:th:te-ns}
	implies that
	\begin{equation}
		\E_{k' \gets K'(\varepsilon)}
			\frac{1}{2}\norm{
				\left(
					\left(\Tr_\tsf{M} \circ \overline{D'}_{k'}\right)
					\tensor
					\Id_{\tsf{A}^\ast}
				\right)
				\circ
				A'_{k,p}
				\circ
				E'_{k'}(
				(m \ast p)
				-
				(\tilde{m} \ast p)
				)
				}_1
		\leq
		2\delta' - \delta'^2.
	\end{equation}
	On the other hand, by definition of $A'_{k,p}$, we have that
	\begin{equation}
	\begin{aligned}
		&
		\left(
			\left(\Tr_\tsf{M} \circ \overline{D'}_{k'}\right)
			\tensor
			\Id_{\tsf{A}^\ast}
		\right)
		\circ
		A'_{k,p}
		\left(
			E'_{k'}((p \ast m) - (p \ast \tilde{m}))
		\right)
		\\&=
		\left(
			\left(\Tr_\tsf{M} \circ \overline{D'}_{k'}\right)
			\tensor
			\Id_{\tsf{A}^\ast}
		\right)
		\circ
		\left(\Tr_{\tsf{C}} \tensor \Id_{\tsf{C}'\tensor\tsf{A}^\ast}\right)
		\circ
		A^\ast\left(
			E_{k}(p) \tensor E'_{k'}(p \ast m)
			-
			E_{k}(p) \tensor E'_{k'}(p \ast \tilde{m})
		\right)
		\\&=
		\left(
			\Tr_\tsf{C}
			\tensor
			\left(\Tr_\tsf{M} \circ \overline{D'}_{k'}\right)
			\tensor
			\Id_{\tsf{A}^\ast}
		\right)
		\left(\rho^{m,p}_{k,k'} - \rho^{\tilde{m},p}_{k,k'}\right).
	\end{aligned}
	\end{equation}
	Thus,
	\begin{equation}
		\E_{k' \gets K'(\varepsilon)}
			\frac{1}{2}\norm{
				\left(
					\Tr_\tsf{C}
					\tensor
					\left(\Tr_\tsf{M} \circ \overline{D'}_{k'}\right)
					\tensor
					\Id_{\tsf{A}^\ast}
				\right)
				\left(
				\rho^{m,p}_{k,k'}
				-
				\rho^{\tilde{m},p}_{k,k'}
				\right)
				}_1
		\leq
		2\delta' - \delta'^2.
	\end{equation}
	As this holds for all $k$ and all $p$, we can also bound the
	expectation over the sampling of both these values:
	\begin{equation}
		\E_{(k,k') \gets K^{\ast}(\varepsilon)}
		\E_{p \gets \mc{M}}
			\frac{1}{2}\norm{
				\left(
					\Tr_\tsf{C}
					\tensor
					\left(\Tr_\tsf{M} \circ \overline{D}_{k}\right)
					\tensor
					\Id_{\tsf{A}^\ast}
				\right)
				\left(
				\rho^{m,p}_{k,k'}
				-
				\rho^{\tilde{m},p}_{k,k'}
				\right)
				}_1
		\leq
		2\delta' - \delta'^2.
	\end{equation}
	Noting that $\E_{p \gets \mc{M}}E_k(p) \tensor
	E'_{k'}(p \ast m) = \E_{p \gets \mc{M}} E_k(p \ast m^{-1}) \tensor
	E'_{k'}(p)$ where $m^{-1}$ is the inverse of $m$ for $\ast$, a
	similar exercise for the first summand will yield
	\begin{equation}
		\E_{(k,k') \gets K^\ast(\varepsilon)}
		\E_{p \gets \mc{M}}
			\frac{1}{2}\norm{
				\left(
					\left(\Tr_\tsf{M} \circ \overline{D}_{k}\right)
					\tensor
					\Tr_{\tsf{C}'}
					\tensor
					\Id_{\tsf{A}^\ast}
				\right)
				\left(
				\rho^{m,p}_{k,k'}
				-
				\rho^{\tilde{m},p}_{k,k'}
				\right)
				}_1
		\leq
		2\delta - {\delta}^2.
	\end{equation}

	We conclude that
	\begin{equation}
	\begin{aligned}
		\E_{(k,k') \gets K^\ast(\varepsilon)}
		\frac{1}{2}
		\norm{
		\left(
				\left(\Tr_\tsf{M} \circ \overline{D^\ast}_{(k,k')}\right)
				\tensor
				\Id_{\tsf{A}^\ast}
			\right)
			(\rho^m_{k,k'} - \rho^{\tilde{m}}_{k,k'})
			}_1
			&\leq
			4(\delta + \delta') - (\delta + \delta')^2 - \delta^2 -
			{\delta'}^2
			\\&\leq 4(\delta + \delta')
	\end{aligned}
	\end{equation}
	where, for simplicity, we drop the higher order terms.
	A final invocation of \cref{tee:th:te-ns} with this bound yields
	that $S^\ast$ is $2\sqrt{\delta + \delta'}$-tamper evident.
\end{proof}

In the following theorem, we precisely state that ciphertexts produced
by $S \ast S'$ can be split into two parts, both of which will be
accepted by the honest decryption procedure. We merely sketch the proof
as the idea is quite simple: leave one unmodified encoded share of the
secret sharing scheme in each part.

\begin{theorem}
	For every $\delta > 0$ and every alphabet $\mc{M}$, there exists a
	perfectly correct $\delta$-tamper-evident AQECM scheme with messages
	$\mc{M}$ as well as a channel $\tilde{A}$ satisfying
	\begin{equation}
		\E_{k \gets K(\varepsilon)}
		\Tr\left[
		\left(\overline{D}_k \tensor \overline{D}_k\right)
		\circ \tilde{A} \circ E_k(m)\right] = 1
	\end{equation}
	for all messages $m \in \mc{M}$.
\end{theorem}
\begin{proof}
	By \cref{te:th:te-exist}, there exists a perfectly correct
	$\frac{1}{2}\delta^2$-tamper-evident AQECM scheme $S'$ with messages
	$\mc{M}$. Consider now $S = S' \ast S'$ which is a perfectly correct
	$\delta$-tamper-evident AQECM with messages $\mc{M}$. It then
	suffices to take the channel $\tilde{A} : \mc{L}(\tsf{C}) \to
	\mc{L}(\tsf{C} \tensor \tsf{C})$ defined by
	\begin{equation}
		\rho
		\mapsto
		\Swap_{\tsf{C}, \tsf{C}', \tsf{C}', \tsf{C}'}^{(1,4,2,3)}
		\left(
			\rho
			\tensor
			\frac{I_{\tsf{C}'}}{\dim \tsf{C}'}
			\tensor
			\frac{I_{\tsf{C}'}}{\dim \tsf{C}'}
		\right)
	\end{equation}
	where $\tsf{C}'$ is the ciphertext space of the underlying AQECM
	scheme $S'$. Thus, for any messages $m$ and key $k = (k_0, k_1)$, we
	have that
	\begin{equation}
	\begin{aligned}
		\tilde{A} \circ E_k(m)
		&=
		\tilde{A}\left(
			\E_{p \gets \mc{M}}
			E'_{k_0}(p)
			\tensor
			E'_{k_1}(p \ast m)
		\right)
		\\&=
		\E_{p \gets \mc{M}}
		\left(E'_{k_0}(p)\tensor\frac{I_{\tsf{C}'}}{\dim\tsf{C}}\right)
		\tensor
		\left(
			\frac{I_{\tsf{C}'}}{\dim\tsf{C}'} \tensor E'_{k_1}(p\ast m)
		\right).
	\end{aligned}
	\end{equation}
	The fact that the underlying scheme $S'$ is perfectly correct
	implies that both instances of the decryption map in
	$\overline{D}_k$ will accept with certainty for any
	$k = (k_0, k_1)$, $m$ and $p$.
\end{proof}

%~~~~~~~~~~~~~~~~~~~~~~~~~~~~~~~~~~~~~~~~~~~~~~~~~~~~~~~~~~~~~~~~~~~~~~%
\subsection{One-Time Pad with a Tamper-Evident Encoded Key}
\label{te:sc:otp-te}
%~~~~~~~~~~~~~~~~~~~~~~~~~~~~~~~~~~~~~~~~~~~~~~~~~~~~~~~~~~~~~~~~~~~~~~%

The goal of this subsection is to formally study an AQECM scheme
obtained by combining a classical one-time pad with a tamper-evident
scheme $S = (K, E, D)$. The idea is to encode a message
$m$ with a random one-time pad $p$ and then encrypt $p$ with the
tamper-evident scheme and make the resulting state part of the
ciphertext. For a key $k$ and a message $m$ taken from an alphabet
$\mc{M}$ equipped with a group action $\ast$, the ciphertexts produced
by this scheme are of the form
\begin{equation}
	\E_{p \gets \mc{M}} E_k(p) \tensor (p \ast m).
\end{equation}
Decryption proceeds as expected. With the key $k$, recover $p$ from the
first component using the decryption map $D_k$ from the AQECM scheme
$S$, after which $m$ can be recovered from the second component
$p \ast m$. Accept if and only if $D_k$ accepted.

As we will see, such a scheme will be a useful example to
demonstrate that tamper evidence does \emph{not} imply
certain properties.

Note that this scheme appears quite similar to the ones constructed 
in \cref{te:sc:ast}. In fact, one could describe
this scheme as simply $S \ast T$ where $T$ is a trivial AQECM scheme on
the same message space $\mc{M}$. We formally give such a trivial
scheme in the next definition.

\begin{definition}
\label{te:df:triv-r-aqecm}
	Let $\mc{M}$ be an alphabet. The \emph{always rejecting trivial
	AQECM scheme for~$\mc{M}$}, denoted $\text{Triv}^{\mc{M},\bz}$, is
	an AQECM scheme with messages $\mc{M}$, a single key
	$\mc{K} = \{\varepsilon\}$, ciphertext space $\tsf{C}=\tsf{M}$, and
	defined by the following three channels:
	\begin{itemize}
		\item
			The key generation channel $K^{\mc{M}, 0}$ is
			the identity on $\C^{\{\varepsilon\}}$, which is isomorphic
			to $\C$.
		\item
			The encryption channel $E^{\mc{M}, 0}$ is the identity on
			$\mc{L}(\tsf{M}) = \mc{L}(\tsf{C})$.
		\item
			The decryption channel $
				D^{\mc{M}, 0} :
				\mc{L}(\tsf{C})
				\to
				\mc{L}(\tsf{M} \tensor \tsf{F})
			$ is defined by $\rho \mapsto \rho \tensor \ketbra{\bz}$.
	\end{itemize}
	We neglect the key space in the domains of
	$E^{\mc{M}, 0}$ and $D^{\mc{M}, 0}$ as $\C \tensor \tsf{M} = \tsf{M}
	= \tsf{C} = \C \tensor \tsf{C}$ and will omit the $\varepsilon$
	subscript for simplicity when discussing the keyed version of these
	channels.
\end{definition}

By construction, $S^* = S \ast \text{Triv}^{\mc{M}, \bz}$ is precisely
the scheme we have described and that we wish to study here.

Note that by always rejecting, the AQECM scheme
$\text{Triv}^{\mc{M},\bz}$ is $0$-tamper-evident. However, this comes at
the cost of not achieving any non-trivial level of correctness as
correctness requires accepting certain ciphertexts.

\begin{lemma}
\label{te:th:triv-r-aqecm}
	For any alphabet $\mc{M}$, the scheme $\text{Triv}^{\mc{M},\bz}$
	is $0$-tamper evident.
\end{lemma}
\begin{proof}
	As $D^{\mc{M},\bz}$ always rejects, the map
	$\overline{D^{\mc{M}, \bz}} : \mc{L}(\tsf{C}) \to \mc{L}(\tsf{M})$
	takes any linear operator to the zero operator
	$\mathbf{0}_\tsf{M} \in \mc{L}(\tsf{M})$. Moreover, for any other
	linear operator $L \in \mc{L}(\tsf{A})$, it holds that
	$\mathbf{0}_\tsf{M} \tensor L = \mathbf{0}_\tsf{M} \tensor
	\mathbf{0}_\tsf{A}$. Hence, for
	any tamper attack $A$ and any messages~$m$ and $\tilde{m}$,
	\begin{equation}
		\frac{1}{2}
		\norm{
			\left(
				\left(
					\Tr_\tsf{M} \circ \overline{D^{\mc{M},0}}
				\right)
				\tensor
				\Id_\tsf{A}
			\right)
			\circ A \circ E(m - \tilde{m})
			}_1
		=
		\frac{1}{2}
		\norm{
			\mathbf{0}_{\tsf{A}} - \mathbf{0}_{\tsf{A}}
		}_1
		=
		0
	\end{equation}
	which yields the desired result.
\end{proof}

By \cref{te:th:ast-s} and the previous lemma, if $S$ is
$\delta$-tamper evident, then $S^\ast$ is
$2\sqrt{\delta}$-tamper evident.\footnote{%
	We conjecture that $S^\ast$ is actually $\delta$-tamper evident.}
However, since $\text{Triv}^{\mc{M}, \bz}$ is $1$-correct, we cannot use
\cref{te:th:ast-c} to show that the scheme $S^\ast$ exhibits non-trivial
correctness. Instead, we will directly prove that $S^\ast$ is as correct
as $S$. This is the content of the following theorem.

\begin{theorem}
\label{te:th:ast-otp}
	Let $S = (K, E, D)$ be an $\epsilon$-correct $\delta$-tamper-evident
	scheme with messages~$\mc{M}$ which admit a
	group operation denoted $\ast$. Then,
	$S \ast \text{Triv}^{\mc{M},0} = (K^\ast, E^\ast, D^\ast)$
	is~$2\sqrt{\delta}$-tamper-evident and $\epsilon$-correct.
\end{theorem}
\begin{proof}
	That $S \ast \text{Triv}^{\mc{M}, 0}$ is
	$2\sqrt{\delta}$-tamper evident follows immediately from
	\cref{te:th:triv-r-aqecm} and \cref{te:th:ast-s}. The generic
	correctness result of \cref{te:th:ast-c} is not sufficient in this
	case as it only yields $(1 + \epsilon)$-correctness. However, direct
	computation yields
	\begin{equation}
	\begin{aligned}
		\E_{k \gets K^\ast(\varepsilon)}
		\bra{m} \overline{D^\ast_k} \circ E^\ast_k(m)\ket{m}
		&=
		\E_{k \gets K(\varepsilon)}
		\E_{p \gets \mc{M}}
		\bra{m} \overline{D^\ast_k}(E_k(p) \tensor (p \ast m)) \ket{m}
		\\&=
		\E_{k \gets K(\varepsilon)}
		\E_{p \gets \mc{M}}
		\bra{p} \overline{D_k} \circ E_k(p) \ket{p}
		\\&\geq
		\E_{p \gets \mc{M}} (1 - \epsilon)
		\\&\geq
		1 - \epsilon
	\end{aligned}
	\end{equation}
	where the first equality follows by definition of $K^\ast$,
	$E^\ast$, and $\text{Triv}^{\mc{M}, 0}$, while the second equality
	from the observation that $D^\ast_k$ will accept and produce the
	correct plaintext if and only if $D$ accepts the first ciphertext
	and outputs $p$.
\end{proof}

For the purposes of our next two examples, we make this construction
even more precise.
Consider a perfectly correct $\delta$-tamper-evident AQECM scheme $S$
whose messages are a single bit, which is to say that $\mc{M} = \bs{}$. Let
$S^\xor = S \xor \text{Triv}^{\mc{M}, 0} = (K^\xor, E^\xor, D^\xor)$ be
the scheme obtained from the construction of \cref{te:df:ast} where the
group operation is taken to be the exclusive-or $\xor$. By
\cref{te:th:ast-otp}, $S^\xor$ is a perfectly correct
$2\sqrt{\delta}$-tamper-evident AQECM scheme.

We now use $S^\xor$ to show seperations between
tamper evidence and two other notions: authentication and the encryption
of quantum states. Note
that by \cref{te:th:te-exist}, we can instantiate $S^\xor$ with a perfectly
correct $\delta'$-tamper-evident scheme $S$ for any $\delta' > 0$.
Hence, we can assume that $S^\xor$ is perfectly correct and
$\delta$-tamper evident for any $\delta > 0$.

%  -  -  -  -  -  -  -  -  -  -  -  -  -  -  -  -  -  -  -  -  -  -  - %
\subsubsection{Tamper Evidence Does Not Imply Authentication or Robustness}
\label{te:sc:te-not-auth}
%  -  -  -  -  -  -  -  -  -  -  -  -  -  -  -  -  -  -  -  -  -  -  - %

The notion of robustness for AQECM schemes was originally stated in
Broadbent and Islam's work on certified deletion \cite{BI20}, but it can
be understood as a restriction of Barnum \etal's notion of quantum
authentication \cite{BCG+02} from any arbitrary quantum states to only
classical messages. In a nutshell, an AQECM scheme is robust if it is
infeasible for an adversary to change a ciphertext in such a way that
it will still be accepted by the decryption map, but that a different
message will be recovered. We do not state a formal definition of this
property.

As Gottesman noted, every quantum authentication scheme is itself a
tamper-evident AQECM scheme \cite{Got03}. Trivially, these are also
robust AQECM schemes. It is also clear that there exist robust AQECMs
which are not tamper evident. For example, simply consider any classical
message authentication code which is information-theoretically secure
(\eg:~\cite[Sec.~4.6]{KL15}). As the encodings produced by these message
authentication codes are classical, it is infeasible to detect an
adversary making and keeping a copy of them.

The above naturally leads to the following question: Does
tamper evidence imply robustness? In the following, we sketch a proof
that the answer is ``no''. There are tamper-evident schemes which are
not robust.

Consider the AQECM scheme $S^\xor$ as previously described in this
section. Further consider the channel
$\tilde{A} : \mc{L}(\tsf{C} \tensor \C^{\bs{}}) \to \mc{L}(\tsf{C} \tensor \C^{\bs{}})$ defined by
\begin{equation}
	\rho
	\mapsto
	\left(I_\tsf{C} \tensor X\right)
		\rho
	\left(I_\tsf{C} \tensor X\right)
\end{equation}
where $X = \ketbra{\bz}{\bo} + \ketbra{\bo}{\bz}$ is the bit-flip
unitary operator. For any plaintext $b \in \bs{}$, we see that
\begin{equation}
\begin{aligned}
	\tilde{A} \circ E^\xor_{k}(m)
	&=
	\tilde{A} \left(\E_{p \gets \bs{}} E_k(p) \tensor (p \xor b)\right)
	\\&=
	\E_{p \gets \bs{}} E_k(p) \tensor (p \xor b \xor \bo)
	\\&=
	E^\xor_k(b \xor \bo).
\end{aligned}
\end{equation}
As we assume that $S$ is perfectly correct, $S^\xor$ is also perfectly
correct. Hence, for any key $k$ and message $b \in \bs{}$, we have
\begin{equation}
	\bra{b \xor \bo}
		\overline{D^\xor}_k \circ \tilde{A} \circ E^\xor_k(b)
	\ket{b \xor \bo}
	=
	1
\end{equation}
which is to say that an adversary using $\tilde{A}$ can change the
encoded plaintext undetected and with certainty.
It follows $S^\xor$ cannot offer any non-trivial level of robustness.

%  -  -  -  -  -  -  -  -  -  -  -  -  -  -  -  -  -  -  -  -  -  -  - %
\subsubsection{Tamper Evidence Does Not Imply the Encryption of Quantum
States}
\label{te:sc:te-not-qenc}
%  -  -  -  -  -  -  -  -  -  -  -  -  -  -  -  -  -  -  -  -  -  -  - %

While AQECM schemes generally only guarantee correctness and various
notions of security for \emph{classical} messages, their encryption and
decryption maps remain quantum channels. Hence, we can still inquire
about their behaviours with respect to arbitrary quantum states.

We claim without further proof that $S^\xor$ is also perfectly correct
with respect to any \emph{quantum state}, in the sense that
$D_k^\xor \circ E^\xor_k(\rho) = \rho \tensor \ketbra{\bo}$ for all keys
$k$ and quantum states $\rho \in \C^{\bs{}}$. This follows
directly from the details of the construction of $S^\xor$ from
\cref{te:df:ast}, in particular from the fact that the group operation
is computed coherently.

On the other hand, recall that $S^\xor$ is
$7\sqrt[4]{\delta}$-encrypting for classical messages due to
\cref{th:te=>enc} and the fact that it is
$2\sqrt{\delta}$-tamper evident. However, we claim that it offers no
non-trivial level of security as an encryption scheme for arbitrary
quantum states. This results from the fact that the
classical one-time-pad fails to change certain quantum states.

More precisely, going through the details of the construction in
\cref{te:df:ast}, we see that the unitary $U_\xor$ which coherently
computes the group operation is, here, simply the controlled-not
linear operator $\text{CNOT}$ which maps $\ket{b,b'}$ to
$\ket{b, b\xor b'}$ for any $b,b' \in \bs{}$.
Thus, for any key $k$ and state $\rho$, we have that
\begin{equation}
	E^\xor_k(\rho)
	=
	\E_{p \gets \bs{}}
	\left(E_k \tensor \Id_{\C^{\bs{}}}\right)
	\text{CNOT}
	\left(p \tensor \rho\right)
	\text{CNOT}^\dag
	=
	\E_{p \gets \bs{}}
	E_k(p) \tensor X^p \rho X^p
\end{equation}
where $X = \ketbra{\bz}{\bo} + \ketbra{\bo}{\bz}$ is once again the
bit-flip unitary operator.

Consider now the states $
	\ket{+}
	=
	\frac{1}{\sqrt{2}}\left(\ket{\bz} + \ket{\bo}\right)
$ and $
	\ket{-}
	=
	\frac{1}{\sqrt{2}}\left(\ket{\bz} - \ket{\bo}\right)
$. Note that both resulting density operators $\ketbra{+}$ and
$\ketbra{-}$ are invariant under conjugation by $X$. Hence, for any two
keys $k$ and $k'$, we have that
\begin{equation}
\begin{aligned}
	\frac{1}{2}
	\norm{
		E^\xor_k(\ketbra{+})
		-
		E^\xor_{k'}(\ketbra{-})
	}_1
	&=
	\frac{1}{2}
	\norm{
		\E_{p \gets \bs{}}E_k(p) \tensor \ketbra{+}
		-
		\E_{p \gets \bs{}}E_{k'}(p) \tensor \ketbra{-}
	}_1
	\\&\geq
	\frac{1}{2}
	\norm{
		\ketbra{+}
		-
		\ketbra{-}
	}_1
	\\&=
	1
\end{aligned}
\end{equation}
where the inequality follows by tracing out the ciphertext from $S$ and
the last equality from \cref{th:td-pure} and $\braket{+}{-} = 0$. In
other words, the encodings of $\ket{+}$ and $\ket{-}$ under $E^\xor_k$
for any keys remain perfectly distinguishable. Hence, this AQECM scheme
does not offer any non-trivial level of security as an encryption scheme
for arbitrary quantum states.

\newcommand{\etalchar}[1]{$^{#1}$}

\end{document}